\documentclass[12pt]{article}
\usepackage[letterpaper,margin=1in]{geometry}
\usepackage{graphicx} 
\usepackage{amsmath,amsfonts,bbm}
\usepackage{mathrsfs,bm}
\usepackage{booktabs,multirow} 
\usepackage[title]{appendix} 
\usepackage{longtable} 
\usepackage{todonotes}
\usepackage[normalem]{ulem} 

\DeclareMathOperator{\Cov}{Cov}
\DeclareMathOperator{\Lip}{Lip}
\DeclareMathOperator{\exponential}{exp}
\DeclareMathOperator{\stably}{-stably}

\usepackage{setspace}

\usepackage{scalerel,stackengine}
\stackMath
\newcommand\reallywidehat[1]{%
\savestack{\tmpbox}{\stretchto{%
  \scaleto{%
    \scalerel*[\widthof{\ensuremath{#1}}]{\kern-.6pt\bigwedge\kern-.6pt}%
    {\rule[-\textheight/2]{1ex}{\textheight}}
  }{\textheight}%
}{0.5ex}}%
\stackon[1pt]{#1}{\tmpbox}%
}

\newcommand{\bfD}{\mathbf{D}}
\newcommand{\bfP}{\mathbf{P}}

\newcommand{\bfV}{\mathbf{V}}
\newcommand{\bfW}{\mathbf{W}}

\newcommand{\bfX}{\mathbf{X}}
\newcommand{\bfY}{\mathbf{Y}}
\newcommand{\bfZ}{\mathbf{Z}}

\newcommand{\bfd}{\mathbf{d}}

\newcommand{\bfl}{\mathbf{l}}
\newcommand{\bfp}{\mathbf{p}}

\newcommand{\bfv}{\mathbf{v}}
\newcommand{\bfx}{\mathbf{x}}
\newcommand{\bfz}{\mathbf{z}}

\newcommand{\bdd}{\partial\hspace{-.5mm}}

\usepackage{natbib}
\usepackage{xcolor}

\usepackage{amsthm}

\newtheorem{thm}{Theorem}[section]
\newtheorem{defi}{Definition}[section]
\newtheorem{rmk}{Remark}[section]
\newtheorem{prop}{Proposition}[section]
\newtheorem{asm}{Assumption}[section]
\newtheorem{lemma}{Lemma}[section]
\newtheorem{cor}{Corollary}[section]
\newtheorem{example}{Example}[section]

\counterwithin{equation}{section}

\allowdisplaybreaks

\title{Heterogeneous Effects of Endogenous Treatments with\\ Interference and Spillovers in a Large Network\thanks{We benefitted from useful comments from the participants at the New York Camp Econometrics XIX, the 40th Canadian Econometrics Study Group Meeting, and the Southern Economic Association 95th Annual Meeting. All the remaining errors are ours. Sasaki gratefully acknowledges the generous financial support of Brian and Charlotte Grove.}}
\author{Lin Chen%
	\thanks{Department of Economics, Vanderbilt University, Email: lin.chen@vanderbilt.edu} \and 
    Yuya Sasaki%
    \thanks{Brian and Charlotte Grove Chair and Professor of Economics, Department of Economics, Vanderbilt University. Email: yuya.sasaki@vanderbilt.edu}
    }
\date{}
\begin{document}

\maketitle

\addcontentsline{toc}{section}{Abstract}
{
\centerline{\textbf{Abstract}}\setlength{\baselineskip}{6.8mm}
This paper studies the identification and estimation of heterogeneous effects of an endogenous treatment under interference and spillovers in a large single-network setting. We model endogenous treatment selection as an equilibrium outcome that explicitly accounts for spillovers and derive conditions guaranteeing the existence and uniqueness of this equilibrium. We then identify heterogeneous marginal exposure effects (MEEs), which may vary with both the treatment status of neighboring nodes and unobserved heterogeneity. We develop estimation strategies and establish their large-sample properties. Equipped with these tools, we analyze the heterogeneous effects of import competition on U.S.\ local labor markets in the presence of interference and spillovers.
We find negative MEEs, consistent with the existing literature. However, these effects are amplified by spillovers in the presence of treated neighbors and among localities that tend to select into lower levels of import competition. These additional empirical findings are novel and would not be credibly obtainable without the econometric framework proposed in this paper.


}

\newpage

\section{Introduction}
The recent literature has seen a growing number of studies on causal effects in large single-network settings \citep[e.g.,][]{Cai3:15,Hu3:22,Li2:22,Leung:22ECTA,leung2022graph,Hoshino2:23_JSAS,Vazquez-Bare:23,Gao2:25}. These studies examine how to learn the effects of exogenous and endogenous treatments in the presence of interference and spillovers. By contrast, much less is known about how to identify and estimate heterogeneous treatment effects in this environment, disentangling the heterogeneous effects of an individual’s own treatment from those of neighbors’ treatments.

Causal inference in this sophisticated environment requires (i) a model that characterizes individuals’ treatment decisions in equilibrium and (ii) a model that describes how treatments affect outcomes in the presence of spillovers. Both treatment decisions and outcomes generally depend on the behavior of others. For example, when assessing the effects of import competition on local labor markets treated as nodes, the labor market response in one region (node) may influence neighboring regions (nodes) through labor commuting/migration. Moreover, local labor market composition, which determines the extent of exposure to trade shocks, may itself depend on conditions in neighboring regions. Hence, a credible model must accommodate interference in both treatment assignment and outcome determination.

We begin by modeling joint treatment decisions under incomplete information and derive conditions for the uniqueness of the rational expectations equilibrium, building on the literature on discrete games on networks \citep{Lee3:14,Xu:18,Lin2:24}. Uniqueness of equilibrium enables point identification of the treatment decision parameters and, more importantly, the propensity score. Leveraging this property, we then identify the treatment effects by adapting identification strategies from the causal inference literature. Following the network causal inference literature \citep{Leung:22ECTA,Hoshino2:23_JSAS}, we take advantage of an exposure mapping, a low-dimensional statistic that summarizes interference effects. Under this assumption, we show that heterogeneous causal effects are identified without imposing additional functional form restrictions on the outcome model.

We estimate the model using the generalized method of moments (GMM). To study the asymptotic properties of GMM under network dependence, we adopt the $\psi$-dependence framework developed by \citet[][hereafter KMS]{Kojevnikov3:21}, which establishes pointwise laws of large numbers, central limit theorems, and variance estimators that are robust to network dependence. Due to the nonlinear nature of our model, however, stronger conditions are required for valid statistical inference. In particular, we rely on the uniform law of large numbers \citep[ULLN,][]{Sasaki:25} in the $\psi$-dependence framework of KMS.
Building on these, we show that the GMM estimator is consistent and converges stably to a mixed normal distribution. Moreover, we demonstrate that standard inference procedures remain valid even when the asymptotic variance-covariance matrix is random. Related results can be found in \citet{Andrews:05}, \citet{Kuersteiner2:13}, and \citet{Kuersteiner2:20ECMA}. Because our asymptotic analysis is conducted conditionally on the common shocks that generate the network, the resulting estimator enjoys desirable asymptotic properties regardless of the underlying network formation model.

We illustrate our method using the \cite{Autor3:13} dataset on Chinese import competition and manufacturing employment shares across U.S. commuting zones (CZs) as nodes. Treating each CZ as an independent open economy, their two-stage least squares (2SLS) analysis documents a large negative effect of import competition on manufacturing employment. However, because CZs are geographically linked, interference may arise both in the extent of exposure to trade shocks and in labor market responses.
We model sectoral composition choices as outcomes of a rational expectations equilibrium on a spatial network. The results indicate strategic substitution in sectoral choices, with regions tilting away from sectors facing greater competition from Chinese imports. We then examine the direct MEEs of import competition on local labor markets. We find substantial heterogeneity across CZs depending on network exposure: the exposure effects differ markedly between CZs with and without treated neighbors, as well as between CZs experiencing above- versus below-median import competition. These findings highlight the importance of accounting for both network-induced interference and heterogeneity when evaluating treatment effects.

This paper contributes to the literature on network causal inference that employs the $\psi$-dependence framework for statistical inference \citep[e.g.,][]{Leung:22ECTA,Hoshino2:23_JSAS,Gao2:25}. Our approach differs from these studies in two respects. First, we focus on marginal responses that allow exposure effects to vary with unobserved heterogeneity. Second, we adopt a conditional $\psi$-dependence framework by conditioning on the common shocks that jointly generate the covariates and the network structure. As a result, the limiting distribution of our GMM estimator is mixed normal, with a random variance-covariance matrix that reflects the randomness induced by these common shocks.

This paper is also related to the literature on endogenous treatment with interference confined within groups or markets: \citet{Hoshino2:23JoE} study how the delinquency of an opposite-gender best friend affects a student’s GPA; \citet{Balat2:23} examine how competition among airlines influences air pollution across U.S.\ cities. In these settings, interference is localized, which allows the authors to identify outcomes under all possible treatment configurations and to conduct the corresponding comparisons. In contrast, we consider a setting in which all individuals are connected through a large single network.

The remainder of the paper is organized as follows. Section~\ref{sec:setup} presents the model setup. Section~\ref{sec:identification} establishes identification. Section~\ref{sec:estimation} proposes the estimators, and Section~\ref{sec:theory} analyzes their large-sample properties. Section~\ref{sec:simulations} evaluates the finite-sample performance of the proposed methods. Section~\ref{sec:application} presents an empirical application, and Section~\ref{sec:conclusion} concludes. All mathematical details are relegated to the appendix.

\section{Setup}\label{sec:setup}
Our observed data consist of $n$ observations $\{(Y_i,X_i,D_i,Z_i)\}_{i=1}^n$, where $Y_i$ denotes the outcome, $X_i$ the vector of covariates, $D_i$ the binary treatment status, and $Z_i$ the vector of exogenous variables. We allow $Z_i$ to include $X_i$ as a subvector, along with additional excluded exogenous variables that serve as instruments. In addition to these variables, we observe the network structure linking the $n$ individuals as nodes, represented by an adjacency matrix $A_n$, where $A_{ij}=1$ indicates that nodes $i$ and $j$ are connected and $A_{ij}=0$ otherwise. Wherever there is little risk of confusion, we suppress the subscript $n$ and simply write the adjacency matrix as $A$. We rule out self-spillover effects by assuming $A_{ii}=0$ for all $i$.

\subsection{Outcome Production}
In a network setting, it is natural to allow an individual’s potential outcomes to depend on the treatment status of other individuals. Let $\bfD=(D_i)_{i=1}^n$ denote the collection of treatment assignments, and let $Y_i(\bfD)$ be the potential outcome for node $i$ under the joint treatment status $\bfD$. Given the treatment status $D_i$ of node $i$ and the collection $\bfD_{-i}=(D_1,\ldots,D_{i-1},D_{i+1},\ldots,D_n)$ of treatment assignments excluding $D_i$, we model the outcome-generating process as
\begin{equation}\label{eq:potentialoutcome}
\begin{aligned}
    Y_i 
    &= \sum_{\bfd \in \{0,1\}^n} \mathbbm{1}\{\bfD=\bfd\}\, Y_i(\bfd), 
    \qquad \text{where} \\
    Y_i(\bfD) 
    &= \mu_Y\!\left(X_i, D_i, \bfD_{-i}, \epsilon_i^{(D_i,\bfD_{-i})}\right),
\end{aligned}
\end{equation}
and $\epsilon_i^{(D_i,\bfD_{-i})}$ denotes the productivity shock for node $i$ under the treatment configuration $(D_i,\bfD_{-i})$. This shock is unobserved by the other nodes, as well as by the econometrician.

\subsection{Treatment Choice}\label{sec:treatment_choice}
The $n$ nodes make simultaneous treatment decisions in an environment of incomplete information. The type of node $i$ is defined by the pair $(Z_i,\nu_i)$, where $Z_i$ is publicly observed by all nodes, while $\nu_i$ is private information known only to node $i$. The network structure $A$ is also publicly observed. Let $n_i=\sum_{j} A_{ij}$ denote the number of neighbors of node $i$.

Given the treatment decisions of other nodes $\bfD_{-i}$, node $i$ receives the payoff
\begin{equation}\label{equ:payoff}
    u_i(Z_i,A,\bfD_{-i},\nu_i)
    = Z_i'\beta_D 
    + \lambda \cdot n_i^{-1} \sum_{j} A_{ij} D_j 
    - \nu_i,
\end{equation}
if it chooses to take the treatment ($D_i=1$), and zero otherwise ($D_i=0$). Here, $\beta_D$ denotes the vector of coefficients associated with $Z_i$, and $\lambda$ captures the spillover effect from treated neighbors. Let $\bm{\beta}_1 = (\beta_D',\lambda)'$ for short hand.

Because node $i$ does not observe other nodes’ private information $\nu_{-i}$, it must base its decision on beliefs about others’ treatment choices. Accordingly, the best response of node $i$ is given by
\begin{equation}
    D_i^*
    = \mathbbm{1}\!\left\{
        Z_i'\beta_D 
        + \lambda \cdot n_i^{-1} \sum_{j} A_{ij} \Pr(D_j^* = 1 \mid \bfZ, A) 
        - \nu_i 
        \ge 0
    \right\}.
\end{equation}

We introduce the shorthand notation $P_i(\bfZ,A) = \Pr(D_i^* = 1 \mid \bfZ, A)$ for the propensity score. The rational-expectations (Bayesian Nash) equilibrium choice probabilities are given by the vector $\bfP(\bfZ,A) = \big(P_1(\bfZ,A), \ldots, P_n(\bfZ,A)\big)$, which solves the following system of equations:
\begin{equation}\label{eq:EQUcp}
\begin{aligned}
    P_1(\bfZ,A)
    &= F_{\nu \mid \bfZ}\!\left(
        Z_1'\beta_D 
        + \lambda \cdot n_1^{-1} \sum_{j} A_{1j} P_j(\bfZ,A)
    \right),\\
    &\vdots\\
    P_n(\bfZ,A)
    &= F_{\nu \mid \bfZ}\!\left(
        Z_n'\beta_D 
        + \lambda \cdot n_n^{-1} \sum_{j} A_{nj} P_j(\bfZ,A)
    \right).
\end{aligned}
\end{equation}

The existence of a solution to \eqref{eq:EQUcp} follows from Brouwer’s fixed-point theorem. Moreover, uniqueness of the equilibrium can be established via a contraction mapping argument. To formally establish these properties, we impose the following assumption.

\begin{asm}[Existence and Uniqueness of the Equilibrium]\label{asm:unique_equ} 
\begin{itemize}
    \item[]
    \item[(i)] Conditional on $\bfX$, $\{\nu_i\}_{i=1}^n$ are $i.i.d$ random variables drawn from the sandard logistic distribution whose distribution function is given by $F(t) = 1/( 1+ exp(-t))$.
    \item[(ii)] $\bfZ$ is independent of $\{\nu_i\}_{i=1}^n$ given $\bfX$.
    \item[(iii)] Moderate peer effects: $|\lambda| < 4$.
\end{itemize}
\end{asm}

Assumption~\ref{asm:unique_equ}(i) is commonly imposed in the literature on games with incomplete information \citep[e.g.,][to list but a few]{Brock2:02,Bajari4:10,Xu:18,Lin2:24}. Assumption~\ref{asm:unique_equ}(ii) extends the standard instrument exogeneity condition to a network setting. Under this assumption, we can write $F_{\nu \mid \bfX}$ in place of $F_{\nu \mid \bfZ}$. 
Regarding Assumption~\ref{asm:unique_equ}(iii), with an abuse of notation, we rewrite the system of equations in \eqref{eq:EQUcp} as $\bfP = \varphi(\bfP)$. \citet{Lee3:14} show that when the selection peer effect $\lambda$ is moderate, the norm of the Jacobian matrix of $\varphi$ is strictly less than one. This contraction property guarantees the existence and uniqueness of the equilibrium. We formally summarize this discussion in the following lemma.

\begin{lemma}[Existence and Uniqueness of the Equilibrium]\label{lemma:UniBNE}
    If Assumption \ref{asm:unique_equ} holds, then there exists a unique solution to the system \eqref{eq:EQUcp}.
\end{lemma}

\noindent
See Appendix \ref{sec:proof:identifiadtion} for proof.

\section{Identification}\label{sec:identification}

Thanks to the uniqueness property established in Lemma~\ref{lemma:UniBNE}, we can directly identify the equilibrium choice probability as
$
P_i(\bfZ,A) = \Pr(D_i = 1 \mid \bfZ, A).
$
We now proceed to identify the model parameters. We refer to the treatment choice equation as the \emph{first stage}, and the outcome-generating process as the \emph{second stage}.

\subsection{Identification of the First-Stage Parameters}\label{sec:identification:first}
Given the uniqueness property already established, the first stage reduces to a standard logistic regression in the network framework, similar to that studied in the literature \citep[e.g.,][]{Xu:18}. In light of Assumption~\ref{asm:unique_equ}(i), the equilibrium propensity score can be written as
\begin{equation}\label{equ:EQCp}
    P_i(\bfZ,A)
    = \frac{
        \exp\!\left(
            Z_i'\beta_D 
            + \lambda \cdot n_i^{-1} \sum_{j} A_{ij} P_j(\bfZ,A)
        \right)
    }{
        1 + \exp\!\left(
            Z_i'\beta_D 
            + \lambda \cdot n_i^{-1} \sum_{j} A_{ij} P_j(\bfZ,A)
        \right)
    }.
\end{equation}
Applying the logistic inversion yields
\begin{equation}\label{equ:Inversion}
    \log\!\big(P_i(\bfZ,A)\big)
    - \log\!\big(1 - P_i(\bfZ,A)\big)
    = Z_i'\beta_D 
    + \lambda \cdot n_i^{-1} \sum_{j} A_{ij} P_j(\bfZ,A).
\end{equation}
With the shorthand notation 
$
\phi_i(\bfZ,A)
= \left(
    Z_i',
    \; n_i^{-1} \sum_{j} A_{ij} P_j(\bfZ,A)
\right)',
$
consider the following matrix invertibility assumption.

\begin{asm}[First-Stage Identification]\label{asm:rank_cond}
$E[\phi_i(\bfZ,A)\phi_i(\bfZ,A)'|A]$ is invertible.
\end{asm}
This assumption is the standard rank condition and corresponds to Assumption~4 of \citet{Xu:18} and Assumption~7 of \citet{Lin2:24}, among others. First, the uniqueness of equilibrium implies that the quantity $P_i(\bfZ,A)$ is identified. Since $\phi_i(\bfZ,A)$ depends only on the first-stage controls $Z_i$ and the propensity scores $\{P_j(\bfZ,A)\}_{j=1}^n$, this object is also identified. Moreover, the right-hand side of \eqref{equ:Inversion} is linear in the parameter vector $\bm{\beta}_1=(\beta_D',\lambda)'$. Consequently, the invertibility of the matrix implies identification of the first-stage structural parameters, as formally stated in the following theorem.

\begin{thm}\label{thm:1stage_iden}
    If Assumptions \ref{asm:unique_equ} and \ref{asm:rank_cond} hold, then the first-stage parameter vector $\bm{\beta}_1=(\beta_D',\lambda)'$ is identified.
\end{thm}

\noindent
See Appendix \ref{sec:proof:identifiadtion} for proof.

Let $V_i = F_{\nu \mid \bfX}(\nu_i)$. We note three properties of this variable. First, under Assumption~\ref{asm:unique_equ}, the random variable $V_i$ follows a uniform distribution on $[0,1]$ and is independent of $V_{-i}$. Second, since $V_i \mid \bfX \sim \mathrm{Uniform}(0,1)$ does not depend on $\bfX$, it follows that the vector $\bfV$ is independent of $\bfX$. Finally, under this notation, $D_i = 1$ if and only if $P_i(\bfZ,A) > V_i$.

\begin{rmk}[On Monotonicity of $D_i$]
    The discussion of first-stage monotonicity is often important from a policy perspective.
    At the individual level, the incomplete-information setup and the uniqueness of equilibrium in the first stage imply a single-index threshold-crossing structure, namely, $D_i = \mathbf{1}\{ P_i(\bfZ,A) \ge V_i \}$. Because $P_i(\bfZ,A)$ serves as a sufficient one-dimensional index of $\bfZ$ in $D_i$, the first-stage monotonicity can be assessed by checking whether $P_i(\bfz',A) - P_i(\bfz,A)$ has the same sign for all $i$. 
    Under strategic substitution ($\lambda < 0$), an incentive that encourages a given unit $i$ to take up the treatment discourages its neighbors through the spillover effect $\lambda$, which generally leads to sign differences in $P_i(\bfz',A) - P_i(\bfz,A)$ across units. Even under strategic complementarity ($\lambda > 0$), monotonicity does not hold in general. 
\end{rmk}

\subsection{General Identification of the Second-Stage Parameters}

We now proceed with the identification of the second-stage parameters. In what follows, we maintain the following assumption.

\begin{asm}[Instruments]\label{asm:2nd_stage}
\begin{itemize}
    \item[]
    \item[(i)] $\bfZ$ is independent of $(\epsilon^{(D_i,\bfD_{-i})},\nu)$ given $\bfX$.
    \item[(ii)] $\bfP(\bfZ,A)=(P_1(\bfZ,A),...,P_n(\bfZ,A))$ is continuously distributed conditional on $\bfX$.
\end{itemize}
\end{asm}

Assumption~\ref{asm:2nd_stage}(i) extends Assumption~\ref{asm:unique_equ}(ii) from Section~\ref{sec:treatment_choice}. It requires that the instrument $\bfZ$ be exogenous in both the first and second stages, conditional on $\bfX$. 
In general, the presence of multiple equilibria may cause the equilibrium propensity score to vary in a non-smooth manner, in which case the existence of a continuous player-specific instrumental variable—say, $\bfZ_{1} = (Z_{11},\ldots,Z_{n1})$—is not sufficient to satisfy Assumption~\ref{asm:2nd_stage}(ii). However, once uniqueness of equilibrium is established under the logistic error specification (Lemma~\ref{lemma:UniBNE}), the implicit function theorem applies. As a result, the existence of at least one continuous coordinate in the player-specific instrument $Z_i$ is sufficient for Assumption~\ref{asm:2nd_stage}(ii).\footnote{As $Z_{i1}$ enters $P_i(\bfZ,A)$ through the linear index $Z_i'\beta_D$, and since our goal is to establish continuity of $\bfP$ with respect to this player-specific continuous coordinate $\bfZ_1$, we may, without loss of generality, restrict attention to the case in which $X_i$ is empty and $Z_i=Z_{i1}$. First, rewrite the system of equations in \eqref{eq:EQUcp} as $\bfP(\bfZ,A)=\varphi(\bfP(\bfZ,A),\bfZ)$. Define $\varPhi(\bfP(\bfZ,A),\bfZ)=\bfP(\bfZ,A)-\varphi(\bfP(\bfZ,A),\bfZ)$. Fix $\bfZ=\bfz$ and let $\bfP(\bfz)=\bfp$. It is immediate that $\varPhi(\bfp,\bfz)$ is continuously differentiable with respect to each component of $\bfp$ and $\bfz$.
The equilibrium condition implies that $\varPhi(\bfp,\bfz)=0$. It therefore remains to verify that $\nabla_{\bfp}\varPhi(\bfp,\bfz)$ is invertible. Let $F_i=F_{\nu|\bfX}(z_i\beta_D+\lambda\cdot n_i^{-1} A_{ik}p_k)$. Noting that $F_i'(\nu)=F_i(\nu)\cdot(1-F_i(\nu))$, we obtain
$
    {\partial F_i} / {\partial p_j} = F_i\cdot(1-F_i)\cdot \lambda\cdot n_i^{-1} A_{ij}.
$
We can then decompose $\nabla_{\bfp}\varPhi(\bfp,\bfz)$ as $\nabla_{\bfp}\varPhi(\bfp,\bfz)=\nabla_{\bfp}\bfp-\nabla_{\bfp}\varphi(\bfp,\bfz)=I-\lambda D_\bfp \widetilde{A}$, where $D_{\bfp}=\text{diag}\{F_1\cdot(1-F_1),\dots,F_n\cdot(1-F_n)\}$ and $\widetilde{A}$ denotes the row-normalized adjacency matrix. Consider the matrix infinity norm $\lVert\cdot\rVert_\infty$. We have $\lVert \lambda D_\bfp \widetilde{A} \rVert_\infty\leq |\lambda| \lVert D_\bfp \rVert_\infty \lVert \widetilde{A}\rVert_\infty=|\lambda|\cdot 1/4<1$. Hence, $I-\lambda D_\bfp\widetilde{A}$ is invertible. By the implicit function theorem, there exists a unique and continuously differentiable function $g^*$ such that $\bfp=g^*(\bfz)$. Finally, by uniqueness of equilibrium, $\bfp$ is a reduced-form function of $\bfz$, and therefore $\bfP(\bfZ)=g^*(\bfZ)$ is continuously differentiable in $\bfZ$.
}
This feature is common in the market entry literature, where the distance between a firm’s headquarters and a given market serves as a cost shifter that affects the firm’s own entry decision but does not directly influence its competitors’ profits, except through the firm’s entry decision.

In a single large-network setting, an individual can in principle have $2^n$ distinct potential outcome distributions, corresponding to all possible treatment configurations in the network. As a result, it is impossible to identify all parameters without imposing additional restrictions. Following the network causal inference literature \citep[e.g.,][]{Cai3:15,Leung:22ECTA,Hoshino2:23_JSAS}, we introduce a low-dimensional summary statistic, referred to as the exposure mapping $T(i,\bfD,A)$. Whenever there is little risk of confusion about the underlying network, we use the simplified notation $T_i(D_i,\bfD_{-i})$. Introducing this statistic amounts to assuming that
\begin{equation}\label{equ:expmap}
    T_i(d_i,\bfd_{-i})=T_i(d_i,\bfd'_{-i})\Longrightarrow Y_i(d_i,\bfd_{-i})=Y_i(d_i,\bfd'_{-i}).
\end{equation}
In other words, all interference effects influence potential outcomes solely through the exposure mapping. 
Consider the following examples commonly considered in the literature.\footnote{In the literature on strategic interactions outside the network context, interference is often assumed to operate through an aggregate state in empirical applications \citep{Menzel:2015}. For example, \citet{Berry:1992} study how the \textit{number} of entrants affects firms’ profits and, in turn, their entry decisions in airline markets. Similarly, \citet{Brock2:01} analyze how an individual’s binary choice depends on the \textit{proportion} of members in a reference group who take a given action.}

\begin{example}\label{examp:linearInMean}
    Suppose that the researcher is interested in both the direct effect of treatment and its spillover effect. The researcher considers the linear model
    \begin{equation*}
        Y_i=\beta_0+\beta_1 D_i + \beta_2 \cdot n_i^{-1} \sum_{j}A_{ij}D_j+u_i,
    \end{equation*}
    where $n_i^{-1} \sum_{j}A_{ij}D_j$ captures the influence of the share of treated neighbors. This setup is common in applied work \citep[e.g.,][]{Cai3:15}. The exposure mapping induced by this model takes the form $T_i(D_i,\bfD_{-i}):=\left(D_i,n_i^{-1} \sum_{j}A_{ij}D_j \right)$.
    \qed
\end{example}

\begin{example}\label{examp:exposure}
    Instead of the weighted share of treated neighbors, a researcher may be interested in simple exposure to at least one treated neighbor. In this case, the exposure mapping of the form
    $
        T_i(D_i,\bfD_{-i}):=\left(D_i, \mathbbm{1}\!\left(\sum_j A_{ij}D_j>0\right)\right)
    $
    may be considered.
    \qed
\end{example}

Let $l_A(i,j)$ denote the path distance between nodes $i$ and $j$ on the network represented by $A$. In addition, for $K>0$, let $N_A(i,K)=\{j:l_A(i,j)\leq K\}$ denote the $K$-neighborhood of node $i$. We also introduce the subnetwork $A_{N_A(i,K)}=(a_{kl}:k,l\in N_A(i,K))$ and subvector $\bfD_{N_A(i,K)}=(D_{k}:k\in N_A(i,K))$.
Withe this notation, consider the following assumption.

\begin{asm}[Exposure Mapping]\label{asm:exposure_mapping}
    There exists $K\in \mathbbm{N}$ such that for all $i$, $(A,\bfD)$, and $(A',\bfD')$, we have $T(i,\bfD,A)=T(i,\bfD',A')$ if $N_A(i,K)=N_{A'}(i,K)$, $A_{N_A(i,K)}=A'_{N_{A'}(i,K)}$ and $\bfD_{N_A(i,K)}=\bfD'_{N_{A'}(i,K)}$.
\end{asm}

In what follows, we denote the $K$-neighborhood $N_A(i,K)$ by $N_i$ for brevity. We also write $\bfV_{N_i}=(V_j)_{j\in N_i}$ and $\bfp_{N_i}=(p_j)_{j\in N_i}$ for the corresponding subvectors of random variables and their realizations, respectively. Letting $\partial_{\bfp_{N_i}}:=\partial^{|N_i|}/\big(\prod_{j\in N_i}\partial p_j\big)$, we define the observationally identifiable quantity
\begin{equation}\label{eq:identifier}
    \psi(t,\bfx,\bfp_{N_i},N_i)=\frac{\partial_{\bfp_{N_i}} E[\mathbbm{1}\{T_i=t\}Y_i\mid\bfX=\bfx,\bfP_{N_i}(\bfZ,A)=\bfp_{N_i}]}{\sum_{\bfd_{N_i}: T_i(\bfd_{N_i},\bfD_{N_i^c})=t} (-1)^{|\bfl_i'|}},
\end{equation}
where $\bfl_i'=\{j\in N_i\mid d_j=0\}$ denotes the set of local nodes whose treatment status equals zero. The definition of this quantity depends on the local network structure $N_i$. However, conditional on a given local network structure and the exposure mapping specified, the denominator is fully determined and involves no additional uncertainty.

The object in \eqref{eq:identifier} can be shown to identify the marginal exposure response (MER), defined as
\begin{equation}\label{equ:MER}
    MER(d,\bfd_{-i},\bfx,\bfp_{N_i})
    =
    E\!\left[ Y_i(d,\bfd_{-i}) \mid \bfX=\bfx,\bfV_{N_i}=\bfp_{N_i} \right],
\end{equation}
in the spirit of \citet{Hoshino2:23JoE}, but extended to a network setting. This result is formally stated in the following theorem.

\begin{thm}[Second-Stage Identification]\label{thm:mtr_iden}
If Assumptions \ref{asm:unique_equ}, \ref{asm:rank_cond}, \ref{asm:2nd_stage}, and \ref{asm:exposure_mapping} hold, then the
MER is identified as
\begin{equation}\label{equ:MTR_pNi}
    MER(d,\bfd_{-i},\bfx,\bfp_{N_i})=\psi(T_i(d,\bfd_{-i}),\bfx,\bfp_{N_i},N_i).
\end{equation}
\end{thm}

\noindent
See Appendix \ref{sec:proof:identifiadtion} for proof.

As in general causal inference frameworks, an individual’s potential outcomes $Y_i(d,\bfd_{-i})$ cannot be identified \textit{per se}. Theorem~\ref{thm:mtr_iden} shows that their conditional mean \eqref{equ:MER}, given the information generated by observed characteristics $\bfX$ and unobserved heterogeneity $\bfV_{N_i}$, can be identified by leveraging the exposure mapping $T_i$ as a low-dimensional summary statistic.

\subsection{Simplifying the Second-Stage Identification}\label{sec:identification:simplifying}
While Theorem~\ref{thm:mtr_iden} establishes identification of the MER under general conditions, this generality poses substantial challenges for practical estimation. The difficulty is that the estimand $\psi(t,\bfx,\bfp_{N_i},N_i)$ depends on the local neighborhood $N_i$ and the vector $\bfp_{N_i}$, among other objects. In other words, it depends on both the local network structure and a vector of continuous random variables whose dimension depends on it. This feature severely limits the effective (local) sample size when the identifying expression in Theorem~\ref{thm:mtr_iden} is used directly for estimation. In light of this practical challenge, the present subsection introduces additional model restrictions that enable dimension reduction and thereby facilitate estimation.

Reduce the dimension of the MER in \eqref{equ:MER} as
\begin{equation}\label{equ:MTR_pi}
    \overline{MER}(d,\bfd_{-i},\bfx,p_i)=E[Y_i(d,\bfd_{-i})|\bfX=\bfx,V_i=p_i].
\end{equation}
We will show that \eqref{equ:MTR_pi} can be identified by the following modification of \eqref{eq:identifier}:
\begin{align}
    &\overline{\psi}(t,\bfx,p,N_i)=
    \notag\\
    &\frac{\partial_{p} \left[ \sum_{\bfd_{N_i^\circ}: T_i(\bfd_{N_i},\bfD_{N^c_i})=t} (-1)^{|\Tilde{\bfl}'|} E[\mathbbm{1}\{T_i=t\}Y_i|\bfX=\bfx,P_i(\bfZ,A)=p,\bfP_{N_i^\circ}(\bfZ,A)=\bfd_{N_i^\circ}] \right]}{\sum_{\bfd_{N_i}: T_i(\bfd_{N_i},\bfD_{N^c_i})=t} (-1)^{|\bfl'|}},
    \label{eq:identifier_pi}
\end{align}
where $N_i^\circ=N_i\backslash \{i\}$ and $\Tilde{\bfl}' = \{j \in N_i^\circ|d_j=0\}$.
Furthermore, it will be shown that this estimand \eqref{eq:identifier_pi} can be significantly simplified under the following set of restrictions.

\begin{asm}[Restrictions for Simplified Identification]\label{asm:simplify_model}
    For each $i$:
    \begin{itemize}
        \item[(i)] $T_i(D_i,\bfD_{-i}) = \left(D_i,f_i(\bfD_{N_i^\circ})\right)$, where $f_i(\cdot)$ does not depend on the identities of neighbors, i.e., $f_i(\bfD_{N_i^\circ})=f_i(\bfD_{\pi_i(N_i^\circ)})$ for any permutation $\pi_i:N_i^\circ\to N_i^\circ$; and
        \item[(ii)] Given $T_i(\bfd)=t$, $Y_i(t)=\mu_X^{(t)}(X_i) + \mu_\epsilon^{(t)}(V_i)+e_i^{(t)}$ with $E[e_i^{(t)}|\bfX,\bfV_{N_i}]=0$.
    \end{itemize}
\end{asm}

Assumption~\ref{asm:simplify_model}(i) is satisfied under the anonymous interference assumption in \citet{Li2:22}, Example~\ref{examp:linearInMean} as in \citet{Cai3:15}, and Example~\ref{examp:exposure}. In general, the exposure mapping may depend on treatment status in an arbitrary manner. Under Assumption~\ref{asm:simplify_model}(i), however, we can clearly separate the role of an individual’s own treatment status $D_i$ from that of neighbors’ treatment assignments, thereby simplifying the latter. We note that the function $f_i$ remains quite general, as it may be multidimensional even under Assumption~\ref{asm:simplify_model}(i).
Assumption~\ref{asm:simplify_model}(ii) imposes an additively separable structure between the observed covariates $\bfX_i$ and the unobservables $\bfV_i$. Even with this restriction, however, the exposure indices $T_i(\bfd)$ may still affect potential outcomes in an arbitrary manner.

The following corollary to Theorem~\ref{thm:mtr_iden} formally establishes the dimension reduction and simplification for second-stage identification.

\begin{cor}[Second-Stage Identification, Simplified]\label{cor:simple_mtr_iden}
If Assumptions \ref{asm:unique_equ}, \ref{asm:rank_cond}, \ref{asm:2nd_stage}, and \ref{asm:exposure_mapping} hold, then
\begin{equation}
    \overline{MER}(d,\bfd_{-i},\bfx,p_i) =\overline{\psi}(T_i(d,\bfd_{-i}),\bfx,p_i,N_i).
\end{equation}
If Assumption \ref{asm:simplify_model} holds in addition, then, with $T_i(d,\bfd_{-i})=t=(t_1,t_2)$,
\begin{equation}
    \overline{\psi}(t,\bfx,p,N_i)=\left\{\begin{array}{cc}
        \partial_{p} \left[ p\cdot (\mu_X^{(t)}(x_i) + E[\mu_\epsilon^{(t)}(V_i)|\bfX=\bfx, V_i\leq p] )\right] & \text{ if } t_1=1,  \\
        \partial_{p} \left[ (p-1)\cdot  (\mu_X^{(t)}(x_i) + E[\mu_\epsilon^{(t)}(V_i)|\bfX=\bfx, V_i> p] )\right] & \text{ if } t_1=0. 
    \end{array} \right.
\end{equation}
\end{cor}

\noindent
See Appendix \ref{sec:proof:identifiadtion} for proof.

Corollary \ref{cor:simple_mtr_iden} implies that $\overline{\psi}(t,\bfx,p,N_i)$ no longer depends on $N_i$.
In light of this result, we suppress this argument and simply write $\overline{\psi}(t,\bfx,p)$ whenever this corollary is applied.
Define the marginal exposure effect (MEE) of changing $t$ to $t'$ by
\begin{equation}\label{equ:MEE}
    MEE(t',t,\bfx,p)=\overline{\psi}(t',\bfx,p)-\overline{\psi}(t,\bfx,p).
\end{equation}
For instance, the MEE with $t=(0,0)$ and $t'=(1,0)$ can be interpreted as the direct effect of own treatment in the absence of treated neighbors. The MEE with $t=(0,0)$ and $t'=(0,1)$ can be interpreted as the effect of exposure to treated neighbors on untreated individuals.

The following section proposes estimation strategies for these exposure effects.

\section{Estimation}\label{sec:estimation}

Following the identification steps in Section~\ref{sec:identification:first}, the first stage reduces to a logit discrete choice model on the network. Under Assumption~\ref{asm:unique_equ}(i), we may employ the log-likelihood
\begin{equation}
    \sum_i D_i \log(P_i(\bfZ,A) ) + (1-D_i) \log(1-P_i(\bfZ,A) ).
\end{equation}

Note that $P_i(\bfZ,A)$ is not directly observed in the data. Under Assumption~\ref{asm:unique_equ}(i) and the treatment selection model, we can compute the propensity scores $P_i(\bfZ,A)$ by solving for the fixed point of the following system of equations:
\begin{equation}
    P_i(\bfZ,A)
    = \frac{\text{exp}\!\left(Z_i'\beta_D + \lambda\cdot n_i^{-1} \sum_j A_{ij}P_j(\bfZ,A)\right)}
    {1+\text{exp}\!\left(Z_i'\beta_D + \lambda \cdot n_i^{-1} \sum_j A_{ij}P_j(\bfZ,A)\right)},
    \qquad \forall i=1,\ldots,n.
\end{equation}
Moreover, given the functional form of $P_i(\bfZ,A)$, we can compute the derivatives of $\bfP$ with respect to the parameters as follows:
\begin{equation}\label{equ:Pderi}
\begin{aligned}
    \left(\begin{array}{c}
         \partial P_1 / \partial \beta_{Dk}  \\
         \vdots\\
         \partial P_n / \partial \beta_{Dk}
    \end{array} \right) &=
    \left(\left(\begin{array}{ccc}
          \frac{1}{P_1(1-P_1)}& \cdots& 0  \\
          \vdots& \ddots&\vdots\\
          0 & \cdots & \frac{1}{P_n(1-P_n)}
    \end{array}   \right)- \lambda A \right)^{-1}\bfZ_k \quad\text{ for }k=0,\dots,\dim(Z),\\
   \left(\begin{array}{c}
         \partial P_1 / \partial \lambda  \\
         \vdots\\
         \partial P_n / \partial \lambda
    \end{array} \right) &=
    \left(\left(\begin{array}{ccc}
         \frac{1}{P_1(1-P_1)}& \cdots& 0  \\
          \vdots& \ddots&\vdots\\
          0 & \cdots & \frac{1}{P_n(1-P_n)}
    \end{array}   \right)- \lambda A \right)^{-1}
    \left(\begin{array}{c}
         \sum_{j=1}^n a_{1j} P_j  \\
         \vdots\\
         \sum_{j=1}^n a_{nj} P_j
    \end{array} \right),
\end{aligned}
\end{equation}
where $\bfZ_k=(Z_{ik})_{i=1}^n$, with $Z_{i0}=1$ for all $i$, corresponding to the intercept term.%
\footnote{In implementation, we clip $P_i$ at $10^{-5}$ from both ends. 
Note that Assumption \ref{asm:unique_equ} (iii) with $|\lambda|<4$ is sufficient to ensure invertibility of the relevant matrix.
To see this, let $D=\text{diag}\!\left(\frac{1}{P_i(1-P_i)}\right)$. Factoring out $D^{-1}$, it suffices to examine the invertibility of $I-\lambda D^{-1}A$. This matrix is invertible if the product of $\lambda$ and the spectral radius of $D^{-1}A$ lies strictly between $-1$ and $1$. Moreover, the spectral radius of $D^{-1}A$ is bounded above by the maximum row sum, which equals $\max_i P_i(1-P_i)\le 1/4$. Hence, $|\lambda|<4$ is a universally sufficient condition.}

Given the derivatives with respect to the parameters, we obtain the first-order conditions for the maximum likelihood estimator (MLE) as
\begin{equation}
\begin{aligned}
    g_i^{\beta_{Dk}} &= (D_i - P_i)(Z_{ik} + \lambda\sum_j A_{ij}\frac{\partial P_j}{\partial \beta_{Dk}}) \quad\text{ for }k=0,\dots,\dim(Z)),\\
    g_i^{\lambda} &= (D_i - P_i)(\sum_j A_{ij}P_j + \lambda\sum_j A_{ij}\frac{\partial P_j}{\partial \lambda}).
\end{aligned}
\end{equation}
These first-stage score functions are subsequently used as moments in our Generalized Method of Moments (GMM) estimation.

For the second stage, we consider a concrete specification following Section~\ref{sec:identification:simplifying}. Assumption~\ref{asm:simplify_model} imposes only an additively separable structure. We parametrize the functions $\mu_X^{(t)}$ and $\mu_\epsilon^{(t)}$ and adopt a linear-in-parameters specification that is nested within Assumption~\ref{asm:simplify_model}.

\begin{asm}[Linear-in-Parameters Specification]\label{asm:est}
\begin{itemize}
    \item[]
    \item[(i)] $\mu_X^{(t)}(X_i)=\mu_X^{(t)}(X_i;\beta_X^{(t)})=X_i'\beta_X^{(t)}$ and $\mu_\epsilon^{(t)}(V_i)=\mu_\epsilon^{(t)}(V_i;\beta_p^{(t)})=V_i\beta_p^{(t)}$.
    \item[(ii)] $(e^{(t)}_i,\nu_i)$ is independent of $(\bfX,\bfZ)$ for all $t \in\mathcal{T}$. 
\end{itemize}    
\end{asm}

\noindent
Under Assumption \ref{asm:est}, the estimand from Corollary \ref{cor:simple_mtr_iden} becomes
\begin{equation}\label{equ:MER_withEstAssumption}
    \overline{\psi}(t,\bfx,p)=\left\{\begin{array}{cc}
        \partial_{p} \left[ p\cdot ( x'\beta_X^{(t)} + E[V_i|V_i\leq p]\beta_p^{(t)}  )\right], & \text{ if } t_1=1,  \\
        \partial_{p} \left[ (p-1)\cdot (x'\beta_X^{(t)} + E[V_i|V_i> p]\beta_p^{(t)}  )\right], & \text{ if } t_1=0. 
    \end{array} \right.
\end{equation}
The parametric assumption further implies that $\overline{\psi}$ now depends only on $x$, rather than on the covariates of the entire sample $\bfx$. 

We next discuss how to estimate the coefficients $(\beta_X^{(t)},\beta_\epsilon^{(t)})$. In what follows, we focus on the case $t_1=1$; the case $t_1=0$ follows analogously. Assumptions~\ref{asm:simplify_model} and~\ref{asm:est}(i) then imply the following linear equation
\begin{equation*}
    Y_i(t) = X_i'\beta_X^{(t)} + V_i \beta_p^{(t)} + e_i^{(t)}.
\end{equation*}
However, this approach is infeasible because $V_i$ is unobserved. Alternatively, consider the following linear estimating equation.
\begin{equation*}
    \mathbbm{1}\{T_i=t\}Y_i = \mathbbm{1}\{T_i=t\}X_i'\beta_X^{(t)}+\mathbbm{1}\{T_i=t\}V_i\beta_p^{(t)} + \mathbbm{1}\{T_i=t\}e_i^{(t)}.
\end{equation*}
Taking the conditional expectation on the both sides and noting $t_1=1$, we have
\begin{equation}
\begin{aligned}
    E[&\mathbbm{1}\{T_i=t\}Y_i|\mathbbm{1}\{T_i=t\},\bfX,P_i,P_{N_i^\circ}]\\
    &= \mathbbm{1}\{T_i=t\}X_i'\beta_X^{(t)}+\mathbbm{1}\{T_i=t\}E[V_i|\mathbbm{1}\{T_i=t\},\bfX,P_i,P_{N_i^\circ}]\beta_p^{(t)}\\
    &\hspace{10mm}+\mathbbm{1}\{T_i=t\}E[e_i^{(t)}|\mathbbm{1}\{T_i=t\},\bfX,P_i,P_{N_i^\circ}]\\
    & = \mathbbm{1}\{T_i=t\}X_i'\beta_X^{(t)}+\mathbbm{1}\{T_i=t\}E[V_i|V_i < P_i]\beta_p^{(t)}+\\
    &\hspace{10mm}\mathbbm{1}\{T_i=t\}E[e_i^{(t)}|\bfX, \cup_{\bfd_{N_i}:T_i(d_i,\bfd_{N_i^\circ},\bfD_{N_i^c})=t}\{ \bfV_{\bfl}\leq \bfP_{\bfl},\bfV_{\bfl'}> \bfP_{\bfl'} \}] \\
    &= \mathbbm{1}\{T_i=t\}X_i'\beta_X^{(t)}+\mathbbm{1}\{T_i=t\}\frac{P_i}{2}\beta_p^{(t)},
\end{aligned}
\end{equation}
where 
the second equality follows from Assumptions \ref{asm:unique_equ}, \ref{asm:2nd_stage} (i), and \ref{asm:est} (ii),
and the last equality is due to $E[e_i^{(t)}|\bfX,\bfV_{N_i}]=0$ under Assumption \ref{asm:simplify_model} (ii) and $E[V_i|V_i < P_i]=P_i/2$ for the uniform random variable. 

This observation motivates the feasible linear estimating equation
\begin{equation}
    \mathbbm{1}\{T_i=t\}Y_i = \mathbbm{1}\{T_i=t\}X_i'\beta_X^{(t)}+\mathbbm{1}\{T_i=t\}\frac{P_i}{2}\beta_p^{(t)} + \Tilde{e}_i^{(t)},
\end{equation}
where $E[\Tilde{e}_i^{(t)}|X_i,P_i]=0$ by construction. 
Thus, as the first-order conditions of the resulting least-squares criterion, we can incorporate the following moments into our GMM estimation:
\begin{equation}
\begin{aligned}
    g_i^{\beta_X^{(t)}} &= \mathbbm{1}\{T_i=t\} 2X_i\left(Y_i-X_i'\beta_Y^{(t)}- \frac{P_i}{2}\beta_p^{(t)}\right),\\
    g_i^{\beta_p^{(t)}}&=\mathbbm{1}\{T_i=t\}P_i\left(Y_i-X_i'\beta_Y^{(t)}- \frac{P_i}{2}\beta_p^{(t)}\right).
\end{aligned}
\end{equation}
For the case $t_1=0$, an analogous argument applies by replacing $E[V_i \mid V_i < P_i]$ with $E[V_i \mid V_i \ge P_i]$, which yields $1+P_i$ in place of $P_i$.

Since our objective is to learn about the effect of changing the exposure mapping from $t$ to $t'$, we need to estimate the parameters corresponding to both $t$ and $t'$. Thus, we use the following moment function
\begin{equation}
    \bm{g}((Y_i,X_i,D_i,T_i,\bfZ),\bm{\beta}) = \left(g_i^{\beta_D}, g_i^{\lambda}, g_i^{\beta_X^{(t)}}, g_i^{\beta_p^{(t)}}, g_i^{\beta_X^{(t')}}, g_i^{\beta_p^{(t')}}\right)',
\end{equation}
where $\bm{\beta} = (\beta_D,\lambda,\beta_X^{(t)},\beta_p^{(t)},\beta_X^{(t')},\beta_p^{(t')})$.
Its dimension, $\dim(\bm{g})=(\dim(Z)+1)+1+2(\dim(X)+1)$, corresponds to the number of unknown parameters in $\bm{\beta}$, yielding just identification in the current specification.
With this said, the general GMM estimator takes the form of
\begin{equation}
    \hat{\bm{\beta}} = \underset{\bm{\beta}}{\arg\,\min}  \ \widehat{\bm{g}}_n(\bm{\beta})'\widehat{\Xi}_n\widehat{\bm{g}}_n(\bm{\beta}),
\end{equation}
where $\widehat{\bm{g}}_n(\bm{\beta})=n^{-1}\sum_i^n \bm{g}((Y_i,X_i,D_i,T_i,\bfZ),\bm{\beta})$, and $\widehat{\Xi}_n$ is a sequence of positive definite weight matrices.

In practice, it is desirable to implement a two-step GMM procedure to improve efficiency in general over-identified settings. In the two-step GMM, one begins by obtaining the first-step estimate
\begin{equation}\label{equ:2stepGMM_start}
    \widehat{\bm{\beta}}^{(1)} = \underset{\bm{\beta}\in\Theta}{\arg\,\min}\, \ \widehat{\bm{g}}_n(\bm{\beta})'\widehat{\bm{g}}_n(\bm{\beta}).
\end{equation}
We take advantage of the network HAC estimator proposed by KMS to construct an optimal weighting matrix.\footnote{To prevent the HAC estimator from being non–positive semi-definite, we apply an eigenvalue (diagonal) decomposition to the variance–covariance matrix and replace any eigenvalues smaller than $10^{-10}$ with this threshold. In our simulations, a small number of iterations required this adjustment when the sample size was 250. However, once the sample size increased to 1000, no iterations required such adjustments, suggesting that the variance estimator converges to a strictly positive-definite matrix.}
With the shorthand notation $\bm{g}_i(\widehat{\bm{\beta}}^{(1)})=\bm{g}((Y_i,X_i,D_i,T_i,\bfZ),\widehat{\bm{\beta}}^{(1)})$, the network HAC estimator of the variance of $\bm{g}$ is given by
\begin{align*}
    \widehat{\Omega}_{g,n} =& \sum_{s\geq 0} \omega_n(s)\widehat{\Sigma}_n(s), \qquad\text{where}
\\
\widehat{\Sigma}_n(s)=&\frac{1}{n}\sum_{i=1}^n \sum_{j\in N(i,s) } (\bm{g}_i(\widehat{\bm{\beta}}^{(1)})-\widehat{\bm{g}}_n(\widehat{\bm{\beta}}^{(1)}))(\bm{g}_j(\widehat{\bm{\beta}}^{(1)})-\widehat{\bm{g}}_n(\widehat{\bm{\beta}}^{(1)}))',
\end{align*}
$\omega_n(s)=\omega(s/b_n)$, $\omega$ is a pre-specified kernel function, and $b_n$ denotes the bandwidth. 
We then obtain the two-step GMM estimator
\begin{equation}\label{equ:2stepGMM_end}
    \widehat{\bm{\beta}} = \underset{\bm{\beta}\in\Theta}{\arg\,\min}\,\widehat{\bm{g}}_n(\bm{\beta})'\widehat{\Xi}_n\widehat{\bm{g}}_n(\bm{\beta}),
    \quad\text{ where }
    \widehat{\Xi}_n=\widehat{\Omega}_{g,n}^{-1}.
\end{equation}

With the parameter estimate $(\widehat{\beta}_X^{(t)},\widehat{\beta}_\epsilon^{(t)})$, we in turn estimate the MER by the plug-in
\begin{equation}\label{equ:linear_MER}
\begin{aligned}
    \widehat{\overline{\psi}}(t,\bfx,p)&=\left\{ \begin{array}{ll}
    \partial_{p} \left[p\times(x'\widehat{\beta}_X^{(t)} + \frac{p}{2}\widehat{\beta}_\epsilon^{(t)})\right]     &\text{if } t_1=1,  \\
    \partial_{p} \left[(p-1)\times(x'\widehat{\beta}_X^{(t)} + \frac{1+p}{2}\widehat{\beta}_\epsilon^{(t)})\right]     &\text{if } t_1=0.
    \end{array} \right.\\
    & = x'\widehat{\beta}_X^{(t)} + p \widehat{\beta}_\epsilon^{(t)}
\end{aligned}
\end{equation}
The marginal exposure effect of changing the exposure mapping from $t$ to $t'$ is estimated by
\begin{equation}
   \widehat{MEE}(t',t,\bfx,p)=\widehat{\overline{\psi}}(t',\bfx,p)-\widehat{\overline{\psi}}(t,\bfx,p)=x'(\widehat{\beta}_X^{(t')}-\widehat{\beta}_X^{(t)}) + p(\widehat{\beta}_\epsilon^{(t')}-\widehat{\beta}_\epsilon^{(t)}).
\end{equation}

\section{Econometric Theory}\label{sec:theory}

\subsection{Network Properties }

Consider a triangular array of random variables, where $\mathcal{N}_n=\{1,\dots,n\}$ denotes the index set in the $n$th row. Accordingly, we extend the notation $W_i$ to $W_{n,i}$ to emphasize that the random variables belong to the $n$th row of the triangular array. Boldface notation $\mathbf{W}_n=(W_{n,i})_{i=1}^n$ is used to denote the entire vector of random variables in the $n$th row.

Samples are drawn from a probability space $(\Omega,\mathcal{F},P)$. For each $n$, let $\mathcal{C}_n\subseteq\mathcal{F}$ denote a $\sigma$-subalgebra representing common shocks. In particular, these common shocks determine both the adjacency matrix $A_n$ and the covariates $\{X_{n,i},Z_{n,i}\}_{i=1}^n$. Moreover, there may exist a $\sigma$-algebra $\mathcal{C}$ representing common factors that is contained in every $\mathcal{C}_n$, that is, $\mathcal{C}\subset\bigcap_{n\ge 1}\mathcal{C}_n$.

We first fix some basic notation.
Define a metric \(d_h\) on $\mathbb{R}^{h \times k}$ by 
\[
    d_h(\mathbf{W}, \widetilde{\mathbf{W}}) = \sum_{r=1}^h \lVert w_r - \widetilde{w}_r \rVert_2
\]
for \(\mathbf{W}, \widetilde{\mathbf{W}} \in \mathbb{R}^{h \times k}\) written as 
\(\mathbf{W} = (w_1, \dots, w_k)\) and \(\widetilde{\mathbf{W}} = (\widetilde{w}_1, \dots, \widetilde{w}_k)\),
where \(\lVert \cdot \rVert_2\) denotes the Euclidean norm. 
Also define its conditional \(L^p\) norm $\lVert \cdot \rVert_{\mathcal{C}_n, L^p}$ with respect to \(\mathcal{C}_n\) by 
\[
    \lVert W_i \rVert_{\mathcal{C}_n, L^p} 
    = \left( E \left[ \sum_{r=1}^k |W_{ir}|^p \,\middle|\, \mathcal{C}_n \right] \right)^{1/p}
\]
for a vector \(W_i = (W_{i1}, \dots, W_{ik}) \in \mathbb{R}^k\).
Its unconditional counterpart is denoted by \(\lVert \cdot \rVert_{L^p}\).

Let \(\mathcal{P}_n(h, h', s)\) denote the collection of all pairs of subsets of nodes of sizes 
\(h\) and \(h'\) such that the path distance between any two nodes from different subsets 
is at least \(s\).
Define
\[
    \mathcal{L}_{w,a} = \left\{ f : \mathbb{R}^{w \times a} \to \mathbb{R} \;:\; 
    \lVert f \rVert_{\infty} < \infty,\; \mathrm{Lip}(f) < \infty \right\}
\]
as the set of bounded Lipschitz functions, where 
\(\lVert f \rVert_{\infty} = \sup_{x} |f(x)|\) and \(\mathrm{Lip}(f)\) denotes the Lipschitz constant of \(f\).
Following KMS, we now formally define the notion of the conditional \(\psi\)-dependence.

\begin{defi}\label{def:psidep}
Let \(\{W_{n,i}\}\) be a triangular array of random vectors taking values in 
\(\mathbb{R}^{\dim(W)}\). 
We say that \(\{W_{n,i}\}\) is \emph{conditionally \(\psi\)-dependent} if, for all pairs 
of node sets \(H\) and \(H'\) of sizes \(h\) and \(h'\), respectively, and for all 
functions \(f \in \mathcal{L}_{\dim(W),h}\) and \(f' \in \mathcal{L}_{\dim(W),h'}\), 
there exists a sequence of \(\mathcal{C}_n\)-measurable constants 
\(\{\theta_{n,s}^W\}_{s > 0}\) and a collection of non-random functions
$
    \{\psi_{a,b}^W\}_{a,b \in \mathbb{N}}, 
    \psi_{a,b}^W : \mathcal{L}_a \times \mathcal{L}_b \to [0, \infty),
$
such that
\[
    \left| \operatorname{Cov}\!\left( f(W_{n,H}),\, f'(W_{n,H'}) \,\middle|\, \mathcal{C}_n \right) \right|
    \leq \psi_{h,h'}^W(f, f')\, \theta_{n,s}^W,
\]
for all \((H, H') \in \mathcal{P}_n(h, h', s)\).
\end{defi}
The central idea of the \(\psi\)-dependence is that the influence of distant observations on local observations diminishes as the path distance between them increases. 
KMS develop a fundamental asymptotic theory under this notion of network dependence, of which we can take advantage.
Thus, as a first step in our analysis, we characterize the \(\psi\)-dependence properties 
of the individual moment conditions for our network model.
Conditioning on \(\mathcal{C}_n\), the moment condition for sample \(i\) takes the 
treatment status \(D_{n,i}\), the exposure mapping \(T_{n,i}\), and the outcome 
\(Y_{n,i}\) as its arguments. 
Accordingly, we will establish the \(\psi\)-dependence of $\{W_{n,i}\}$ conditional on \(\mathcal{C}_n\), where
$
    W_{n,i} = (Y_{n,i},\, T_{n,i},\, D_{n,i}).
$

Note that $\{D_{n,i}\}_{i=1}^n$ is a transformation of $\bfZ_n=\{Z_{n,i}\}_{i=1}^n$ and $\bm{\nu}_n=\{\nu_{n,i}\}_{i=1}^n$, as
$
    D_{n,i}=\sigma_{n,i}^D(\bm{\nu}_n,\bfZ_n).    
$
Conditional on $\mathcal{C}_n$, $D_{n,i}$ can be viewed as a function of $\bm{\nu}_n$. Thus, we suppress the argument of $\bfZ_n$ and simply write $D_{n,i}=\sigma_{n,i}^D(\bm{\nu}_n)$ for brevity. We also define 
$
    D_{n,i}^{(s)}=\sigma_{n,i}^D(\bm{\nu}_n^{(s,i)}),   
$
where $\bm{\nu}_n^{(s,i)} = \left( \mathbbm{1}\{j\in N_{A_n}(i,s)\}\cdot \nu_{n,j} \right)_{j=1}^n$ replaces the distant shock with zero. In the following proposition, this quantity will be shown to play a role in characterizing the $\psi$-dependence of $\{D_{n,i}\}$.

\begin{prop}\label{prop:D_psiDep}
Suppose that \(\{\sigma_{n,i}^D\}_{i=1}^n\) is a sequence of measurable functions. 
Then, for all \(H, H' \in \mathcal{P}(h, h', 2s + 1)\) and for all 
\(f \in \mathcal{L}_{1,h}\) and \(f' \in \mathcal{L}_{1,h'}\), we have
\begin{equation}
    \left| \operatorname{Cov}\left( f(\mathbf{D}_{n,H}),\, f'(\mathbf{D}_{n,H'}) \,\middle|\, \mathcal{C}_n \right) \right| 
    \leq \psi_{h,h}^D(f,f')\, \theta_{n,s}^D \hspace{2mm} \text{a.s.},
\end{equation}
where \(\psi_{h,h}^D(f,f') = h \lVert f' \rVert_{\infty} \mathrm{Lip}(f) + 
h' \lVert f \rVert_{\infty} \mathrm{Lip}(f')\) and 
\(\theta_{n,s}^D = 2 \, \underset{1 \leq i \leq n}{\max} \, 
E[ | D_{n,i} - D_{n,i}^{(s)} | \mid \mathcal{C}_n ]\)
\end{prop}

Similarly, the outcomes \(\{Y_{n,i}\}_{i=1}^n\) are transformations of 
\((\mathbf{D}_n, \mathbf{X}_n, \mathbf{Z}_n, \bm{\nu}_n, \{\bm{e}_n^{(t)}\}_{t \in \mathcal{T}})\), 
where the dependence on \(\mathbf{D}_n\) arises through the exposure mapping. 
Suppressing \(\mathbf{X}_n\) and \(\mathbf{Z}_n\), we denote this as
$
    Y_{n,i} = \sigma_{n,i}^Y(\mathbf{D}_n, \nu_{n,i}, \{ e_{n,i}^{(t)} \}_{t \in \mathcal{T}}).
$
We note that while \(\mathbf{D}_n\) depends on the entire vector \(\bm{\nu}_n\), only 
\(\mathbf{D}_n\), \(\nu_{n,i}\), and \(\{ e_{n,i}^{(t)} \}_{t \in \mathcal{T}}\) directly enter 
\(Y_{n,i}\) in the potential outcome, either as part of the regressor 
\(V_i = F_{\nu|\mathbf{X}}(\nu_{n,i})\) or as error terms. 
We similarly define
$
    Y_{n,i}^{(s)} = \sigma_{n,i}^Y(\mathbf{D}_n^{(s,i)}, \nu_{n,i}, \{ e_{n,i}^{(t)} \}_{t \in \mathcal{T}}),
$
with $\bfD_n^{(s,i)} = \left( \mathbbm{1}\{j\in N_{A_n}(i,s)\}\cdot D_{n,j} \right)_{j=1}^n$. Although \(\nu_{n,i}\) and \(\{ \epsilon_{n,i}^{(t)} \}\) are i.i.d. random variables conditional 
on \(\mathcal{C}_n\), the treatment assignments \(\{ D_{n,i} \}\) are \(\psi\)-dependent.
Thus, we cannot directly apply the previous result. Nevertheless, we can still show that 
\(\{Y_{n,i}\}\) is conditionally \(\psi\)-dependent. We state the result in the following theorem.

\begin{prop}\label{prop:Y_psiDep}
Suppose Assumption in Proposition \ref{prop:D_psiDep} hold.
Let \(\{\sigma_{n,i}^Y\}_{i=1}^n\) be a sequence of measurable functions.
For all 
\(H, H' \in \mathcal{P}(h, h', 4s + 1)\) and for all 
\(f \in \mathcal{L}_{1,h}\) and \(f' \in \mathcal{L}_{1,h'}\), we have
\begin{equation*}
    \left| \operatorname{Cov}\left( f(\mathbf{Y}_{n,H}),\, f'(\mathbf{Y}_{n,H'}) \,\middle|\, \mathcal{C}_n \right) \right| 
    \leq \psi_{h,h}^Y(f,f')\, \theta_{n,s}^Y \hspace{2mm} \text{a.s.},
\end{equation*}
where
    $\theta_{n,s}^Y = 2 \underset{1 \leq i \leq n}{\max}\, E\!\left[\, | Y_{n,i} - Y_{n,i}^{(s)} | \,\middle|\, \mathcal{C}_n \right] 
    + \overline{N}_{A_n}(s) \times \theta_{n,s}^D$,
    $\psi_{h,h}^Y(f,f') = h \lVert f' \rVert_{\infty} \mathrm{Lip}(f) \, (1 + \overline{\sigma}^Y) 
    + h' \lVert f \rVert_{\infty} \mathrm{Lip}(f') \, (1 + \overline{\sigma}^Y)$,
\(\overline{N}_{A_n}(s) = \underset{1 \leq i \leq n}{\max}\, | N_{A_n}(i, s) |\), and
$
    \overline{\sigma}^Y = \underset{n \geq 1}{\sup}\, \underset{1 \leq i \leq n}{\max}\, \mathrm{Lip}(\sigma_{n,i}^Y).
$
\end{prop}

We now move on to characterize the conditional \(\psi\)-dependence of the exposure mapping 
\(\{T_{n,i}\}\). There are two distinct features to note. First, unlike the outcome variables, 
\(T_{n,i}\) depends only on \(\mathbf{D}_n\). Second, \(T_{n,i}\) is a vector-valued function. 
Bearing these two features in mind, we write
\begin{equation*}
    T_{n,i} = \sigma_{n,i}^T(\mathbf{D}_n) = \bigl( \sigma_{n,i}^{T,1}(\mathbf{D}_n), \dots, \sigma_{n,i}^{T,\dim(\mathcal{T})}(\mathbf{D}_n) \bigr)
\end{equation*}
to emphasize its vector-valued nature. Similarly, we define
\begin{equation*}
    T_{n,i}^{(s)} = \sigma_{n,i}^T(\mathbf{D}_n^{(s,i)}) = \bigl( \sigma_{n,i}^{T,1}(\mathbf{D}_n^{(s,i)}), \dots, \sigma_{n,i}^{T,\dim(\mathcal{T})}(\mathbf{D}_n^{(s,i)}) \bigr).
\end{equation*}


\begin{prop}\label{prop:T_psiDep}
Suppose Assumption in Proposition \ref{prop:D_psiDep} hold.
Let \(\{\sigma_{n,i}^T\}_{i=1}^n\) be a sequence of measurable functions. For all 
\(H, H' \in \mathcal{P}(h, h', 4s + 1)\) and for all 
\(f \in \mathcal{L}_{\dim(\mathcal{T}),h}\) and 
\(f' \in \mathcal{L}_{\dim(\mathcal{T}),h'}\), we have
\begin{equation*}
    \left| \operatorname{Cov}\left( f(\mathbf{T}_{n,H}),\, f'(\mathbf{T}_{n,H'}) \,\middle|\, \mathcal{C}_n \right) \right| 
    \leq \psi_{h,h}^T(f,f')\, \theta_{n,s}^T \hspace{2mm} \text{a.s.},
\end{equation*}
where
$
    \theta_{n,s}^T = 2 \underset{i}{\max}\, E\!\left[ \lVert T_{n,i} - T_{n,i}^{(s)} \rVert_2 \,\middle|\, \mathcal{C}_n \right] 
    + \overline{N}_{A_n}(s) \times \theta_{n,s}^D,
$,
$
    \psi_{h,h}^T(f,f') = h \lVert f' \rVert_{\infty} \mathrm{Lip}(f) (1 + \overline{\sigma}^T) 
    + h' \lVert f \rVert_{\infty} \mathrm{Lip}(f') (1 + \overline{\sigma}^T),
$
\(\overline{N}_{A_n}(s) = \underset{i}{\max}\, | N_{A_n}(i, s) |\), and
$
    \overline{\sigma}^T = \underset{n \geq 1}{\sup}\, \underset{i}{\max}\, 
    \sum_{r = 1}^{\dim(\mathcal{T})} \mathrm{Lip}(\sigma_{n,i}^{T,r}).
$
\end{prop}

We omit the proof of this proposition, as it follows the same structure as the proof of 
Proposition~\ref{prop:Y_psiDep}. Finally, we are ready to establish the conditional \(\psi\)-dependence of \(\{W_{n,i}\}\) with \(W_{n,i} = (Y_{n,i}, T_{n,i}, D_{n,i})\).

\begin{prop}\label{prop:W_psiDep}
    Suppose Assumptions in Proposition \ref{prop:D_psiDep}, \ref{prop:Y_psiDep}, and \ref{prop:T_psiDep} hold.
    For all $H,H'\in \mathcal{P}(h,h',4s+1)$, and all $f\in \mathcal{L}_{2+\dim(\mathcal{T}),h}, f'\in \mathcal{L}_{2+\dim(\mathcal{T}),h'}$, we have
    \begin{equation}
        \left| \Cov(f(\bfW_{n,H}),f'(\bfW_{n,H'})|\mathcal{C}_n) \right| \leq \psi_{h,h}^W(f,f') \theta_{n,s}^W\hspace{2mm} a.s.,
    \end{equation}
    where
$
        \theta_{n,s}^W=2(\underset{1 \leq i \leq n}{\max}E[| Y_{n,i}-Y_{n,i}^{(s)}| |\mathcal{C}_n] + \underset{1 \leq i \leq n}{\max}E[\lVert T_{n,i}-T_{n,i}^{(s)}\rVert_2 |\mathcal{C}_n])+ \overline{N}_{A_n}(s) \times \theta_{n,s}^D
$$
        \psi_{h,h}^W(f,f')= h\lVert f' \rVert_{\infty} \mathrm{Lip}(f) (1+\overline{\sigma}^Y+\overline{\sigma}^T) + h' \lVert f \rVert_{\infty} \mathrm{Lip}(f')(1+\overline{\sigma}^Y+\overline{\sigma}^T),
$
    $\overline{N}_{A_n}(s)=\underset{1 \leq i \leq n}{\max}\, |N_{A_n}(i,s)|$,
$
        \overline{\sigma}^Y=\underset{n\geq 1}{\sup}\, \underset{1 \leq i \leq n}{\max}\, \mathrm{Lip}(\sigma_{n,i}^Y),
$
and
$
        \overline{\sigma}^T=\underset{n\geq 1}{\sup}\, \underset{1 \leq i \leq n}{\max}\, \sum_{r=1}^{\dim(\mathcal{T})} \mathrm{Lip}(\sigma_{n,i}^{T,r}).
$
\end{prop}

Observe that \(W_{n,i}\) is a vector formed by stacking \(Y_{n,i}\), \(T_{n,i}\), and 
\(D_{n,i}\), with the first two components constructed from the last. 
It therefore follows that the conditional \(\psi\)-dependence of \(\{W_{n,i}\}\) 
inherits both the functional form and the \(\theta\)-sequence from those of 
\(\{Y_{n,i}\}\) and \(\{T_{n,i}\}\). As such, the result immediately follows from propositions \ref{prop:D_psiDep}, \ref{prop:Y_psiDep}, and \ref{prop:T_psiDep}.

\subsection{Asymptotic Properties}

We now establish the asymptotic properties of our estimator and discuss their implications for statistical inference. We extend the analysis of \citet{Sasaki:25} to derive large-sample properties conditional on the common sub-$\sigma$-field $\mathcal{C}$. The key advantage of this conditional framework is that it allows the underlying network-generating process to be non-deterministic. As shown by KMS, this approach accommodates a broad class of models, including network formation models, random fields on graphs, and conditional dependency graphs.

We are going to use the following version of the definition of stable convergence, as given by \citet[Definition~A3.2.III, p.~419]{Daley2:03}.

\begin{defi}\label{def:stable_cov}
    Let $\{W_n\}_{n\geq 1}$ and $W$ be $\mathbbm{R}^{p}$-valued random vectors on a probability space $(\Omega,\mathcal{F},P)$ and $\mathcal{C}\subseteq\mathcal{F}$ be sub-$\sigma$-field. We say
    \begin{equation}
        W_n\overset{d}{\to} W\hspace{2mm} (\mathcal{C}\stably)
    \end{equation}
    if for all $U \in \mathcal{F}$ and all Borel sets $A$ from Borel $\sigma$-field $\mathcal{B}(\mathbbm{R}^{p})$ with $P(W_n\in \bdd A)=0$,
    \begin{equation}
        \underset{n\to \infty}{\lim} P(\{W_n\in A\} \cap U) = P(\{W\in A\} \cap U)
    \end{equation}
    holds as $n\to \infty$, where $\bdd A$ denotes the boundary of set $A$.
\end{defi}

There exist several equivalent formulations of stable convergence, as noted by \citet{Daley2:03} and \citet{Häusler2:15}. We summarize these equivalences in Proposition~\ref{app_prop:Daley_stableEQU} in Appendix~\ref{sec:useful_propositions}. A special case arises when $W$ is $\mathcal{C}$-measurable, in which case stable convergence is equivalent to convergence in probability \citep[][Corollary~3.6, p.~26]{Häusler2:15}. Since the true parameter vector $\bm{\beta}_0$ is treated as fixed, this corollary implies the usual convergence-in-probability result for consistency.

We begin by establishing the consistency of the GMM estimator. Let $N_{A_n}^{\partial}(i,s)$ denote the set of nodes that are exactly $s$ path-lengths away from node $i$, and define
\begin{equation}
    \delta^{\partial}_n(s) = \frac{1}{n} \sum_{i=1}^n \big| N_{A_n}^{\partial}(i,s) \big|.
\end{equation}
With this notation, we impose the following assumptions.

\begin{asm}\label{asm:regularityLLN}
\begin{itemize}
    \item[]
    \item[(i)] There exists $\epsilon>0$ such that $\lVert W_{n,i}\rVert_{\mathcal{C}_n,L^{1+\epsilon}}<\infty\hspace{2mm} a.s.$
    \item[(ii)] $\frac{1}{n}\sum_{s\geq 1} \delta^{\partial}_n(s)\theta_{n,s}^W\to 0\hspace{2mm} a.s.$
    \item[(iii)] $\bm{\beta}\in \Theta$ and $\Theta$ is a compact subset of $\mathbbm{R}^{\dim(g)}$.
    \item[(iv)] For all $\bm{\beta}\in \Theta$ and for all $(x,\bfz)$ in the support of $(X_{n,i},\bfZ)$ for all $n$, $\bm{g}((\cdot,x,\bfz),\bm{\beta})$ belongs to $\mathcal{L}_{3,1}$ .
    \item[(v)](Uniform Equicontinuity) There exists $\overline{L}>0$ such that for all $(w,x,\bfz)$ in the support of $(W_{n,i},X_{n,i},\bfZ)$ for all $n$ and for all $\bm{\beta},\bm{\beta}'\in\Theta$,
    \begin{equation}\label{equ:unifequ}
        \lVert \bm{g}((w,x,\bfz),\bm{\beta}) - \bm{g}((w,x,\bfz),\bm{\beta}')\rVert_2 \leq \overline{L}\lVert \bm{\beta}-\bm{\beta}'  \rVert_2.
    \end{equation}
\end{itemize}    
\end{asm}

Parts~(i) and~(ii) of Assumption~\ref{asm:regularityLLN} correspond to Assumptions~3.1 and~3.2 of KMS, respectively. 
These conditions are used to establish the pointwise law of large numbers. 
Notably, Assumption~\ref{asm:regularityLLN}(ii) characterizes the trade-off between network density and cross-sample dependence. 
A standard requirement for the consistency of nonlinear GMM estimators is the uniform convergence of the GMM objective function \citep[e.g.,][]{Newey2:94, Vaart:98}. 
Accordingly, we invoke Assumptions~\ref{asm:regularityLLN}(iii)--(v) in addition to construct a $\delta$-net and extend the pointwise law of large numbers to a uniform law of large numbers under network interference similarly to \citet{Sasaki:25}.

Recall from~\eqref{equ:2stepGMM_end} that our sample objective function is given by 
\[
\widehat{Q}_n(\bm{\beta}) = \widehat{\bm{g}}_n(\bm{\beta})'\widehat{\Xi}_n\widehat{\bm{g}}_n(\bm{\beta}).
\]
Since our model explicitly depends on the network structure, we consider a sequence of objective functions that are $\mathcal{C}_n$-measurable. 
The corresponding population objective function $Q_n$ is defined by
\begin{equation}
    Q_n(\bm{\beta}) = \bm{g}_n(\bm{\beta})'\Xi_n\bm{g}_n(\bm{\beta}),
\end{equation}
where $\bm{g}_n(\bm{\beta}) = E[\bm{g}((W_{n,i},X_{n,i},\bfZ), \bm{\beta}) \mid \mathcal{C}_n]$, and $\Xi_n$ denotes a random matrix that is $\mathcal{C}_n$-measurable, has finite elements, and is positive definite almost surely.

As in standard nonlinear GMM or M estimation frameworks, a crucial requirement is the uniform law of large numbers (ULLN) for the objective function \citep[e.g.,][]{Newey2:94, Vaart:98}. 
Following the conditional framework of KMS, we establish the conditional ULLN for the objective function by leveraging the conditional moment results of \citet{Sasaki:25}. 
Lemma~\ref{app_lemma:QnULLN} in the appendix states this result formally. 
Specifically, for all $\epsilon > 0$,
\begin{equation}
    P\!\left(\left. \sup_{\bm{\beta} \in \Theta} 
    \big| \widehat{Q}_n(\bm{\beta}) - Q_n(\bm{\beta}) \big| 
    > \epsilon \,\right|\, \mathcal{C}_n \right) \to 0 
    \quad \text{a.s.}
\end{equation}
From this almost sure convergence, it follows that 
\[
\sup_{\bm{\beta} \in \Theta} 
\big| \widehat{Q}_n(\bm{\beta}) - Q_n(\bm{\beta}) \big|
\overset{P}{\to} 0.
\]

Define the true parameter vector as the solution to the sequential conditional moment conditions in the sense that
\begin{equation}
    \bm{g}_n(\bm{\beta}) = E\!\left[ g((W_{n,i},X_{n,i},\bfZ), \bm{\beta}) \mid \mathcal{C}_n \right] = 0 
    \quad \text{a.s. for all } n \geq 1 
    \text{ if and only if } \bm{\beta} = \bm{\beta}_0.
\end{equation}
Here, recall that the identification (i.e., the ``only if'' implication) is justified on economic grounds in Section \ref{sec:identification}.
Let $\lVert \cdot \rVert_F$ denote the Frobenius norm, defined by 
$\lVert B \rVert_F = \sqrt{\mathrm{trace}(B'B)}$ for any real matrix $B$. 
The following theorem establishes that the GMM estimator converges in probability to the true parameter value.

\begin{thm}\label{thm:GMM_consistency}
    Suppose that Assumptions \ref{asm:unique_equ}(i) and (iii), \ref{asm:rank_cond}, \ref{asm:2nd_stage}(ii), \ref{asm:simplify_model}, \ref{asm:est}, \ref{asm:regularityLLN} hold, $\underset{n\geq 1}{\sup} E[\lVert \widehat{\Xi}_n \rVert_F|\mathcal{C}_n]<\infty$, and $E[\lVert\widehat{\Xi}_n -\Xi_n\rVert_F |\mathcal{C}_n] \to 0\,\,a.s $, where $\Xi_n$ is $\mathcal{C}_n$-measurable with finite elements and positive definite a.s. Then, for all $\epsilon>0$, we have
    \begin{equation}
        P(\lVert\widehat{\bm{\beta}} -\bm{\beta}_0\rVert_2> \epsilon |\mathcal{C}_n) \to 0\,\, a.s., 
    \end{equation}
    and thus we have
    \begin{equation}
        \widehat{\bm{\beta}} \overset{P}{\to} \bm{\beta}_0.
    \end{equation}
\end{thm}

Next, we are going to establish the asymptotic normality of the GMM estimator. 
By Assumption~\ref{asm:regularityLLN} (iv) and the conditional $\psi$-dependence of $\{ W_{n,i} \}$, Corollary \ref{cor:cg_psidep} in Appendix \ref{sec:proofs:network} shows that
$c'\bm{g}((W_{n,i},X_{n,i},\bfZ), \bm{\beta}_0)$ is also conditionally $\psi$-dependent for any vector 
$c \in \mathbbm{R}^{\mathrm{\dim}(g)}$, with functional $\psi^{cg}$ and coefficients 
$\{ \theta_{n,s}^{cg} \}$, defined in Corollary \ref{cor:cg_psidep}. 
Let 
\[
S_n = \sum_{i \in \mathcal{N}_n} \bm{g}((W_{n,i},X_{n,i},\bfZ), \bm{\beta}_0)
\quad \text{and} \quad 
\Omega_{g,n} = \mathrm{Var}(S_n \mid \mathcal{C}_n).
\]
Given $c \in \mathbbm{R}^{\mathrm{dim}(g)}$, define 
$\sigma_n^2(c) = \mathrm{Var}(c'S_n \mid \mathcal{C}_n) = c'\Omega_{g,n}c$. 
For notational convenience, we also define
\begin{align*}
    \Delta_n(s,m;k) 
    &= \frac{1}{n} \sum_{i} 
    \max_{j \in N_n^{\partial}(i; s)} 
    \big| N_n(i; m) \setminus N_n(j; s-1) \big|^k
    \quad\text{and}
    \\
    c_n(s,m;k) 
    &= \inf_{\alpha > 1} 
    \left[ \Delta_n(s,m;k\alpha) \right]^{\frac{1}{\alpha}} 
    \left[ \delta_n^{\partial}\!\left(s; \frac{\alpha}{1-\alpha}\right) \right]^{1 - \frac{1}{\alpha}}.
\end{align*}
We now impose the following conditions to establish asymptotic normality.

\begin{asm} \label{asm:regularityCLT}
\begin{itemize}
    \item[]
    \item[(i)] For some $p>4$, $\sup_{n\geq 1} \max_{i=1,...,n} \lVert\bm{g}((W_{n,i},X_{n,i},\bfZ),\bm{\beta}_0)\rVert_{\mathcal{C}_n,L^p}<\infty$ a.s.
    \item[(ii)] Given $p>4$ from Assumption 5.2 (i), there exist some $r>p$ and $q>$ such that $1/p+1/q=1/r$ and there exists a positive sequence $m_n\to \infty$ such that for $k=1,2$
    \begin{equation}
    \begin{aligned}
        &\underset{c\in \mathbbm{R}^{\dim(g)}:\lVert c \rVert_q= 1}{\sup} \left|\frac{n}{\sigma_n^{2+k}(c)} \sum_{s\geq 0} c_n(s,m_n;k) (\theta_{n,s}^g)^{1-\frac{2+k}{p}}\right| \to_{a.s.} 0 
        \quad\text{and}\\
        &\underset{c\in \mathbbm{R}^{\dim(g)}:\lVert c \rVert_q= 1}{\sup} \left|\frac{n^2 (\theta_{n,m_n}^{\bm{g}})^{1-(1/p)}}{\sigma_n(c)}\right| \to_{a.s.} 0
        \quad\text{ as } n\to \infty.
    \end{aligned}
    \end{equation}
    \item[(iii)] There is an open set $\mathcal{N}_{\beta_0}\subset\Theta$ such that $\bm{\beta}_0\in\mathcal{N}_{\beta_0}$.
    \item[(iv)] Let $\widehat{G}_n(\bm{\beta})=\nabla_{\bm{\beta}}\widehat{\bm{g}}_n(\bm{\beta})$. There is $G_n(\bm{\beta})$ that is $\mathcal{C}_n$-measurable, of full rank, and continuous at $\bm{\beta}_0$ a.s. such that
    \begin{equation}
        \underset{\bm{\beta}\in\mathcal{N}_{\beta_0}}{\sup} \lVert \widehat{G}_n(\bm{\beta})-G_n(\bm{\beta}) \rVert\overset{P}{\to}0.
    \end{equation}
    \item[(v)] 
    $\Xi_n\overset{P}{\to}\Xi_0$, $n^{-1}\Omega_{g,n}\overset{P}{\to}\Omega_{g,0}$, and $G_n(\bm{\beta}_0) \overset{P}{\to} G_0(\bm{\beta}_0)=G_0$
    where $\Xi_0$, $\Omega_{g,0}$, and $G_0$ are $\mathcal{C}$-measurable. Moreover, $\Xi_0$ and $\Omega_{g,0}$ are positive definite and $G_0$ has full column rank a.s.
\end{itemize}
\end{asm}

By Assumptions~\ref{asm:regularityCLT} (i)–(ii) and the conditional $\psi$-dependence of 
$c'\bm{g}((W_{n,i},X_{n,i},\bfZ), \bm{\beta}_0)$, we extend the conditional central limit theorem in 
Theorem~3.2 of KMS to the multivariate case of $\bm{g}$ via the Cramér–Wold theorem. 
Assumptions~\ref{asm:regularityCLT} (iii)--(v) impose standard regularity conditions for 
nonlinear GMM estimators that are widely adopted in the literature 
\citep[][condition~(v) in Theorem~3.4]{Newey2:94}.

The following theorem establishes the conditional asymptotic normality for $\widehat\beta$.

\begin{thm}\label{thm:GMM_CLT}
   In addition to the assumptions invoked in Theorem \ref{thm:GMM_consistency}, suppose that Assumption \ref{asm:regularityCLT} holds. Then, we have
    \begin{equation*}
        \sqrt{n}(\widehat{\bm{\beta}}-\bm{\beta}) \overset{d}{\to}\Omega_{\beta,0}^{1/2} Z \quad (\mathcal{C}\stably),
    \end{equation*}
    where 
    $
        \Omega_{\beta,0} = (G_0'\Xi_0 G_0)^{-1}G_0'\Xi_0\Omega_{g,0}\Xi_0 G_0 (G_0'\Xi_0 G_0)^{-1}
    $
    is $\mathcal{C}$-measurable, positive definite a.s. and $Z$ follows $N(0,I_{\dim(g)})$, which is independent of $\mathcal{C}$ and thus of $\Omega_{\beta,0}$.
\end{thm}

Conditional convergence in distribution to normal distributions generally implies a non-normal unconditional limiting distribution. At first glance, this seems to preclude standard normal-based inference; nevertheless, valid statistical inference remains feasible.
Let $\widehat{\Omega}_{\beta}$ denote a consistent estimator of the limiting covariance matrix 
$\Omega_{\beta,0}$. 
We consider testing the $k$-dimensional null hypothesis 
$H_0\!: r(\bm{\beta}_0) = 0$ against the alternative 
$H_1\!: r(\bm{\beta}_0) \neq 0$, 
where $r(\cdot)$ is continuously differentiable in a neighborhood of $\bm{\beta}_0$. 
Let 
$R(\bm{\beta}) = \nabla_{\bm{\beta}} r(\bm{\beta})$ 
be a $\mathrm{dim}(g) \times k$ matrix with full column rank, where $k \leq \mathrm{dim}(g)$. 
The Wald statistic is then given by
\begin{equation}
    T_n = n \cdot r(\widehat{\bm{\beta}})' 
    \left( R(\widehat{\bm{\beta}})' 
    \widehat{\Omega}_{\beta} 
    R(\widehat{\bm{\beta}}) \right)^{-1} 
    r(\widehat{\bm{\beta}}).
\end{equation}

To operationalize the result of Theorem \ref{thm:GMM_CLT}, the following theorem establishes the asymptotic validity of the Wald test, even when the variance–covariance matrix is random due to the conditioning on $\mathcal{C}$.

\begin{thm}\label{thm:GMM_Wald}
   Suppose that the assumptions invoked in Theorem \ref{thm:GMM_CLT} hold and there is $\widehat{\Omega}_\beta$ such that $\widehat{\Omega}_\beta\overset{P}{\to}\Omega_{\beta,0}$ and $P(\widehat{\Omega}_\beta \text{ is p.d})=1$. Then, we have
    \begin{equation*}
        \widehat{\Omega}_\beta^{-1/2}\sqrt{n}(\widehat{\bm{\beta}}-\bm{\beta}_0)\overset{d}{\to}Z\sim N(0,I_{\dim(g)}).
    \end{equation*}
    Moreover, 
    \begin{equation*}
        P(T_n>\chi_{k,1-\alpha}^2)\to \alpha
    \end{equation*}
    where $\chi_{k,1-\alpha}^2$ is the $1-\alpha$ quantile of the chi-square distribution with $k$ degrees of freedom.
\end{thm}

This theorem is readily applicable in statistical inference in practice.

\section{Monte Carlo Simulation Studies}\label{sec:simulations}

This section investigates the finite-sample performance of our proposed method through simulation studies.

\subsection{The Model And Data-Generating Process}
We consider two types of network structures in our simulation study. The first is the \textit{ring network}, in which each node is connected to its two immediate neighbors. The second is the \textit{random geometric graph (RGG)}, constructed following \citet{Leung:22ECTA}.

For the RGG model, a network is generated as follows. Let $\xi_i=(\xi_{1i},\xi_{2i})$, $i=1,\ldots,n$, be independent random variables drawn from $\mathrm{Uniform}([0,1]^2)$. Nodes $i$ and $j$ are connected if $\lVert \xi_i - \xi_j \rVert_2 \le r_n$, where $r_n=\left(\frac{\kappa}{\pi n}\right)^{1/2}$ and $\kappa$ denotes the expected degree. To ensure realism, we set $\kappa$ equal to the average degree of the empirical network used in our application, which is approximately $5.63$.

The exposure mapping is specified by
$
T_i(D_i,\mathbf{D}_{-i})=\left(D_i,\mathbbm{1}\left\{\sum_{j\neq i}A_{ij}D_j>0\right\}\right),
$
so that $t\in\mathcal{T}=\{(1,1),(1,0),(0,1),(0,0)\}$. For the treatment decision and outcome models, we largely follow the setup of \citet{Hoshino2:23JoE}, adapted to our framework:
\begin{equation}\label{equ:sim_model}
\begin{aligned}
    Y_i(t) &= \beta_{X0}^{(t)} + \beta_{X1}^{(t)} X_i + \epsilon_i^{(t)}, \\
    u_i(Z_i,\mathbf{D}_{-i},\nu_i) &=
    \begin{cases}
        \beta_{D0} + \beta_{D1} Z_i + \lambda \cdot n_i^{-1} \sum_{j} A_{ij} D_j - \nu_i, & \text{if } D_i = 1, \\
        0, & \text{if } D_i = 0.
    \end{cases}
\end{aligned}
\end{equation}
The covariates and instrumental variables are drawn from standard normal distributions, \( X_i, Z_i \sim N(0,1) \). 
The parameter values for the outcome model are set as \( \beta_Y^{(1,1)} = \beta_Y^{(1,0)} = \beta_Y^{(0,1)} = (2,1) \) and \( \beta_Y^{(0,0)} = (1,2) \), where \( \beta_Y^{(t)} = (\beta_{X0}^{(t)}, \beta_{X1}^{(t)})' \). 
For the selection model, the parameters are set to \( (\beta_{D0}, \beta_{D1}) = (-1, 2) \) and \( \lambda = 1 \).
The unobservable is generated as
\begin{equation}
    \epsilon_i^{(t)} = \beta^{(t)}_{p}\cdot V_i+e_i,
\end{equation}
where $\beta^{(1,1)}_{p}=1.5,\beta^{(1,0)}_{p}=\beta^{(0,1)}_{p}=1,\beta^{(0,0)}_{p}=0.5$, and $e_i\overset{i.i.d}{\sim}N(0,1)$.

\subsection{Estimation and Inference Procedure}
For estimation, we implement our two-step GMM estimator using the HAC variance estimation procedure described between Equations~\eqref{equ:2stepGMM_start} and~\eqref{equ:2stepGMM_end}.\footnote{For the just-identified case, a two-step estimation procedure is not necessary. Nevertheless, we implement the two-step method for the sake of generality in our code, so that it can readily accommodate more general over-identified settings as extensions.} 
A crucial component of this estimator is the choice of the bandwidth and the kernel function. Following KMS, we set the bandwidth as
\begin{equation*}
    b_n = c \times \frac{\log n}{\log(\text{avg.deg} \vee (1 + \epsilon))},
\end{equation*}  
and use the Parzen kernel:
\begin{equation*}
    \omega(z) = 
    \begin{cases}
        1 - 6 z^2 + 6 |z|^3, & 0 \leq |z| \leq \tfrac{1}{2}, \\
        2 (1 - |z|)^3, & \tfrac{1}{2} < |z| \leq 1, \\
        0, & \text{otherwise}.
    \end{cases}
\end{equation*}
Here, \(\text{avg.deg}\) denotes the average degree of nodes in the given network. 
We fix \(\epsilon = 0.05\) and consider \( c \in \{ 1.7, 1.8, 1.9, 2.0, 2.1, 2.2 \} \) as in KMS. 
We report results with \( c = 2.0 \), as the conclusions are qualitatively similar across the other values of \( c \).

For each simulation design, we consider sample sizes $n=250$ and $n=1000$ to assess the convergence rate. In particular, root-$n$ convergence is evaluated by comparing the ratio of the standard deviations of the estimates across these two sample sizes to two. Each experiment consists of 1,000 Monte Carlo replications. 
In each replication, we compute the estimate $\widehat{\bm{\beta}}$, its standard error, and construct the $95\%$ asymptotic confidence interval as
$\widehat{\bm{\beta}} \pm \operatorname{diag}\!\left(\widehat{\Omega}_\beta / n\right)^{1/2}$,
which we use to evaluate coverage probabilities. We additionally compute the corresponding marginal exposure responses (MERs), $\widehat{\overline{\psi}}(t,x,p)$, fixing $x=1$ and considering $p\in\{0.2,0.5,0.8\}$.

\subsection{Simulation Results}

The simulation results for $\bm{\beta}$ under the ring network design are reported in Table~\ref{tab:sim_betas_ring}. We make the following observations. First, the bias when $n=1000$ is generally smaller than when $n=250$, reflecting improved estimation accuracy with larger sample sizes. Second, the standard deviations for $n=1000$ are approximately half the size of those for $n=250$, indicating a $\sqrt{n}$ convergence rate, which is consistent with our theoretical results. Finally, the Monte Carlo coverage frequencies of the $95\%$ confidence intervals are close to the nominal level, particularly when $n=1000$.

Table~\ref{tab:sim_MER_ring} reports the simulation results for the MER estimator $\widehat{\overline{\psi}}$ under the ring network design. Similar to the results for $\bm{\beta}$, the MER estimator also exhibits a $\sqrt{n}$ convergence rate, and its coverage accuracy improves as the sample size increases.

One potential concern the reader may have is that the favorable results reported above may be driven by the simplicity of the ring network structure. Tables~\ref{tab:sim_betas_RGG} and~\ref{tab:sim_MER_RGG} report the results for networks generated by the RGG model. Even though the network structure is now more irregular, the estimator continues to exhibit a $\sqrt{n}$ convergence rate, and our inference procedure achieves coverage probabilities close to the nominal $95\%$ level when $n=1000$.

\begin{table}
\centering
\small
\caption{Simulation Results for $\bm{\beta}$ on a Ring Network}\label{tab:sim_betas_ring}
\begin{tabular}{lcccccccc}
    \toprule
    & \multicolumn{4}{c}{$n=250$} & \multicolumn{4}{c}{$n=1000$}  \\
    \cmidrule(lr){2-5}  \cmidrule(lr){6-9}
    & Bias & SD & RMSE & 95\% Coverage & Bias & SD & RMSE & 95\% Coverage\\
    \midrule
    $\beta_{D0}$         & -0.027 & 0.330 & 0.453 & 0.916 & -0.016 & 0.168 & 0.230 & 0.943 \\
    $\beta_{D1}$         & 0.055  & 0.265 & 0.370 & 0.938 & 0.013  & 0.132 & 0.179 & 0.946 \\
    $\lambda$            & 0.011  & 0.712 & 0.983 & 0.926 & 0.025  & 0.361 & 0.491 & 0.951 \\
    $\beta_{X0}^{(0,0)}$ & -0.023 & 0.736 & 1.032 & 0.917 & -0.007 & 0.386 & 0.540 & 0.932 \\
    $\beta_{X1}^{(0,0)}$ & 0.001  & 0.127 & 0.182 & 0.894 & -0.002 & 0.067 & 0.092 & 0.928 \\
    $\beta_{p}^{(0,0)}$  & 0.044  & 1.180 & 1.665 & 0.914 & 0.015  & 0.621 & 0.866 & 0.932 \\
    $\beta_{X0}^{(1,0)}$ & -0.001 & 0.426 & 0.611 & 0.897 & 0.000  & 0.222 & 0.308 & 0.925 \\
    $\beta_{X1}^{(1,0)}$ & 0.005  & 0.171 & 0.241 & 0.905 & -0.003 & 0.089 & 0.122 & 0.940 \\
    $\beta_{p}^{(1,0)}$  & 0.006  & 1.295 & 1.842 & 0.905 & 0.004  & 0.672 & 0.927 & 0.927 \\
    $\beta_{X0}^{(0,1)}$ & 0.004  & 0.563 & 0.793 & 0.914 & 0.011  & 0.289 & 0.398 & 0.929 \\
    $\beta_{X1}^{(0,1)}$ & 0.001  & 0.102 & 0.143 & 0.924 & -0.001 & 0.053 & 0.072 & 0.942 \\
    $\beta_{p}^{(0,1)}$  & -0.011 & 0.874 & 1.238 & 0.910 & -0.014 & 0.448 & 0.617 & 0.929 \\
    $\beta_{X0}^{(1,1)}$ & 0.004  & 0.324 & 0.450 & 0.919 & 0.002  & 0.167 & 0.231 & 0.934 \\
    $\beta_{X1}^{(1,1)}$ & 0.000  & 0.116 & 0.168 & 0.888 & 0.000  & 0.061 & 0.083 & 0.942 \\
    $\beta_{p}^{(1,1)}$  & -0.015 & 0.928 & 1.291 & 0.918 & 0.001  & 0.478 & 0.659 & 0.941 \\
    \bottomrule
\end{tabular}
\par \smallskip
\parbox{14cm}{\footnotesize Note: We consider the bandwidth and Parzen kernel HAC considered in KMS with constant term and $\epsilon$ in $b_n$ being 2 and 0.05.}
\end{table}

\begin{table}
\small
\caption{\centering Simulation Results for the MER $\overline{\psi}$ on a Ring Network}\label{tab:sim_MER_ring}
\begin{tabular}{cccccccccc}
    \toprule
    & & \multicolumn{4}{c}{$n=250$} & \multicolumn{4}{c}{$n=1000$}  \\
    \cmidrule(lr){3-6}  \cmidrule(lr){7-10}
    & $p$ & Bias & SD & RMSE & 95\% Coverage & Bias & SD & RMSE & 95\% Coverage\\
    \midrule
    \multirow{3}{*}{$MER^{(0,0)}$}
      & .2 & -0.013 & 0.521 & 0.729 & 0.913 & -0.006 & 0.274 & 0.380 & 0.936 \\
      & .5 &  0.000 & 0.224 & 0.315 & 0.911 & -0.001 & 0.117 & 0.163 & 0.930 \\
      & .8 &  0.013 & 0.285 & 0.409 & 0.894 &  0.003 & 0.150 & 0.208 & 0.936 \\\hline
        \multirow{3}{*}{$MER^{(1,0)}$}
      & .2 &  0.005 & 0.279 & 0.398 & 0.904 & -0.002 & 0.144 & 0.201 & 0.942 \\
      & .5 &  0.007 & 0.354 & 0.505 & 0.901 & -0.001 & 0.184 & 0.250 & 0.944 \\
      & .8 &  0.009 & 0.690 & 0.981 & 0.902 &  0.000 & 0.358 & 0.490 & 0.941 \\\hline
        \multirow{3}{*}{$MER^{(0,1)}$}
      & .2 &  0.002 & 0.406 & 0.569 & 0.918 &  0.007 & 0.208 & 0.286 & 0.926 \\
      & .5 & -0.001 & 0.187 & 0.257 & 0.922 &  0.003 & 0.096 & 0.132 & 0.933 \\
      & .8 & -0.005 & 0.207 & 0.289 & 0.907 & -0.001 & 0.106 & 0.145 & 0.940 \\\hline
        \multirow{3}{*}{$MER^{(1,1)}$}
      & .2 &  0.001 & 0.201 & 0.283 & 0.912 &  0.003 & 0.106 & 0.147 & 0.927 \\
      & .5 & -0.004 & 0.236 & 0.327 & 0.921 &  0.003 & 0.121 & 0.166 & 0.947 \\
      & .8 & -0.008 & 0.476 & 0.659 & 0.916 &  0.003 & 0.243 & 0.334 & 0.944 \\
    \bottomrule
\end{tabular}
\par\smallskip
\vspace{0.5pt}
\parbox{14cm}{\footnotesize Note: We consider the bandwidth and Parzen kernel HAC considered in KMS with constant term and $\epsilon$ being 2 and 0.05. The results are based on 1,000 Monte Carlo iterations.}
\end{table}

\begin{table}
\centering
\small
\caption{Simulation Results for $\bm{\beta}$ on a Random Geometric Graph}\label{tab:sim_betas_RGG}
\begin{tabular}{lcccccccc}
    \toprule
    & \multicolumn{4}{c}{$n=250$} & \multicolumn{4}{c}{$n=1000$}  \\
    \cmidrule(lr){2-5}  \cmidrule(lr){6-9}
    & Bias & SD & RMSE & 95\% Coverage & Bias & SD & RMSE & 95\% Coverage\\
    \midrule
      $\beta_{D0}$         & 0.001  & 0.400 & 0.552 & 0.928 & -0.014 & 0.207 & 0.286 & 0.937 \\
      $\beta_{D1}$         & 0.020  & 0.259 & 0.354 & 0.934 & 0.001  & 0.130 & 0.176 & 0.951 \\
      $\lambda$            & 0.008  & 0.897 & 1.253 & 0.916 & 0.013  & 0.471 & 0.655 & 0.934 \\
      $\beta_{X0}^{(0,0)}$ & -0.076 & 1.273 & 1.850 & 0.884 & 0.022  & 0.655 & 0.909 & 0.925 \\
      $\beta_{X1}^{(0,0)}$ & 0.004  & 0.209 & 0.327 & 0.828 & 0.002  & 0.113 & 0.158 & 0.917 \\
      $\beta_{p}^{(0,0)}$  & 0.127  & 2.037 & 2.972 & 0.869 & -0.023 & 1.048 & 1.465 & 0.921 \\
      $\beta_{X0}^{(1,0)}$ & -0.017 & 1.029 & 1.344 & 0.859 & -0.011 & 0.370 & 0.519 & 0.905 \\
      $\beta_{X1}^{(1,0)}$ & -0.019 & 0.315 & 0.481 & 0.794 & -0.005 & 0.150 & 0.211 & 0.917 \\
      $\beta_{p}^{(1,0)}$  & -0.018 & 2.912 & 3.819 & 0.871 & 0.032  & 1.133 & 1.591 & 0.914 \\
      $\beta_{X0}^{(0,1)}$ & -0.021 & 0.480 & 0.670 & 0.923 & 0.011  & 0.250 & 0.337 & 0.953 \\
      $\beta_{X1}^{(0,1)}$ & -0.005 & 0.087 & 0.121 & 0.917 & 0.003  & 0.045 & 0.060 & 0.954 \\
      $\beta_{p}^{(0,1)}$  & 0.029  & 0.747 & 1.040 & 0.921 & -0.017 & 0.390 & 0.526 & 0.942 \\
      $\beta_{X0}^{(1,1)}$ & 0.006  & 0.277 & 0.391 & 0.922 & 0.004  & 0.142 & 0.192 & 0.938 \\
      $\beta_{X1}^{(1,1)}$ & 0.006  & 0.104 & 0.142 & 0.935 & 0.001  & 0.054 & 0.074 & 0.932 \\
      $\beta_{p}^{(1,1)}$  & -0.035 & 0.810 & 1.140 & 0.913 & -0.011 & 0.416 & 0.564 & 0.943 \\
    \bottomrule
\end{tabular}
\par \smallskip
\parbox{14cm}{\footnotesize Note: We set $x=1$ and consider the bandwidth and Parzen kernel HAC considered in KMS with constant term and $\epsilon$ in $b_n$ being 2 and 0.05.  The results are based on 1,000 Monte Carlo iterations.}
\end{table}

\begin{table}
\small
\caption{\centering 1,000 Monte Carlo Simulations for the marginal exposure response ($\overline{\psi}$) on a Random Geometric Graph}\label{tab:sim_MER_RGG}
\begin{tabular}{cccccccccc}
    \toprule
    & & \multicolumn{4}{c}{$n=250$} & \multicolumn{4}{c}{$n=1000$}  \\
    \cmidrule(lr){3-6}  \cmidrule(lr){7-10}
    & $p$ & Bias & SD & RMSE & 95\% Coverage & Bias & SD & RMSE & 95\% Coverage\\
    \midrule
    \multirow{3}{*}{$MER^{(0,0)}$}
        & .2 & -0.047 & 0.900 & 1.317 & 0.882 &  0.020 & 0.464 & 0.641 & 0.931 \\
        & .5 & -0.009 & 0.375 & 0.558 & 0.869 &  0.013 & 0.200 & 0.275 & 0.928 \\
        & .8 &  0.029 & 0.483 & 0.718 & 0.872 &  0.006 & 0.250 & 0.355 & 0.912 \\\hline
        \multirow{3}{*}{$MER^{(1,0)}$}
        & .2 & -0.040 & 0.670 & 0.913 & 0.830 & -0.009 & 0.240 & 0.338 & 0.915 \\
        & .5 & -0.046 & 0.587 & 0.865 & 0.858 &  0.000 & 0.311 & 0.436 & 0.918 \\
        & .8 & -0.051 & 1.377 & 1.873 & 0.870 &  0.010 & 0.606 & 0.850 & 0.915 \\\hline
        \multirow{3}{*}{$MER^{(0,1)}$}
        & .2 & -0.020 & 0.345 & 0.485 & 0.918 &  0.011 & 0.179 & 0.243 & 0.950 \\
        & .5 & -0.011 & 0.158 & 0.223 & 0.911 &  0.006 & 0.081 & 0.110 & 0.942 \\
        & .8 & -0.002 & 0.177 & 0.244 & 0.934 &  0.001 & 0.092 & 0.124 & 0.946 \\\hline
        \multirow{3}{*}{$MER^{(1,1)}$}
        & .2 &  0.005 & 0.177 & 0.248 & 0.929 &  0.003 & 0.090 & 0.124 & 0.938 \\
        & .5 & -0.005 & 0.210 & 0.292 & 0.926 & -0.001 & 0.108 & 0.149 & 0.938 \\
        & .8 & -0.016 & 0.418 & 0.585 & 0.916 & -0.004 & 0.215 & 0.294 & 0.932 \\
    \bottomrule
\end{tabular}
\par\smallskip
\vspace{0.5pt}
\parbox{14cm}{\footnotesize Note: We set $x=1$ and consider the bandwidth and Parzen kernel HAC considered in KMS with constant term and $\epsilon$ being 2 and 0.05.  The results are based on 1,000 Monte Carlo iterations.}
\end{table}

\newpage

\section{Global Economy and Local Labor Markets}\label{sec:application}
\subsection{Background and Objectives}

The impact of the global economy on local economies has long been of great interest to economists. 
In particular, following China's accession to the WTO in 2001, the effect of Chinese exports on U.S.\ local labor markets has received considerable attention. 
A central concern in using observational data for such empirical analysis is the potential endogeneity of local exposure to imports from China. 
In their notable study, \citet{Autor3:13} exploit supply-side factors in international trade to address the endogeneity of local import competition and provide evidence of the effects of such exposure.
Under the assumption that the U.S.\ local labor markets are presumably not directly linked to China's supply-side shocks, their analysis provides convincing empirical evidence.

\citet{Autor3:13} compile a dataset of 722 commuting zones (CZs) in the mainland United States as units of observation representing local labor markets. 
They construct a first-differenced sample from the panel of three decadal periods to capture the growth between 1991–2000 and 2000–2007.\footnote{In their original work, they use import volumes for the years 1991, 2000, and 2007 from the UN Comtrade Database and scale the first-differenced values by 10/9 and 10/7 to obtain decadal figures.} 
The key regressor is the change in Chinese import exposure per worker in a CZ, and the outcome of interest is the decadal change in the manufacturing employment share of the working-age population.
We illustrate the CZs on a map in Figure \ref{fig:CZs}.

\begin{figure}[!t]
\centering
\includegraphics[width=\linewidth]{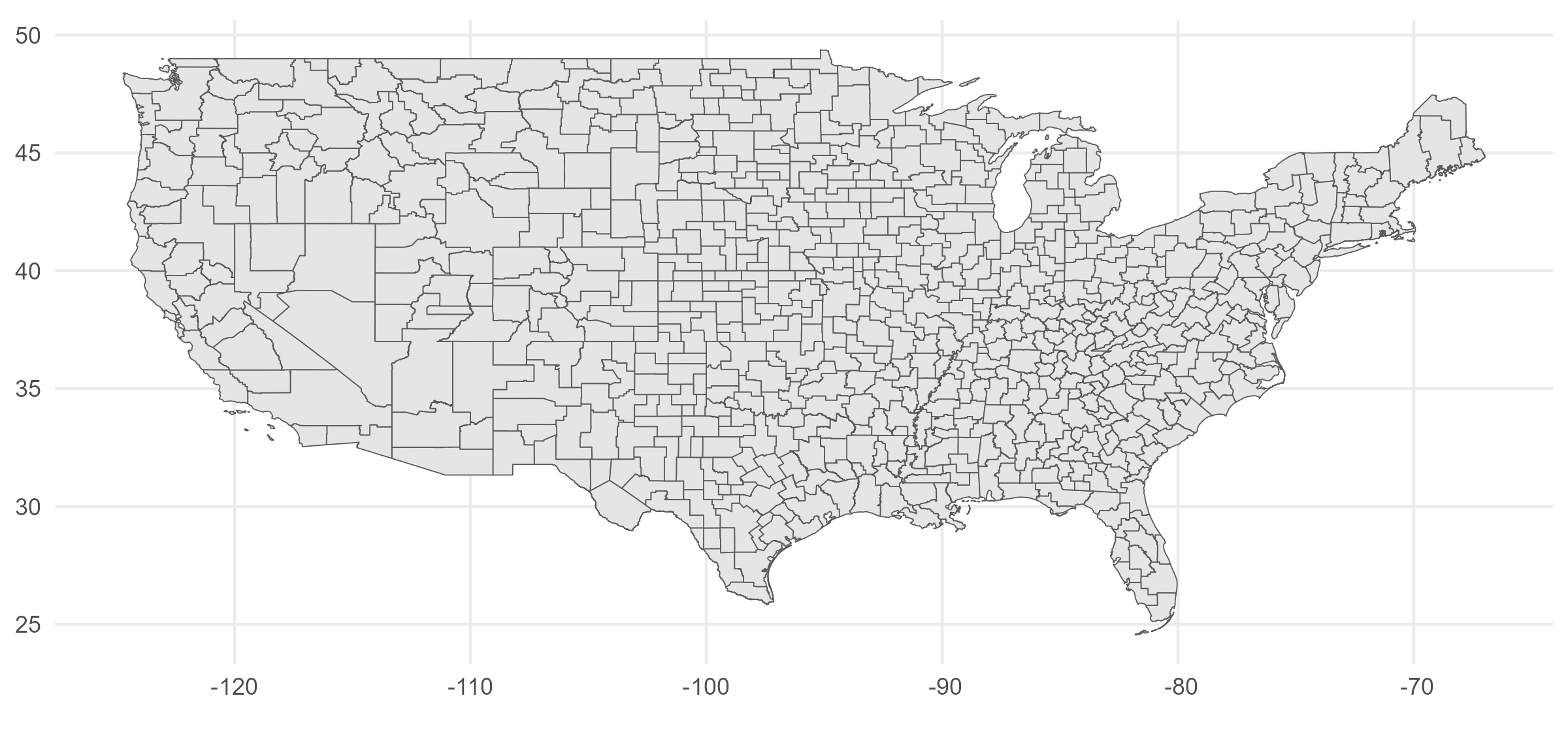}
\caption{Commuting Zones (CZs) in the Mainland United States}\label{fig:CZs}
\end{figure}

While the existing literature assumes that the effects of import competition on local labor markets are confined within each commuting zone (CZ), in reality, spillovers across CZs are likely. Moreover, the extent of local import competition may not only be endogenously determined but also interdependent with that of neighboring CZs, for instance, due to economies of scale in import costs within adjacent geographic areas.
In other words, the degree of exposure is determined as an equilibrium choice under spatial interactions.
To address these features, we employ a model that accommodates spillovers and interference, while also allowing for endogeneity and heterogeneity.

\subsection{The Empirical Model}

Incorporating the network structure, we consider the following specification:
\begin{equation}\label{eq:empirical_specification}
\begin{aligned}
    Y_i(t) &= \beta_{X0} 
        + \mathbbm{1}\{T_i=(1,0)\}\theta_{0} 
        + \mathbbm{1}\{T_i=(1,1)\}\theta_{1} 
        + X_i'\beta_{X} 
        + V_i\beta_p^{(t)} 
        + e_i, \\
    u_i(Z_i,\mathbf{D}_{-i},\nu_i) &= 
        \begin{cases}
            \beta_{D0} + Z_i'\beta_{D} + \lambda \cdot n_i^{-1} \sum_{j } A_{ij} D_j - \nu_i, & \text{if } D_i=1, \\
            0, & \text{if } D_i=0.
        \end{cases}
\end{aligned}
\end{equation}
The outcome variable $Y_i$ denotes the decadal change in the manufacturing employment share of the working-age population in the $i$-th CZ. 
The treatment variable $D_i$ represents the per-worker exposure to the change in Chinese imports in the $i$-th CZ. 
As instruments, following \citet{Autor3:13}, we use the composition and growth of Chinese imports in eight other developed countries. 
The idea is to exploit variation in supply-side factors within China to isolate the corresponding exogenous variation in the treatment variable.
For the control variables $X_i$, we include the period dummy and the share of manufacturing in the start-of-period employment.

The exposure mapping is specified as
\[
T_i(D_i,\mathbf{D}_{-i}) 
    = \left( D_i, \, \mathbbm{1}\!\left\{\sum_{j\neq i} A_{ij} D_j \geq 1 \right\} \right).
\]
As discussed in Example~\ref{examp:exposure}, this exposure mapping allows us to capture heterogeneity in the treatment effect by accounting for whether a unit is exposed to a treated neighbor. 
We focus on the direct MEE while controlling for the treatment status of neighbors as
\begin{equation}\label{equ:MEE_linearInV}
MEE(t',t,x,p) = 
\left\{
\begin{array}{cc}
   \theta_1 + p(\beta_p^{(t')}-\beta_p^{(t)}),  & \text{ for }  t_2'=t_2=1,\\
   \theta_0 + p(\beta_p^{(t')}-\beta_p^{(t)}),  & \text{ for }  t_2'=t_2=0.
\end{array}
\right.
\end{equation}

The empirical model \eqref{eq:empirical_specification} allows the MEE \eqref{equ:MEE_linearInV} to be heterogeneous depending on whether there are treated neighbors. 
Specifically, the difference between $\theta_{1}$ and  $\theta_{0}$ captures the heterogeneity of the direct effect depending on the presence of a treated neighbor. 
Moreover, the model \eqref{eq:empirical_specification} also accommodates unobserved heterogeneity in the MEE \eqref{equ:MEE_linearInV}. 
In particular, the difference $\beta_p^{(t')}-\beta_p^{(t)}$ captures such heterogeneity, which itself may vary depending on whether a neighbor is treated.
While this empirical model differs from the specification in \eqref{equ:sim_model} used for our simulation study, it is designed to capture these empirically relevant features for the current application.

\subsection{The Empirical Results}

\begin{figure}[t]
    \centering
    \includegraphics[width=\linewidth]{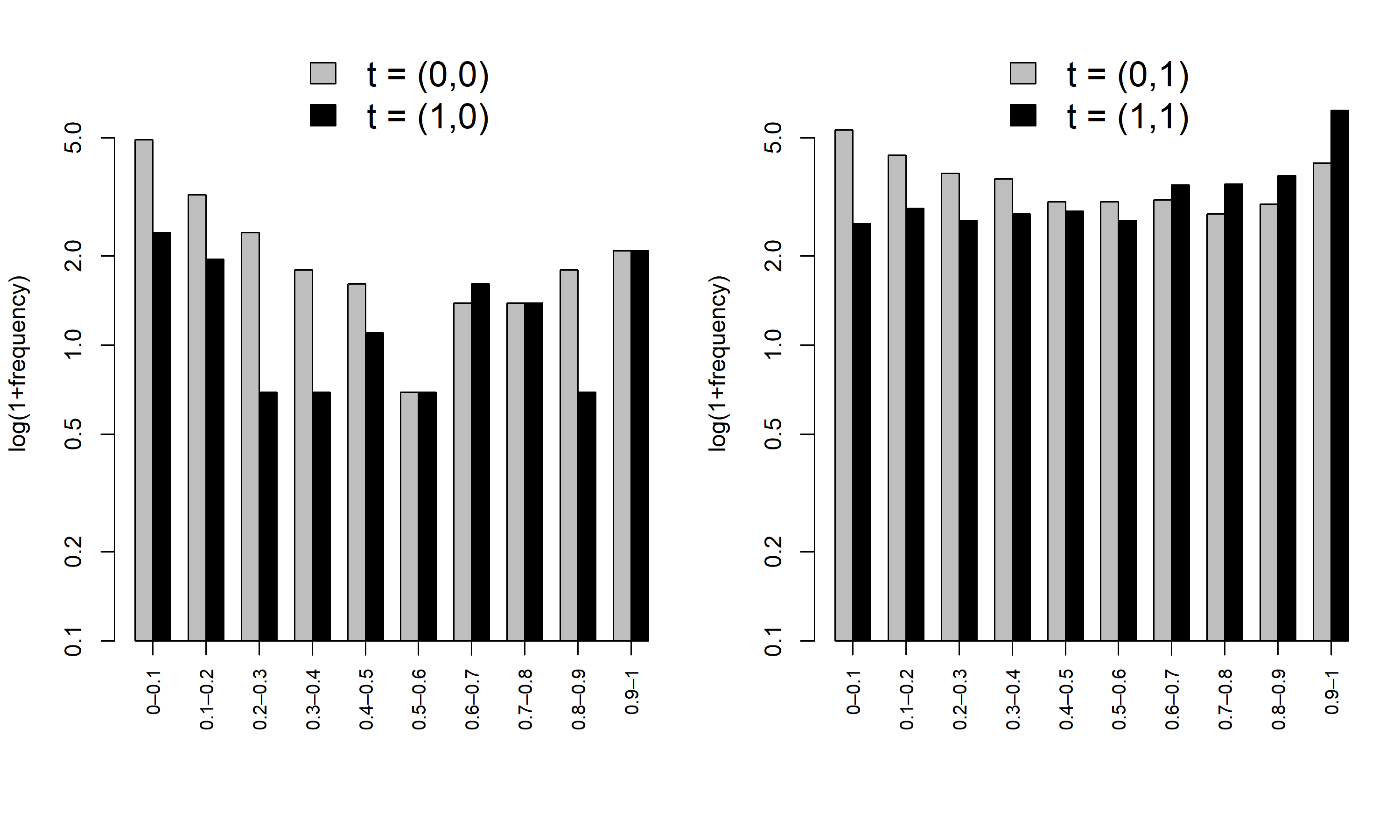}
    \caption{Histograms of the Propensity Scores.}
    \label{fig:MEE_Psupp}
\end{figure}

Before proceeding with the estimation of the MEEs, we first present the routine overlap check. 
Figure~\ref{fig:MEE_Psupp} displays the histograms of the estimated propensity scores. 
Unlike conventional models, however, our setting involves multi-dimensional treatments due to the coexistence of direct and neighbor effects. 
The left panel of the figure shows the histograms for $t=(0,0)$ and $t=(1,0)$, while the right panel shows those for $t=(0,1)$ and $t=(1,1)$. 
In both cases, we observe substantial overlap in this data set.

\begin{figure}[t]
    \centering
    \includegraphics[width=0.95\linewidth]{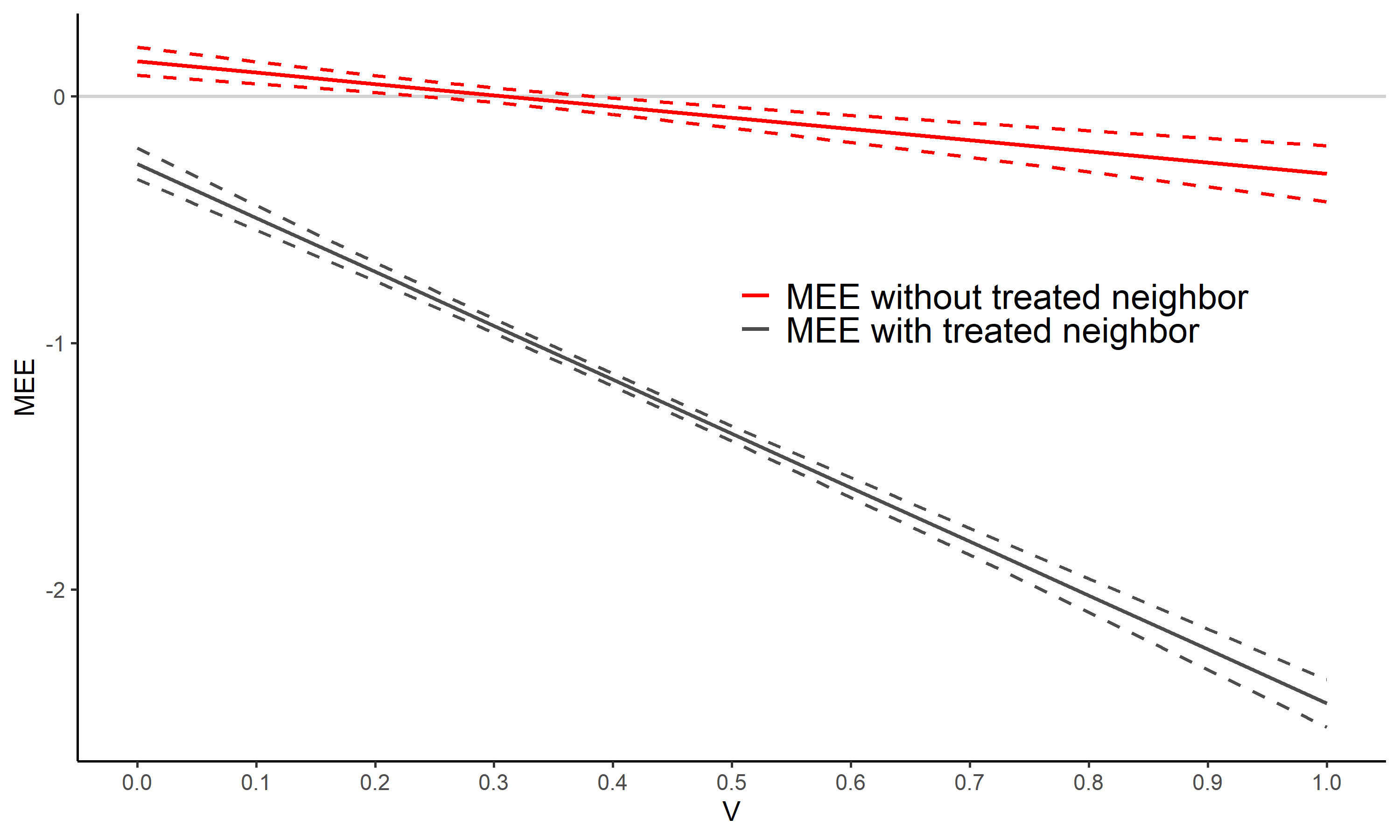}
    \caption{Marginal Exposure Effects with Point-Wise 95\% Confidence Intervals.}
    \label{fig:MEE}
\end{figure}

Figure~\ref{fig:MEE} displays the estimated MEE values along with their pointwise 95\% confidence intervals.
Before discussing the results, it is worth noting that 83.59\% of the CZs have at least one treated neighbor.\footnote{
More specifically, 72.85\% and 94.32\% of the first and second first-differenced samples, respectively, have treated neighbors.  
}
First, in this light, the MEEs are significantly negative for the vast majority of CZs. 
This finding is consistent with earlier results in the literature and aligns with the idea that stronger import competition from China displaces manufacturing workers into other sectors. 
Second, we observe that the MEE without a treated neighbor is always smaller (in absolute value) than the MEE with a treated neighboring zone, suggesting that spillovers amplify the trade shock in local labor markets. 
Finally, the downward-sloping pattern of the MEEs indicates that CZs with higher values of $V$ face more severe trade effects. 
Note that CZs with larger values of $V$ are those less likely to be exposed to import competition. 
It is, therefore, natural that the CZs that suffer more from import competition are also those that tend to \textit{select out} of that treatment status.

\section{Conclusion}\label{sec:conclusion}

This paper develops an econometric framework for identifying and estimating heterogeneous effects of an endogenous treatment in the presence of network interference and spillovers. We characterize endogenous treatment choices as outcomes of a rational‐expectations equilibrium that incorporates spillover considerations, and we provide conditions ensuring that such an equilibrium exists and is unique. Building on this structure, we obtain identification of heterogeneous marginal exposure effects (MEEs), effects that vary with both neighbors’ treatment statuses and with unobserved heterogeneity, and we establish estimation procedures together with their asymptotic properties.

Applying our methods to the impact of import competition on U.S. local labor markets, we document negative MEEs in line with prior findings. Importantly, we show that these negative effects intensify when neighboring regions are also exposed to treatment and for areas that tend to self‐select into lower levels of import competition. These results, which reveal nuanced spillover‐driven heterogeneity, would not have been accessible without the econometric tools developed in this paper.

\clearpage
\newpage

\begin{appendices}
\section{Proofs and Auxiliary Results}
\subsection{Proofs for the Identification Results}\label{sec:proof:identifiadtion}
\begin{proof}[\textbf{Proof of Lemma \ref{lemma:UniBNE}}]
    Fix $n$ and $\bfZ$. Write the equation system \eqref{eq:EQUcp} as $\bfP=\varphi(\bfP)$. Note that logistic distribution assumption in \ref{asm:unique_equ} (i) and $|\lambda|<4$ jointly guarantee that the gradient of the norm of $\varphi$ is sub-unit, as shown by \citet[][Page 405]{Lee3:14}. Thus, we can conclude that $\varphi(\cdot)$ is a contraction mapping, and we have the desired result.
    
\end{proof}

\begin{proof}[\textbf{Proof of Theorem \ref{thm:1stage_iden}}]
    Given network $A$, the equilibrium restriction can be rewritten as
\begin{equation*}
    log(P_i(\bfZ,A))-log(1-P_i(\bfZ,A)) = \phi_i(\bfZ,A)' \cdot \left( \begin{array}{c}
         Z_i'\beta_D  \\
         \lambda
    \end{array} \right),
\end{equation*}
under Assumption \ref{asm:unique_equ}.
Pre-multipling both sides by $\phi_i(\bfZ)$ and then taking the conditional expectation, we obtain
\begin{align}
    E[\phi_i(\bfZ,A) \left( log(P_i(\bfZ,A))-log(1-P_i(\bfZ,A)) \right) |A]  = E[\phi_i(\bfZ,A) \phi_i(\bfZ,A)' |A] \cdot \left( \begin{array}{c}
         \beta_D \\
         \lambda
    \end{array} \right).
\end{align}
Since $P_i(\bfZ,A)$ and $\phi_i(\bfZ,A)$ are both identified, the parameter $\bm{\beta}_1=(\beta_D',\lambda)'$ is identified under the invertibility of the matrix $E[\phi_i(\bfZ,A)\phi_i(\bfZ,A)' |A]$, which is guaranteed by Assumption \ref{asm:rank_cond}.
\end{proof}

\begin{proof}[\textbf{Proof of Theorem \ref{thm:mtr_iden}}]
In this proof, we write short hand notation $\bfP=\bfP(\bfZ,A)$. Recall that we have $D_i=1 \iff P_i\geq V_i$ and $D_i=0 \iff P_i< V_i$, and that for given $N_i$ and treatment status $\bfd_{N_i}$ we have corresponding index sets $\bfl_i=\{j\in N_i|d_j=1\}$ and $\bfl_i'=\{j\in N_i|d_j=0\}$. 

Now, observe that
\begin{equation*}
\begin{aligned}
    E&[\mathbbm{1}\{T_i=t\}Y_i|\bfX=\bfx,\bfP_{N_i}=\bfp_{N_i}]\\
    =&\sum_{(d,\bfd_{N_i^\circ})\in \{0,1 \} \times\{0,1 \}^{|N_i^\circ|}} E[\mathbbm{1}\{T_i=t\}Y_i|\bfX=\bfx,\bfP_{N_i}=\bfp_{N_i},\bfD_{N_i}=\bfd_{N_i}]
    \cdot P(\bfD_{N_i}=\bfd_{N_i}|\bfX=\bfx,\bfP_{N_i}=\bfp_{N_i})\\
    =&\sum_{(d,\bfd_{N_i^\circ}): T_i(d,\bfd_{N_i^\circ},\bfD_{N_i^c})=t} E[\mathbbm{1}\{T_i=t\}Y_i|\bfX=\bfx,\bfP_{N_i}=\bfp_{N_i},\bfD_{N_i}=\bfd_{N_i}]\cdot P(\bfD_{N_i}=\bfd_{N_i}|\bfX=\bfx,\bfP_{N_i}=\bfp_{N_i})\\
    =&\sum_{(d,\bfd_{N_i^\circ}): T_i(d,\bfd_{N_i^\circ},\bfD_{N_i^c})=t} E[Y_i(d,\bfd_{-i})|\bfX=\bfx,\bfP_{N_i}=\bfp_{N_i},\bfV_{\bfl_i} \leq \bfp_{\bfl_i},\bfV_{\bfl_i'} \geq \bfp_{\bfl_i'}]\\
    &\hspace{15mm} \cdot P(\bfV_{\bfl_i} \leq \bfp_{\bfl_i},\bfV_{\bfl_i'} \geq \bfp_{\bfl_i'}|\bfX=\bfx,\bfP_{N_i}=\bfp_{N_i})\\
    =&\sum_{(d,\bfd_{N_i^\circ}): T_i(d,\bfd_{N_i^\circ},\bfD_{N_i^c})=t} E[Y_i(d,\bfd_{-i})|\bfX=\bfx,\bfV_{\bfl_i} \leq \bfp_{\bfl_i},\bfV_{\bfl_i'} \geq \bfp_{\bfl_i'}] \cdot P(\bfV_{\bfl_i} \leq \bfp_{\bfl_i},\bfV_{\bfl_i'} \geq \bfp_{\bfl_i'}|\bfX=\bfx),
\end{aligned}
\end{equation*}
where
the first equality holds by the law of iterated expectation,
the second equality removes the $\bfd$'s where the indicator is zero, and
the last equality follows from Assumption \ref{asm:unique_equ}(ii) that $\bfV$ are independent of $\bfP$ conditional on $\bfX$. 
Notice that
\begin{equation*}
\begin{aligned}
    &E[Y_i(d,\bfd_{-i})|\bfX=\bfx,\bfV_{\bfl_i} \leq \bfp_{\bfl_i},\bfV_{\bfl_i'} \geq \bfp_{\bfl_i'}]\\
    &=\frac{1}{P(\bfV_{\bfl_i} \leq \bfp_{\bfl_i},\bfV_{\bfl_i'} \geq \bfp_{\bfl_i'}|\bfX=\bfx)} \int_{\mathbf{0}}^{\bfp_{\bfl_i}} \int_{\bfp_{\bfl_i'}}^{\mathbf{1}}E[Y_i(d,\bfd_{-i})|\bfX=\bfx,\bfV_{\bfl_i} = \bfv_{\bfl_i},\bfV_{\bfl_i'} = \bfv_{\bfl_i'}] f_{\bfV_{N_i}|\bfX}(\bfv_{N_i}) d\bfv_{N_i}\\
    &=\frac{1}{P(\bfV_{\bfl_i} \leq \bfp_{\bfl_i},\bfV_{\bfl_i'} \geq \bfp_{\bfl_i'}|\bfX=\bfx)} \int_{\mathbf{0}}^{\bfp_{\bfl_i}} \int_{\bfp_{\bfl_i'}}^{\mathbf{1}}E[Y_i(d,\bfd_{-i})|\bfX=\bfx,\bfV_{\bfl_i} = \bfv_{\bfl_i},\bfV_{\bfl_i'} = \bfv_{\bfl_i'}] f_{\bfV_{N_i}}(\bfv_{N_i}) d\bfv_{N_i}\\
    &=\frac{1}{P(\bfV_{\bfl_i} \leq \bfp_{\bfl_i},\bfV_{\bfl_i'} \geq \bfp_{\bfl_i'}|\bfX=\bfx)} \int_{\mathbf{0}}^{\bfp_{\bfl_i}} \int_{\bfp_{\bfl_i'}}^{\mathbf{1}}E[Y_i(d,\bfd_{-i})|\bfX=\bfx,\bfV_{\bfl_i} = \bfv_{\bfl_i},\bfV_{\bfl_i'} = \bfv_{\bfl_i'}] d\bfv_{N_i},
\end{aligned}
\end{equation*}
where the second equality follows that $\bfV$ is independent of $\bfX$,
and the last equality follows from the fact that $\bfV$ are independent of each other and uniformly distributed. To conclude, we have
\begin{equation*}
\begin{aligned}
    E&[\mathbbm{1}\{T_i=t\}Y_i|\bfX=\bfx,\bfP_{N_i}=\bfp_{N_i}]\\
    &=\sum_{(d,\bfd_{N_i^\circ}): T_i(d,\bfd_{N_i^\circ},\bfD_{N_i^c})=t} \int_{\mathbf{0}}^{\bfp_{\bfl_i}} \int_{\bfp_{\bfl_i'}}^{\mathbf{1}}E[Y_i(d,\bfd_{-i})|\bfX=\bfx,\bfV_{\bfl_i} = \bfv_{\bfl_i},\bfV_{\bfl_i'} = \bfv_{\bfl_i'}] d\bfv_{N_i}
\end{aligned}
\end{equation*}

Applying Leibniz rule and using the fact that all the remaining $\bfd_{N_i}$ in the summation results in the identical value of the exposure mapping and thus the identical expected value, we have
\begin{equation*}
\begin{aligned}
    \partial_{\bfp_{N_i}} E[\mathbbm{1}\{T_i=t\}Y_i|\bfX=\bfx,\bfP_{N_i}=\bfp_{N_i}]&=\sum_{\bfd_{N_i}: T_i(\bfd_{N_i},\bfD_{N^c_i})=t} (-1)^{|\bfl'|} MTR(d,\bfd_{-i},\bfx,\bfp_{N_i})\\
    &=MTR(d,\bfd_{-i},\bfx,\bfp_{N_i})\sum_{\bfd_{N_i}: T_i(\bfd_{N_i},\bfD_{N^c_i})=t} (-1)^{|\bfl'|},
\end{aligned}
\end{equation*}
where the equality use the fact that $MTR$ remains the same whenever we fix $T_i=t$ and $N_i$. Rearranging the terms gives
\begin{equation*}
    \frac{\partial_{\bfp_{N_i}} E[\mathbbm{1}\{T_i=t\}Y_i|\bfX=\bfx,\bfP_{N_i}=\bfp_{N_i}]}{\sum_{\bfd_{N_i}: T_i(\bfd_{N_i},\bfD_{N^c_i})=t} (-1)^{|\bfl'|}}=MTR(d,\bfd_{-i},\bfx,\bfp_{N_i}),
\end{equation*}
as claimed in the theorem.
\end{proof}

\begin{proof}[\textbf{Proof of Corollary \ref{cor:simple_mtr_iden}}]
    In this proof, we write short hand notation $\bfP=\bfP(\bfZ,A)$. Letting $\Tilde{\bfl}_i'=\{j\in N_i^\circ|d_j=0\}$, we have
\begin{align*}
    &\overline{MER}_i(\bfd,\bfx,p_i) \\
    &= E[Y(\bfd)|\bfX=\bfx,V_i=p_i]\\
    &= E[E[Y(\bfd)|\bfX=\bfx,V_i=p_i,\bfV_{N_i^\circ}]|\bfX=\bfx,V_i=p_i]\\
    &= \int_0^1\cdots\int_0^1 E[Y(\bfd)|\bfX=\bfx,V_i=p_i,\bfV_{N_i^\circ}=\bfv_{N_i^\circ}]f_{\bfv_{N_i^\circ}|\bfX}(\bfv_{N_i^\circ}) d\bfv_{N_i^\circ} \\
    &= \int_0^1\cdots\int_0^1  E[Y(\bfd)|\bfX=\bfx,V_i=p_i,\bfV_{N_i^\circ}=\bfv_{N_i^\circ}] d\bfv_{N_i^\circ}\\
    &= \int_0^1\cdots\int_0^1 \psi(T_i(\bfd),\bfx,p_i,\bfv_{N_i^\circ}) d\bfv_{N_i^\circ}\\
    &= \int_0^1\cdots\int_0^1 \frac{\partial_{p_i}\partial_{\bfv_{N_i^\circ}} E[\mathbbm{1}\{T_i=T_i(\bfd)\}Y_i|\bfX=\bfx,P_i=p_i,\bfP_{N_i^\circ}=\bfv_{N_i^\circ}]}{\sum_{\Tilde{\bfd}_{N_i}: T_i(\Tilde{\bfd}_{N_i},\bfD_{N_i^c})=T_i(\bfd)} (-1)^{|\bfl_i'|}}d\bfv_{N_i^\circ}\\
    &= \frac{1}{\sum_{\Tilde{\bfd}_{N_i}: T_i(\Tilde{\bfd}_{N_i},\bfD_{N_i^c})=T_i(\bfd)} (-1)^{|\bfl_i'|}}  \\
    &\hspace{1cm}\times \int_0^1\cdots\int_0^1 \partial_{p_i}\partial_{\bfv_{N_i^\circ}} E[\mathbbm{1}\{T_i=T_i(\bfd)\}Y_i|\bfX=\bfx,P_i=p_i,\bfP_{N_i^\circ}=\bfv_{N_i^\circ}]  d\bfv_{N_i^\circ}\\
    &=\frac{1}{\sum_{\Tilde{\bfd}_{N_i}: T_i(\Tilde{\bfd}_{N_i},\bfD_{N_i^c})=T_i(\bfd)} (-1)^{|\bfl_i'|}}  \\
    &\hspace{1cm}\times \partial_{p_i} \left[ \sum_{\Tilde{\bfd}_{N_i^\circ}\in \{0,1\}^{|N_i^\circ |}} (-1)^{|\Tilde{\bfl}_i'|} E[\mathbbm{1}\{T_i=T_i(\bfd)\}Y_i|\bfX=\bfx,P_i=p_i,\bfP_{N_i^\circ}=\Tilde{\bfd}_{N_i^\circ}]\right]\\
    &=\frac{1}{\sum_{\Tilde{\bfd}_{N_i}: T_i(\Tilde{\bfd}_{N_i},\bfD_{N_i^c})=T_i(\bfd)} (-1)^{|\bfl_i'|}} \\
    &\hspace{1cm}\times \partial_{p_i} \left[ \sum_{\Tilde{\bfd}_{N_i^\circ}: T_i(\Tilde{\bfd}_{N_i},\bfD_{N_i^c})=T_i(\bfd)} (-1)^{|\Tilde{\bfl}_i'|} E[\mathbbm{1}\{T_i=T_i(\bfd)\}Y_i|\bfX=\bfx,P_i=p_i,\bfP_{N_i^\circ}=\Tilde{\bfd}_{N_i^\circ}]\right],
\end{align*}
where
the fourth equality follows from the fact that $\bfV\perp\bfX$ and $\bfV$ are independent and uniformly distributed on $[0,1]$ with density equal to 1,
the fifth equality follows from Theorem \ref{thm:mtr_iden}, and
the eighth equality is due to the fundamental theorem of calculus.
Notice that we have an indicator function in the expectation. 
This allows us to drop those $\Tilde{\bfd}_{N_i^\circ}$ such that $T(d_i,\Tilde{\bfd}_{N_i^\circ},\bfD_{N_i^c})$ is not equal to $t$. 
This completes the first statement of the corollary.

For the second statement, consider
\begin{equation*}
\begin{aligned}
    E[&\mathbbm{1}\{T_i=t\}Y_i|\bfX,P_i,\bfP_{N_i^\circ}] \\
    &= P(\mathbbm{1}\{T_i=t\}=1|\bfX,P_i,\bfP_{N_i^\circ}) \times E[\mathbbm{1}\{T_i=t\}Y_i|\mathbbm{1}\{T_i=t\}=1,\bfX,P_i,\bfP_{N_i^\circ}].
\end{aligned}
\end{equation*}
The first factor on the right-hand side is
\begin{equation*}
\begin{aligned}
    P(\mathbbm{1}\{T_i=t\}=1|\bfX,P_i,\bfP_{N_i^\circ})=
    \left\{
    \begin{array}{cc}
     P_i\left( \sum_{\bfd_{N_i^\circ}:T_i(d_i,\bfd_{N_i^\circ},\bfD_{N_i^c})=t} \bfP_{\Tilde{\bfl}_i}(1-\bfP_{\Tilde{\bfl}_i'}) \right)  &  \text{if } t_1=1, \\
     (1-P_i)\left( \sum_{\bfd_{N_i^\circ}:T_i(d_i,\bfd_{N_i^\circ},\bfD_{N_i^c})=t} \bfP_{\Tilde{\bfl}_i}(1-\bfP_{\Tilde{\bfl}_i'}) \right)  & \text{if } t_1=0.
    \end{array}\right.
\end{aligned}
\end{equation*}
Under Assumption \ref{asm:simplify_model} (i), we can exchange the values of $d_{N_i^\circ}$. 
Thus, the summation in the parentheses equals to 1, and the whole term simplifies to $P_i$ if $t_1=1$, and $(1-P_i)$ otherwise. 

For the second factor, Assumption \ref{asm:simplify_model} (ii) yields
\begin{equation*}
\begin{aligned}
    E[&\mathbbm{1}\{T_i=t\}Y_i|\mathbbm{1}\{T_i=t\},\bfX,P_i,\bfP_{N_i^\circ}]\\
    &=E[\mathbbm{1}\{T_i=t\}( (\mu_X^{(t)}(X_i) +\mu_\epsilon^{(t)}(V_i) +e_i^{(t)})|\mathbbm{1}\{T_i=t\}=1,\bfX,P_i,\bfP_{N_i^\circ}]\\
    &=\mu_X^{(t)}(X_i) + E[\mu_\epsilon^{(t)}(V_i)|\bfX,\bigcup_{\bfd_{N_i}:T_i(d_i,\bfd_{N_i^\circ},\bfD_{N_i^c})=t}\{ \bfV_{\bfl_i}\leq \bfP_{\bfl_i},\bfV_{\bfl_i'}> \bfP_{\bfl_i'} \} ]  \\
    &\hspace{5mm} + E[e_i^{(t)}|\bfX,\bigcup_{\bfd_{N_i}:T_i(d_i,\bfd_{N_i^\circ},\bfD_{N_i^c})=t}\{ \bfV_{\bfl_i}\leq \bfP_{\bfl_i},\bfV_{\bfl_i'}> \bfP_{\bfl_i'} \} ] \\
    & = \left\{\begin{array}{cc}
        \mu_X^{(t)}(X_i) + E[\mu_\epsilon^{(t)}(V_i)|\bfX, V_i\leq P_i] & \text{ if } t_1=1,  \\
        \mu_X^{(t)}(X_i) + E[\mu_\epsilon^{(t)}(V_i)|\bfX, V_i > P_i] & \text{ if } t_1=0,
    \end{array} \right.
\end{aligned}
\end{equation*}
where the second equality follows from that $\bfP$ are independent of $\bfV$ conditional on $\bfX$, and last equation follows from that $E[e_i^{(t)}|\bfX,\bfV_{N_i}]=0$ and that $\bfV$ are independent conditional on $\bfX$. 

By the above calculations, $E[\mathbbm{1}\{T_i=t\}Y_i|\bfX,P_i,\bfP_{N_i^\circ}]$ is identical across different value of $\bfP_{N_i^\circ}$. 
Thus, for $t_1=1$, we have
\begin{equation*}
\begin{aligned}
    \overline{\psi}(t,\bfx,p,N_i)&=\frac{\partial_{p}\left[ \sum_{\bfd_{N_i^\circ}: T_i(\bfd_{N_i})=t} (-1)^{|\Tilde{\bfl_i}'|} E[\mathbbm{1}\{T_i=t\}Y_i|\bfX=\bfx,P_i=p,\bfP_{N_i^\circ}=\bfd_{N_i^\circ}]\right]}{\sum_{\bfd_{N_i}: T_i(\bfd_{N_i},\bfD_{N_i^c})=t} (-1)^{|\bfl_i'|}} \\
    &= \partial_{p} \left[ p\cdot (\mu_X^{(t)}(x_i) + E[\mu_\epsilon^{(t)}(V_i)|\bfX=\bfx, V_i\leq p])\right] \frac{\sum_{\bfd_{N_i^\circ}: T_i(\bfd_{N_i},\bfD_{N_i^c})=t} (-1)^{|\Tilde{\bfl}_i'|}}{\sum_{\bfd_{N_i}: T_i(\bfd_{N_i},\bfD_{N_i^c})=t} (-1)^{|\bfl_i'|}}\\
    &= \partial_{p} \left[ p\cdot (\mu_X^{(t)}(x_i) + E[\mu_\epsilon^{(t)}(V_i)|\bfX=\bfx, V_i\leq p])\right],
\end{aligned} 
\end{equation*}
where the last equation hold because the fraction simply cancel out when $t_1=1$. Likewise, when $t_1=0$, the fraction will be $-1$, and 
\begin{equation*}
    \overline{\psi}(t,\bfx,p,N_i) = \partial_{p} \left[ (p-1)\cdot (\mu_X^{(t)}(x_i) + E[\mu_\epsilon^{(t)}(V_i)|\bfX=\bfx, V_i> p])\right],
\end{equation*}
where we have $p-1$ instead of $1-p$ due to the fraction of the two summations being $-1$. In summary, we obtain 
\begin{equation*}
    \overline{\psi}(t,\bfx,p_i,N_i)=\left\{\begin{array}{cc}
        \partial_{p} \left[ p\cdot  (\mu_X^{(t)}(x_i) + E[\mu_\epsilon^{(t)}(V_i)|\bfX=\bfx, V_i\leq p] )\right] & \text{ if } t_1=1,  \\
        \partial_{p} \left[ (p-1)\cdot  (\mu_X^{(t)}(x_i) + E[\mu_\epsilon^{(t)}(V_i)|\bfX=\bfx, V_i> p] )\right] & \text{ if } t_1=0,
    \end{array} \right.
\end{equation*}
as claimed in the second statement of the corollary.
\end{proof}

\subsection{Proofs for the Network Properties}\label{sec:proofs:network}
\begin{proof}[\textbf{Proof of Proposition \ref{prop:D_psiDep}}]
    The proof of this theorem follows Proposition 2.3 in KMS. For completeness, we still provide the proof here. Let $(H,H')\in \mathcal{P}_n(h,h',2s+1)$. Define $\xi = f(\bfD_{n,H})$, $\zeta=f'(\bfD_{n,H'})$, and
\begin{equation*}
    \xi^{(s)}=f(D_{n,i}^{(s)}:i\in H),\hspace{2mm} \zeta^{(s)}=f(D_{n,i}^{(s)}:i\in H').
\end{equation*}
$\xi^{(s)}$ and $\zeta^{(s)}$ are independent conditional on $\mathcal{C}_n$. Now, compute the covariance
\begin{equation*}
    \left| \Cov(\xi,\zeta|\mathcal{C}_n) \right| \leq \left| \Cov(\xi-\xi^{(s)},\zeta|\mathcal{C}_n) \right| + \left| \Cov(\xi^{(s)},\zeta-\zeta^{(s)}|\mathcal{C}_n) \right| + \left| \Cov(\xi^{(s)},\zeta^{(s)}|\mathcal{C}_n) \right|\hspace{2mm} a.s.
\end{equation*}
Note that the last term is $0$ since $\xi^{(s)}$ and $\zeta^{(s)}$ have non-overlapping $\nu$.
The first two terms are bounded as
\begin{equation*}
\begin{aligned}
    &\left| \Cov(\xi-\xi^{(s)},\zeta|\mathcal{C}_n) \right| + \left| \Cov(\xi^{(s)},\zeta-\zeta^{(s)}|\mathcal{C}_n) \right|\\
    &\leq 2 \lVert f' \rVert_{\infty} E[|\xi - \xi^{(s)}|\mathcal{C}_n] + 2 \lVert f \rVert_{\infty} E[|\zeta - \zeta^{(s)}|\mathcal{C}_n]\\
    &\leq 2 (h \lVert f' \rVert_{\infty} \mathrm{Lip}(f) + h' \lVert f \rVert_{\infty} \mathrm{Lip}(f'))\times \underset{i}{\max}\, E[|D_{n,i}-D_{n,i}^{(s)}|\mathcal{C}_n].
\end{aligned}
\end{equation*}
Letting $\psi_{h,h'}^D(f,f')= (h \lVert f' \rVert_{\infty} \mathrm{Lip}(f) + h' \lVert f \rVert_{\infty} \mathrm{Lip}(f'))$ and $\theta_{n,s}^D=2E[|D_{n,i}-D_{n,i}^{(s)}|\mathcal{C}_n]$ gives the result.
\end{proof}

Before proving Proposition \ref{prop:Y_psiDep}, we first show that we can use $\psi_{h,h'}^D$ (with modified domain) and the $\theta^D_{n,s}$ coefficient as in the Proposition \ref{prop:D_psiDep} to characterize the vector of stacking $\psi$-dependent variables with independent shocks. Let $W_{n,i}^D=(D_{n,i},\nu_{n,i},(e_{n,i}^{(t)})_{t\in\mathcal{T}})$ with dimension $dim(W^D)=2+|\mathcal{T}|$, and define
\begin{equation*}
    W_{n,i}^{D,(s)} = (D_{n,i}^{(s)},\nu_{n,i},(e_{n,i}^{(t)})_{t\in\mathcal{T}}).
\end{equation*}

\begin{lemma}\label{app_lemma:WD_psiDep}
    Let $\{W_{n,i}^D\}$ be the vector described above. Then, for any $H,H'\in \mathcal{P}(h,h',2s+1)$ and $f\in \mathcal{L}_{dim(W^D),h}, f'\in \mathcal{L}_{dim(W^D),h'}$,
    \begin{equation*}
        \left| \Cov(f(\bfD_{n,H}),f'(\bfD_{n,H'})|\mathcal{C}_n) \right| \leq \psi_{h,h}^D(f,f') \theta_{n,s}^D\hspace{2mm} a.s.,
    \end{equation*}
    where $\psi_{h,h}^D(f,f')=h\lVert f' \rVert_{\infty} \Lip(f) + h'\lVert f \rVert_{\infty} \Lip(f')$ and $\theta_{n,s}^D= 2\, \underset{i}{max}\, E[|D_{n,i}-D_{n,i}^{(s)}|\mathcal{C}_n]$.
\end{lemma}
\begin{proof}[\textbf{Proof of Lemma \ref{app_lemma:WD_psiDep}}]
    Let $(H,H')\in \mathcal{P}_n(h,h',2s+1)$. Define $\xi = f(\bfW_{n,H}^D)$, $\zeta=f'(\bfW_{n,H'}^D)$, and
\begin{equation*}
    \xi^{(s)}=f(W_{n,i}^{D,(s)}:i\in H),\hspace{2mm} \zeta^{(s)}=f(W_{n,i}^{D,(s)}:i\in H').
\end{equation*}
Since $\bfW_H^D$ and $\bfW_{H'}^D$ are independent conditional on $\mathcal{C}_n$, we have 
\begin{equation*}
    \left| \Cov(\xi,\zeta|\mathcal{C}_n) \right| \leq \left| \Cov(\xi-\xi^{(s)},\zeta|\mathcal{C}_n) \right| + \left| \Cov(\xi^{(s)},\zeta-\zeta^{(s)}|\mathcal{C}_n) \right|\hspace{2mm} a.s.
\end{equation*}
The right-hand side is bounded as
\begin{equation*}
\begin{aligned}
    &\left| \Cov(\xi-\xi^{(s)},\zeta|\mathcal{C}_n) \right| + \left| \Cov(\xi^{(s)},\zeta-\zeta^{(s)}|\mathcal{C}_n) \right|\\
    &\leq 2 \lVert f' \rVert_{\infty} E[|\xi - \xi^{(s)}|\mathcal{C}_n] + 2 \lVert f \rVert_{\infty} E[|\zeta - \zeta^{(s)}|\mathcal{C}_n]\\
    &\leq 2 (h \lVert f' \rVert_{\infty} \Lip(f) + h' \lVert f \rVert_{\infty} \Lip(f'))\times \underset{i}{\max}\, E[\lVert W_{n,i}^D-W_{n,i}^{D,(s)}\rVert_2 \mathcal{C}_n].
\end{aligned}
\end{equation*}
By the $L_1$- and $L_2$-norm inequality, we have 
\begin{equation*}
\begin{aligned}
    \lVert W_{n,i}^D-W_{n,i}^{D,(s)}\rVert_2 &\leq |D_{n,i}-D_{n,i}^{(s)}|+|\nu_{n,i}-\nu_{n,i}|+\sum_{t\in \mathcal{T}}|e_{n,i}^{(t)}-e_{n,i}^{(t)}|\\
    &=|D_{n,i}-D_{n,i}^{(s)}|.
\end{aligned}
\end{equation*}
Thus, we can use the same functional $\psi_{h,h'}^D$ and sequence $\theta_{n,s}^D$ as in the proof of Proposition \ref{prop:D_psiDep} to characterize the $\psi$-dependence of $\{W_{n,i}^D\}$.
\end{proof}

With this lemma in hand, we are now ready to prove Proposition \ref{prop:Y_psiDep}.

\begin{proof}[\textbf{Proof of Proposition \ref{prop:Y_psiDep}}]
    Let $(H,H')\in \mathcal{P}_n(h,h',4s+1)$. Define $\xi = f(\bfY_{n,H})$, $\zeta=f'(\bfY_{n,H'})$, and
\begin{equation*}
    \xi^{(s)}=f(Y_{n,i}^{(s)}:i\in H),\hspace{2mm} \zeta^{(s)}=f(Y_{n,i}^{(s)}:i\in H').
\end{equation*}
Now decompose the covariance as
\begin{equation*}
    \left| \Cov(\xi,\zeta|\mathcal{C}_n) \right| \leq \left| \Cov(\xi-\xi^{(s)},\zeta|\mathcal{C}_n) \right| + \left| \Cov(\xi^{(s)},\zeta-\zeta^{(s)}|\mathcal{C}_n) \right| + \left| \Cov(\xi^{(s)},\zeta^{(s)}|\mathcal{C}_n) \right|\hspace{2mm} a.s.
\end{equation*}
Note that $\xi^{(s)}$ and $\zeta^{(s)}$ here are actually dependent conditional on $\mathcal{C}_n$, unlike Proposition \ref{prop:D_psiDep} and Lemma \ref{app_lemma:WD_psiDep}. Thus, we cannot ignore the term. 
For the first two terms, we use the same argument as in the proof of proposition~\ref{prop:D_psiDep} and Lemma~\ref{app_lemma:WD_psiDep} and write
\begin{equation*}
\begin{aligned}
    &\left| \Cov(\xi-\xi^{(s)},\zeta|\mathcal{C}_n) \right| + \left| \Cov(\xi^{(s)},\zeta-\zeta^{(s)}|\mathcal{C}_n) \right|\\
    &\leq 2 \lVert f' \rVert_{\infty} E[|\xi - \xi^{(s)}|\mathcal{C}_n] + 2 \lVert f \rVert_{\infty} E[|\zeta - \zeta^{(s)}|\mathcal{C}_n]\\
    &\leq 2 (h \lVert f' \rVert_{\infty} \Lip(f) + h' \lVert f \rVert_{\infty} \Lip(f'))\times \underset{i}{\max}\, E[|Y_{n,i}-Y_{n,i}^{(s)}|\mathcal{C}_n].
\end{aligned}
\end{equation*}
For the last term, $\left| \Cov(\xi^{(s)},\zeta^{(s)}|\mathcal{C}_n) \right|$, we first define the index set $N_{A_n}(H,s)=\bigcup_{i\in H} N_{A_n}(i,s)$, and similarly for $N_{A_n}(H',s)$. To further simplify the notation, let $W_{n,i}^D=(D_{n,i},\nu_{n,i},(e_{n,i}^{(t)})_{t\in\mathcal{T}})$. 
Note that when $H,H'$ are at least $4s+1$ apart in network, $N_{A_n}(H,s)$ and $N_{A_n}(H',s)$ are at lest $2s+1$ apart. Thus, we can use Lemma \ref{app_lemma:WD_psiDep} to conclude
\begin{equation}\label{app_equ:W_psiDep}
    \left| \Cov(f(\bfW_{N_{A_n}(H,s)}^D),f'(\bfW_{N_{A_n}(H',s)}^D)|\mathcal{C}_n) \right| \leq \psi_{|N_{A_n}(H,s)|,|N_{A_n}(H',s)|}^D(f,f') \theta_{n,s}^D\hspace{2mm} a.s.
\end{equation}
for $f\in\mathcal{L}_{dim(W^D),|N_{A_n}(H,s)|}$ and 
$f'\in\mathcal{L}_{dim(W^D),|N_{A_n}(H',s)|}$.
Now, we define $f_1$ and $f_1'$ by
\begin{equation*}
\begin{aligned}
    f_1(\bfW_{N_{A_n}(H,s)}^D)&=f( (\sigma_{n,i}^Y(\bfW^{D,(s,i)}))_{i\in H}  ) \quad\text{and}\\
    f_1'(\bfW_{N_{A_n}(H',s)}^D)&=f'( (\sigma_{n,i}^Y(\bfW^{D,(s,i)}))_{i\in H'}  ),
\end{aligned}
\end{equation*}
respectively.
By \eqref{app_equ:W_psiDep}, we have
\begin{equation*}
\begin{aligned}
    \left| \Cov(f(\bfY_{n,H}^{(s)}),f'(\bfY_{n,H'}^{(s)}) \right| &\leq \psi_{|N_{A_n}(H,s)|,|N_{A_n}(H',s)|}^D(f_1,f_1')\theta_{n,s}^D\\
    &\leq \left( |N_{A_n}(H,s)|\lVert f'_1 \rVert_{\infty} \Lip(f_1) + |N_{A_n}(H',s)|\lVert f_1 \rVert_{\infty} \Lip(f'_1) \right) \theta_{n,s}^D
\end{aligned}
\end{equation*}
Note that $\lVert f_1 \rVert_{\infty}=\lVert f \rVert_{\infty}$, $\lVert f'_1 \rVert_{\infty}=\lVert f' \rVert_{\infty}$, and $|N_{A_n}(H,s)| \leq \overline{N}_{A_n}(s)\times h$, $|N_{A_n}(H',s)| \leq \overline{N}_{A_n}(s)\times h'$.

Now, it remains to compute $Lip(f_1)$ and $Lip(f_1')$. 
From the Lipschitz property of $f$, given $\bfW_n^{D,(s,i)}$ and $\widetilde{\bfW}_n^{D,(s,i)}$ being two points in $\left([0,1]\times \mathbbm{R}^{ (1+|\mathcal{T}|)}\right)^n$, we have
\begin{equation*}
\begin{aligned}
    &\hspace{-5mm}\left| f \left( (\sigma_{n,i}^Y(\bfW_n^{D,(s,i)}))_{i\in H}  \right) - f\left( (\sigma_{n,i}^Y(\widetilde{\bfW}_n^{D,(s,i)}))_{i\in H}  \right) \right| \\
    &\hspace{2mm} \leq \Lip(f)\, d_h \left( (\sigma_{n,i}^Y(\bfW_n^{D,(s,i)}))_{i\in H}, (\sigma_{n,i}^Y(\widetilde{\bfW}_n^{D,(s,i)}))_{i\in H} \right).
\end{aligned}
\end{equation*}
The right-hand side of the inequality equals to
\begin{equation*}
\begin{aligned}
    \Lip(f)& \sum_{i\in H} \sqrt{\left(\sigma_{n,i}^Y (\bfW_n^{D,(s,i)}) - \sigma_{n,i}^Y(\widetilde{\bfW}_n^{D,(s,i)}) \right)^2 }\\
    &\leq \Lip(f) \sum_{i\in H}\, \left| \sigma_{n,i}^Y (\bfW_n^{D,(s,i)}) - \sigma_{n,i}^Y(\widetilde{\bfW}_n^{D,(s,i)}) \right| \\
    &\leq \Lip(f) \times \underset{i\in\mathcal{N}_n}{\max}\, \Lip(\sigma_{n,i}^Y) \times \sum_{i\in \mathcal{N}_n}\lVert \bfW_n^{D,(s,i)} - \widetilde{\bfW}_n^{D,(s,i)} \rVert_2\\
    & = \Lip(f) \times \underset{i\in\mathcal{N}_n}{\max}\, \Lip(\sigma_{n,i}^Y)\times d_n(\bfW_n^{D,(s)},\widetilde{\bfW}_n^{D,(s)}).
\end{aligned}
\end{equation*}
Thus, we can let $ \Lip(f_1)= \Lip(f) \bar{\sigma}^Y$. We can similarly let $ \Lip(f_1')= \Lip(f') \bar{\sigma}^Y$. Combining these terms yields the desired result. 
\end{proof}

\begin{proof}[\textbf{Proof of Proposition \ref{prop:W_psiDep}}]
    Let $(H,H')\in \mathcal{P}_n(h,h',4s+1)$. Define $\xi = f(\bfW_{n,H})$, $\zeta=f'(\bfW_{n,H'})$, and
\begin{equation*}
    \xi^{(s)}=f(W_{n,i}^{(s)}:i\in H),\hspace{2mm} \zeta^{(s)}=f(W_{n,i}^{(s)}:i\in H').
\end{equation*}
Again, we have the decomposition
\begin{equation*}
    \left| \Cov(\xi,\zeta|\mathcal{C}_n) \right| \leq \left| \Cov(\xi-\xi^{(s)},\zeta|\mathcal{C}_n) \right| + \left| \Cov(\xi^{(s)},\zeta-\zeta^{(s)}|\mathcal{C}_n) \right| + \left| \Cov(\xi^{(s)},\zeta^{(s)}|\mathcal{C}_n) \right|\hspace{2mm} a.s.
\end{equation*}
The first two terms can be bounded as
\begin{equation*}
\begin{aligned}
    &\left| \Cov(\xi-\xi^{(s)},\zeta|\mathcal{C}_n) \right| + \left| \Cov(\xi^{(s)},\zeta-\zeta^{(s)}|\mathcal{C}_n) \right|\\
    &\leq 2 (h \lVert f' \rVert_{\infty} \Lip(f) + h' \lVert f \rVert_{\infty} \Lip(f'))\times \underset{i}{\max}\, E[\lVert W_{n,i}-W_{n,i}^{(s)}\rVert_2 \mathcal{C}_n]\\
    &\leq  (h \lVert f' \rVert_{\infty} \Lip(f) + h' \lVert f \rVert_{\infty} \Lip(f')) \\
    &\hspace{2mm}\times 2(\underset{i}{max}\, E[| Y_{n,i}-Y_{n,i}^{(s)}| \mathcal{C}_n]+ \underset{i}{\max}\, E[| D_{n,i}-D_{n,i}^{(s)}| \mathcal{C}_n] + \underset{i}{\max}\, E[\lVert T_{n,i}-T_{n,i}^{(s)}\rVert_2 | \mathcal{C}_n]),
\end{aligned}    
\end{equation*}
where the second equality uses the fact that
\begin{equation*}
\begin{aligned}
    \lVert W_{n,i}-W_{n,i}^{(s)}\rVert_2 &\leq |Y_{n,i}-Y_{n,i}^{(s)}|+ |D_{n,i}-D_{n,i}^{(s)}| + \lVert T_{n,i}-T_{n,i}^{(s)}\rVert_2.
\end{aligned}
\end{equation*}
Now we focus on the term $\left| \Cov(\xi^{(s)},\zeta^{(s)}|\mathcal{C}_n) \right|$. 
First, we write out $\xi^{(s)}$ and $\zeta^{(s)}$ as
\begin{equation*}
\begin{aligned}
    \xi^{(s)}=f(W_{n,i}^{(s)}:i\in H) = f((Y_{n,i}^{(s)},D_{n,i}^{(s)},T_{n,i}^{(s)}):i\in H)
\end{aligned}
\end{equation*}
and
\begin{equation*}
    \zeta^{(s)} = f'((Y_{n,j}^{(s)},D_{n,j}^{(s)},T_{n,j}^{(s)}):j\in H').
\end{equation*}
Since $H,H'$ are at least $4s+1$ away and $D_{n,i}^{(s)}$ and $D_{n,j}^{(s)}$ use non-overlapping $\nu$, they are independent. On the other hand, by Proposition \ref{prop:Y_psiDep} and \ref{prop:T_psiDep}, $(Y_{n,i}^{(s)},T_{n,i}^{(s)})$ and $(Y_{n,j}^{(s)},T_{n,j}^{(s)})$ are dependent due to the $\psi$-dependence of $\bfD_n^{(s,i)}$ and $\bfD_n^{(s,j)}$. 
This setup is the same as that of Lemma \ref{app_lemma:WD_psiDep}, where we consider a vector stacking the $\psi$-dependent variables with independent variables and the same functional and $\theta$ sequence will constitute the bound. 
Thus, by the same argument as in the proof of Lemma \ref{app_lemma:WD_psiDep} and Proposition \ref{prop:Y_psiDep}, we obtain 
\begin{equation*}
\begin{aligned}
    &\left| \Cov(\xi^{(s)},\zeta^{(s)}|\mathcal{C}_n) \right| \\
    &\leq \left( |N_{A_n}(H,s)|\lVert f' \rVert_{\infty} \Lip(f) \left( \bar{\sigma}^Y + \bar{\sigma}^T  \right) + |N_{A_n}(H',s)|\lVert f \rVert_{\infty} \Lip(f') \left( \bar{\sigma}^Y + \bar{\sigma}^T  \right) \right) \theta_{n,s}^D.
\end{aligned}
\end{equation*}
Combining all these results together yields the desired bound.
\end{proof} 

Before we end this subsection, we further provide the following two corollaries to show the conditional $\psi$-dependence properties of $c'\bm{g}_n$, which will be used in proving the asymptotic distribution of the GMM estimator later.
\begin{cor}\label{cor:g_psidep}
Consider the triangular array $\{\bm{g}_{n,i}\}$ defined by $\bm{g}_{n,i}=\bm{g}((W_{n,i},X_{n,i},\bfZ),\bm{\beta}_0)$. Let $f\in \mathcal{L}_{\dim(g),h}, f'\in \mathcal{L}_{\dim(g),h'}$. Then, for any $H,H'\in \mathcal{P}(h,h',6s+1)$, we have
\begin{equation}
    |\Cov(f(\bm{g}_{n,H}),f'(\bm{g}_{n,H'}) |\mathcal{C}_n)| \leq \psi_{h,h}^g(f,f') \theta_{n,s}^g\hspace{2mm} a.s.,
\end{equation}
where
\begin{equation}
    \begin{aligned}
        \theta_{n,s}^g&=\overline{N}_{A_n}(s)\times \theta_{n,s}^W\\
        \psi_{h,h'}^g(f,f')&= h\lVert f'\rVert_{\infty}\Lip(f)(1+\Lip(\bm{g})) + h'\lVert f\rVert_{\infty}\Lip(f')(1+\Lip(\bm{g})).
    \end{aligned}    
    \end{equation}
\end{cor}
\begin{proof}
    By Assumption \ref{asm:regularityLLN}(iv), $\bm{g}(\cdot,\bm{\beta}_0)$ is Lipschitz with $\Lip(\bm{g})<\infty$. Thus, we can use the same argument as Proposition \ref{prop:Y_psiDep} and get $\psi_{h,h'}^g(f,f')$ as described in the statement with
    \begin{equation}
        \theta_{n,s}^g=2\times \underset{i\in \mathcal{N}_n}{\max} E[\lVert \bm{g}_{n,i}-\bm{g}_{n,i}^{(s)} \rVert |\mathcal{C}_n] + \overline{N}_{A_n}(s)\times \theta_{n,s}^W.
    \end{equation}
    Note that $\bm{g}$ only takes variables from sample $i$, $W_{n,i}$, as argument. We can conclude that $\bm{g}_{n,i}^{(s)}=\bm{g}_{n,i}$, and we have the desired result.
\end{proof}

In the following, we further characterize the conditional $\psi$-dependence of $c'\bm{g}_{n,i}$ for any given $c\in\mathbbm{R}^{\dim(g)}$ such that $\lVert c \rVert_q= 1$ with $q$ given in Assumption \ref{asm:regularityCLT} (ii). Notice that for all $\Tilde{f}:\mathbbm{R}^{\dim(g)\times h}\to\mathbbm{R}$ defined by $\Tilde{f}(\bm{g}_{n,H})=f(c'\bm{g}_{n,H})$ where $c'\bm{g}_{n,H}=(c'\bm{g}_{n,i})_{i\in H}$, $\Tilde{f}$ is belongs to $\mathcal{L}_{h,h}$. Moreover, given $\lVert c \rVert_q=1$:
\begin{equation}
    \Lip(\Tilde{f})=\Lip(f) \text{ and } \lVert \Tilde{f}\rVert_{\infty}=\lVert f\rVert_{\infty}.
\end{equation}
We can similarly define $\Tilde{f}'$. These steps suggest that $\psi_{h,h'}^g(\Tilde{f},\Tilde{f}')$ is well defined, and we can see why $c'\bm{g}_{n,i}$ is also conditional $\psi$-dependent. We state the result in the following.

\begin{cor}\label{cor:cg_psidep}
    Consider the triangular array $\{\bm{g}_{n,i}\}$ in Corollary \ref{cor:g_psidep}. Then, for any $c\in\mathbbm{R}^{dim(g)},\lVert c\rVert_q= 1$, all $H,H'\in \mathcal{P}(h,h',6s+1)$, and all $f\in \mathcal{L}_{h,h}, f'\in \mathcal{L}_{h',h'}$, we have
\begin{equation}
    |\Cov(f(c'\bm{g}_{n,H}),f'(c'\bm{g}_{n,H'})|\mathcal{C}_n)| \leq \psi_{h,h}^{cg}(f,f') \theta_{n,s}^{cg}\hspace{2mm} a.s.,
\end{equation}
where
\begin{equation}
    \begin{aligned}
        \theta_{n,s}^{cg}&=\theta_{n,s}^{g} \quad\text{and}\\
        \psi_{h,h'}^{cg}(f,f')&= h\lVert f'\rVert_{\infty}\Lip(f)(1+\Lip(\bm{g})) + h'\lVert f\rVert_{\infty}\Lip(f')(1+\Lip(\bm{g})).
    \end{aligned}    
    \end{equation}
\end{cor}
\begin{proof}
    Let $\Tilde{f}(\cdot)=f(c'\cdot),\Tilde{f}'(\cdot)=f'(c'\cdot)$. The result immediately follows from the fact that $|\Cov(f(c'\bm{g}_{n,H}),f'(c'\bm{g}_{n,H'})|\mathcal{C}_n)| = |\Cov(\Tilde{f}(\bm{g}_{n,H}),\Tilde{f}'(\bm{g}_{n,H'})|\mathcal{C}_n)|$, and that
\begin{equation}
\begin{aligned}
    &Lip(\Tilde{f})=\Lip(f), \hspace{3mm} \Lip(\Tilde{f}')=\Lip(f'),\\
    &\lVert \Tilde{f}\rVert_{\infty}=\lVert f\rVert_{\infty}, \text{ and } \lVert \Tilde{f}'\rVert_{\infty}=\lVert f'\rVert_{\infty}.
\end{aligned}
\end{equation}
Now the proof is complete.
\end{proof}

\subsection{Proofs for the Asymptotic Results}
In the following, we proceed to prove the consistency and limit distributional property of our proposed GMM estimator. To prove the consistency of our GMM estimator, we first show the uniform weak law of large numbers for the GMM objective function.

\begin{lemma}\label{app_lemma:QnULLN}
Suppose that Assumption \ref{asm:regularityLLN} holds. Suppose also that the weight matrix $\widehat{\Xi}_n$ satisfies $\underset{n\geq 1}{\sup} E[\lVert \widehat{\Xi}_n \rVert_F|\mathcal{C}_n]<\infty$ and $E[\lVert\widehat{\Xi}_n -\Xi_n\rVert_F |\mathcal{C}_n] \to 0\,\,a.s $ with $\Xi_n$ being $\mathcal{C}_n$-measurable with finite elements and positive definite a.s. Then, for all $\epsilon>0$, we have
\begin{equation}
    P\left(\left. \underset{\bm{\beta}\in\Theta}{\sup}| \widehat{Q}_n(\bm{\beta})-Q_n(\bm{\beta})  |>\epsilon \right|\mathcal{C}_n \right) \to 0\,\, a.s.,
\end{equation}
and thus we have
\begin{equation}
   \underset{\bm{\beta}\in \Theta}{\sup} | \widehat{Q}_n(\bm{\beta})- Q_n(\bm{\beta}) |\overset{P}{\to} 0.
\end{equation}
\end{lemma}
\begin{proof}[\textbf{Proof of Lemma \ref{app_lemma:QnULLN}}]
We first establish ULLN for the moments conditional on the $\sigma$-field $\mathcal{C}_n$. 
First, given $\overline{\sigma}^Y$ and $\overline{\sigma}^T$ from Proposition \ref{prop:Y_psiDep} and \ref{prop:T_psiDep}, by letting $C_W=1+\overline{\sigma}^Y+\overline{\sigma}^T$, we have $\psi_{h,h'}^W$ satisfying Assumption 1 in \cite{Sasaki:25}. Moreover, our Assumptions \ref{asm:regularityLLN} (i)-(v) correspond to Assumptions 2-5 of \cite{Sasaki:25}, and we can derive the same results as \citet[][Theorem 1]{Sasaki:25}. That is,
\begin{equation}\label{app_equ:yuya_ULLN}
    \left\lVert \underset{\bm{\beta}\in\Theta}{\sup}\lVert\frac{1}{n} \sum_{i\in \mathcal{N}_n} (\bm{g}((W_{n,i},X_{n,i},\bfZ),\bm{\beta})-E[\bm{g}((W_{n,i},X_{n,i},\bfZ),\bm{\beta})|\mathcal{C}_n])\rVert_2 \right\rVert_{\mathcal{C}_n,L^1}\to 0\,\,a.s.
\end{equation}
Now, to derive the conditional ULLN for the objective function, we note the deterministic inequalities
\begin{equation*}
    \begin{aligned}
        |\widehat{Q}_n(\bm{\beta})&-Q_n(\bm{\beta})|\\
        &\leq |[ \widehat{\bm{g}}_n(\bm{\beta})-\bm{g}_n(\bm{\beta})]'\widehat{\Xi}_n [ \widehat{\bm{g}}_n(\bm{\beta})-\bm{g}_n(\bm{\beta})]| + | \bm{g}_n(\bm{\beta})'(\widehat{\Xi}_n + \widehat{\Xi}_n') [ \widehat{\bm{g}}_n(\bm{\beta})-\bm{g}_n(\bm{\beta})]|\\
        &\hspace{2mm} + |\bm{g}_n(\bm{\beta})' (\widehat{\Xi}_n-\Xi_n)\bm{g}_n(\bm{\beta})|\\
        &\leq \lVert \widehat{\bm{g}}_n(\bm{\beta})-\bm{g}_n(\bm{\beta}) \rVert_2 ^2 \lVert \widehat{\Xi}_n \rVert_F + 2 \lVert \bm{g}_n(\bm{\beta}) \rVert_2 \lVert \widehat{\bm{g}}_n(\bm{\beta})-\bm{g}_n(\bm{\beta})\rVert_2 \lVert \widehat{\Xi}_n \rVert_F\\
        &\hspace{2mm} + \lVert \bm{g}_n(\bm{\beta}) \rVert_2^2 \lVert \widehat{\Xi}_n-\Xi_n \rVert_F\\
        &= A_n + B_n + C_n.
    \end{aligned}
\end{equation*}
Now take supremum over $\bm{\beta}$ and conditional expectation conditioning on $\mathcal{C}_n$. $A_n$ and $B_n$ go to $0$ almost surely by \eqref{app_equ:yuya_ULLN} and uniformly conditional boundedness of $\hat{\Xi}_n$ maintained in the statement. In addition, $C_n$ go to zero almost surely as well by the assumption stated in the statement and boundedness of $g_n$ due to Assumption \ref{asm:regularityLLN} (iv).
In summary, we have $E\left[\left. \underset{\bm{\beta}\in\Theta}{\sup}| \widehat{Q}_n(\bm{\beta})-Q_n(\bm{\beta})  | \right|\mathcal{C}_n \right] \to 0\,\, a.s$. By Markov's inequality, we conclude
\begin{equation}\label{app_equ:cond_QULLN}
    P\left(\left. \underset{\bm{\beta}\in\Theta}{\sup}| \widehat{Q}_n(\bm{\beta})-Q_n(\bm{\beta})  | >\epsilon\right|\mathcal{C}_n \right) \leq \epsilon^{-1} E\left[ \left. \underset{\bm{\beta}\in\Theta}{\sup}| \widehat{Q}_n(\bm{\beta})-Q_n(\bm{\beta})  | \right|\mathcal{C}_n\right]
    \to 0\,\,a.s.
\end{equation}
for all $\epsilon>0$.
This shows the first result. Furthermore, applying the law of iterated expectation yields
\begin{equation*}
    P\left(\underset{\bm{\beta}\in\Theta}{\sup}| \widehat{Q}_n(\bm{\beta})-Q_n(\bm{\beta})  | > \epsilon\right) 
    \leq 
    E\left[ \underset{\to 0\,\,a.s.\text{ by previous equation.}}{\underline{P\left( \left. \underset{\bm{\beta}\in\Theta}{\sup}| \widehat{Q}_n(\bm{\beta})-Q_n(\bm{\beta})  |>\epsilon \right|\mathcal{C}_n\right)}}   \right] \to 0.
\end{equation*}
Thus, we have the desired result.
\end{proof}

\begin{proof}[\textbf{Proof of Theorem \ref{thm:GMM_consistency}:}]
    We adopt the proof of Theorem 5.7 in \cite{Vaart:98} to our setup. First, note that the true value $\bm{\beta}_0$ uniquely minimizes the objective function. This statement can be written as
    \begin{equation}\label{app_equ:UniqMini}
        \underset{\bm{\beta}:\lVert \bm{\beta}-\bm{\beta}_0 \rVert_2\geq \epsilon}{\inf} Q_n(\bm{\beta}) > Q_n(\bm{\beta}_0)\,\,a.s.,
    \end{equation}
    for every $\epsilon>0$. Second, since $\bm{g}$ is Lipschitz in $\bm{\beta}$ and thus continuous by Assumption \ref{asm:regularityLLN} (v), and $\Theta$ is compact by \ref{asm:regularityLLN} (iii), we attain the minimum of $\widehat{Q}_n$ on $\Theta$. That is, when $\widehat{\bm{\beta}}= \arg\,\underset{\bm{\beta}\in \Theta}{\min}\widehat{Q}_n(\bm{\beta})$, we have $\widehat{Q}_n(\widehat{\bm{\beta}})= \underset{\bm{\beta}\in \Theta}{\inf}\, \widehat{Q}_n(\bm{\beta})$. Thus,
    \begin{equation*}
    \begin{aligned}
        |\widehat{Q}_n(\widehat{\bm{\beta}})-\widehat{Q}_n(\bm{\beta}_0)| &=\left|\widehat{Q}_n(\widehat{\bm{\beta}})-Q_n(\bm{\beta}_0)-\underset{r_{1n}}{\underbrace{(\widehat{Q}_n(\bm{\beta}_0) - Q_n(\bm{\beta}_0))}}\right|\\
        &=\left| \widehat{Q}_n(\widehat{\bm{\beta}})-\underset{\bm{\beta}\in\Theta}{\inf}Q_n(\bm{\beta})-r_{1n}   \right|\\
        & = \left| \underset{=0}{\underbrace{\widehat{Q}_n(\widehat{\bm{\beta}})-\underset{\bm{\beta}\in\Theta}{\inf}\widehat{Q}_n(\bm{\beta})}}-r_{1n} - \underset{r_{2n}}{\underbrace{ (\underset{\bm{\beta}\in\Theta}{\inf}\widehat{Q}_n(\bm{\beta}) - \underset{\bm{\beta}\in\Theta}{\inf}Q_n(\bm{\beta}))}  }\right|\leq |r_{1n}| + |r_{2n}|.
    \end{aligned}    
    \end{equation*}
    
    Now we show $P\left(|r_{1n}|>\epsilon\Big|\mathcal{C}_n\right)$ and $P\left(|r_{2n}|>\epsilon\Big|\mathcal{C}_n\right)$ go to zero as the result of \eqref{app_equ:cond_QULLN}.
    For all $\epsilon>0$, 
    \begin{equation*}
    P\left(|r_{1n}|>\epsilon\Big|\mathcal{C}_n\right)=P\left(|\widehat{Q}_n(\bm{\beta}_0)-Q_n(\bm{\beta}_0)|>\epsilon\Big|\mathcal{C}_n\right)\leq P\left(\left. \underset{\bm{\beta}\in\Theta}{\sup}| \widehat{Q}_n(\bm{\beta})-Q_n(\bm{\beta})  | >\epsilon\right|\mathcal{C}_n \right) \to 0
    \end{equation*} 
    a.s. by \eqref{app_equ:cond_QULLN}. Moreover, as inf functional is 1-Lipschitz with respect to the sup-norm, given $\epsilon>0$, we have
    \begin{equation*}
    \begin{aligned}
        P\left(|r_{2n}|>\epsilon\Big|\mathcal{C}_n\right)&=P\left(|\underset{\bm{\beta}\in\Theta}{\inf}\widehat{Q}_n(\bm{\beta}) - \underset{\bm{\beta}\in\Theta}{\inf}Q_n(\bm{\beta})|>\epsilon\Big|\mathcal{C}_n\right)\\ &\leq P\left(\left. \underset{\bm{\beta}\in\Theta}{\sup}| \widehat{Q}_n(\bm{\beta})-Q_n(\bm{\beta})  | >\epsilon\right|\mathcal{C}_n \right) \to 0\,\, a.s.
    \end{aligned}
    \end{equation*}also as a result of \eqref{app_equ:cond_QULLN}.
    Thus,
    \begin{equation}\label{app_equ:hatQn_SmallDist}
        P(|\widehat{Q}_n(\widehat{\bm{\beta}})-\widehat{Q}_n(\bm{\beta}_0)|>\epsilon |\mathcal{C}_n)\to 0\,\, a.s.
    \end{equation}
    
    Now, by (\ref{app_equ:UniqMini}), given an $\epsilon>0$, there is $\delta>0$ such that
    \begin{equation*}
    \begin{aligned}
        P&\left(  \lVert \widehat{\bm{\beta}}-\bm{\beta}_0 \rVert_2>\delta|\mathcal{C}_n \right) \leq P\left(|Q_n(\widehat{\bm{\beta}}) - Q_n(\bm{\beta})| > \epsilon|\mathcal{C}_n\right)\\
        &\leq P\left (|Q_n(\widehat{\bm{\beta}}) - \widehat{Q}_n(\widehat{\bm{\beta}})|>\epsilon/3|\mathcal{C}_n\right)
        + P\left (|\widehat{Q}_n(\bm{\beta}_0) - \widehat{Q}_n(\widehat{\bm{\beta}})|>\epsilon/3|\mathcal{C}_n\right)\\
        &\hspace{2mm}+
        P\left (|\widehat{Q}_n(\widehat{\bm{\beta}}) - Q_n(\bm{\beta}_0)|>\epsilon/3|\mathcal{C}_n\right),
    \end{aligned}    
    \end{equation*}
    where the first and the last terms go to $0\,\, a.s$. by the first result of Lemma \ref{app_lemma:QnULLN} and the middle term also goes to $0\,\, a.s.$ by \eqref{app_equ:hatQn_SmallDist}. Therefore,$P\left(  \lVert \widehat{\bm{\beta}}-\bm{\beta}_0 \rVert_2>\delta|\mathcal{C}_n \right)\to 0\,\,a.s$. Applying the law of iterated expectation gives $\widehat{\bm{\beta}} \overset{P}{\to}\bm{\beta}$.   
\end{proof}

\begin{proof}[\textbf{Proof of Theorem \ref{thm:GMM_CLT}}] 
    Let $p$, $q$, and $r$ satisfy the conditions in Assumption \ref{asm:regularityCLT} (i) and (ii).
    First, we show the CLT of $\{c'\bm{g}_{n,i}\}$ for all $c \in \mathbbm{R}^{\dim(g)},\lVert c\rVert_q=1$ by checking the conditions of Theorem 3.2 of KMS. Recall that we have established the conditional $\psi$-dependence of $\{c'\bm{g}_{n,i}\}$ in Corollary \ref{cor:cg_psidep} with $\theta^{cg}_{n,s}=\theta^{g}_{n,s}$ (defined in Corollary \ref{cor:g_psidep}) and
    \begin{equation*}
        \psi_{h,h'}^{cg}(f,f')= h\lVert f'\rVert_{\infty}\Lip(f)(1+\Lip(\bm{g})) + h'\lVert f\rVert_{\infty}\Lip(f')(1+\Lip(\bm{g})).
    \end{equation*}
    Now, given Assumption \ref{asm:regularityLLN}(ii), we have $\theta_{n,s}^g\to 0$. In addition, we can take $C_g=1+\Lip(\bm{g})$. These steps jointly imply that $\{c'\bm{g}_{n,i}\}$ is conditional $\psi$-dependent and satisfies Assumption 2.1 of KMS. 

    Moreover, by H\"older's inequality, we have
    \begin{equation*}
        \lVert c'\bm{g}_{n,i} \rVert_{\mathcal{C}_n,r}\leq \lVert c\rVert_q \lVert \bm{g}_{n,i}\rVert_{\mathcal{C}_n,p} = \lVert \bm{g}_{n,i}\rVert_{\mathcal{C}_np}.
    \end{equation*}
    Thus, by taking $\underset{n\geq 1}{\sup}\,\underset{i\in \mathcal{N}_n}{\max}$ on both sides, we satisfy Assumption 3.3 of KMS by our Assumption \ref{asm:regularityCLT}(i). Also, our Assumption \ref{asm:regularityCLT} (ii) corresponds to Assumption 3.4 of KMS. Thus, we have
    \begin{equation}\label{app_equ:KMS_CLT}
        \sup_{t\in\mathbbm{R}}\left| Pr\left\{ \frac{\sum_{i=1}^n c'\bm{g_{n,i}}}{\sigma_n(c)}\leq t | \mathcal{C}_n  \right\} - \Phi(t)   \right| \overset{a.s.}{\to} 0
    \end{equation}
    as $n \to \infty$, where $\Phi(t)$ is the cumulative distribution function of the standard normal distribution.
    We show that \eqref{app_equ:KMS_CLT} implies
    \begin{equation*}
        H_n\equiv \frac{c'S_n}{\sigma_n(c)} \overset{d}{\to} Z_c\sim N(0,1) \,\,(\mathcal{C}\stably)
    \end{equation*}
    with $Z_c$ being independent of $\mathcal{C}$.
    To see why this is true, consider any $\mathcal{C}$-measurable $P$-essentially bounded variable $\zeta$ and a bounded function $f$. As $\mathcal{C}\subset \mathcal{C}_n$ for all $n$, $\zeta$ is also $\mathcal{C}_n$-measurable. We apply the law of iterated expectation and write $E[\zeta f(H_n)]=E[\zeta E[f(H_n) | \mathcal{C}_n]]$.
    From \eqref{app_equ:KMS_CLT} and the boundedness of $f$, we know $E[f(H_n)|\mathcal{C}_n]\to E[f(Z_c)]\, a.s$. Thus, we have
    \begin{equation*}
        E[\zeta f(H_n)]=E[\zeta E[f(H_n) | \mathcal{C}_n]] \to E[\zeta f(Z)].
    \end{equation*}
    Therefore, in light of Proposition \ref{app_prop:Daley_stableEQU} (ii), we can draw the conclusion that $H_n\overset{d}{\to} Z_c\sim N(0,1) \,\,(\mathcal{C}\stably)$.
    Furthermore, by Assumption \ref{asm:regularityCLT} (v), for each $c$, 
    \begin{equation*}
        \frac{1}{n}\sigma_n^2(c) = c'\frac{1}{n}\Omega_{g,n}c \overset{P}{\to} c'\Omega_{g,0}c,
    \end{equation*}
    which is $\mathcal{C}$-measurable. 
    We can further apply the continuous mapping theorem and get $\sigma_n(c)/\sqrt{n}\overset{P}{\to}v_c^{1/2}$ with $v_c=c'\Omega_{g,0}c$. 
    Thus, using the conditional version of Slutsky and continuous mapping theorems in Theorem \ref{app_thm:cond_slutsky}(i) and (ii), we conclude that
    \begin{equation*}
        \frac{c'S_n}{\sqrt{n}}= \frac{c'S_n}{\sigma_n(c)} \frac{\sigma_n(c)}{\sqrt{n}} \overset{d}{\to} v_c^{1/2}Z_c\,\,(\mathcal{C}\stably)
    \end{equation*} 
    with $v_c=c'\Omega_{g,0}c$ and $Z_c \sim N(0,1)$ and is independent of $\mathcal{C}$ and thus $v_c$. The conclusion holds for all vectors $c$ with $\lVert c \rVert_q=1$. Thus in light of Proposition \ref{app_prop:KP2013_propA2}(a), we have
    \begin{equation}\label{ass_equ:Sn_stable}
        \frac{S_n}{\sqrt{n}} \overset{d}{\to} \Omega_{g,0}^{1/2} Z\,\,(\mathcal{C}\stably),
    \end{equation}
    where $Z\sim N(0,I_{dim(g)})$ and is independent of $\mathcal{C}$ and thus $\Omega_{g,0}$. Now we adopt the proof strategy of Theorem 3.3 in \cite{Newey2:94}. Note that we find the GMM estimator by solving the first order condition $2\widehat{G}_n(\widehat{\bm{\beta}})'\widehat{\Xi}_n\widehat{\bm{g}}_n(\widehat{\bm{\beta}})=0$. We expand $\widehat{\bm{g}}_n(\widehat{\bm{\beta}})$ around true parameter value $\bm{\beta}_0$ and get
    \begin{equation*}
        0= \widehat{G}_n(\widehat{\bm{\beta}})'\widehat{\Xi}_n \frac{S_n}{\sqrt{n}}+(\widehat{G}_n(\widehat{\bm{\beta}})'\widehat{\Xi}_n(\widehat{G}_n(\overline{\bm{\beta}}))\sqrt{n}(\widehat{\bm{\beta}}-\bm{\beta}_0),
    \end{equation*}
    where $\overline{\bm{\beta}}$ is a point between $\widehat{\bm{\beta}}$ and $\bm{\beta}_0$. By Assumption \ref{asm:regularityCLT}(iv) and the assumptions in the statement that $\underset{\bm{\beta}\in\mathcal{N}_{\beta_0}}{\sup}\lVert \widehat{G}_n(\bm{\beta})-G_n(\bm{\beta}) \rVert \overset{P}{\to} 0$ and $\widehat{\Xi}_n\overset{P}{\to}\Xi_n$, and $\widehat{\bm{\beta}}\overset{P}{\to}\bm{\beta}$, we have
    \begin{equation*}
        \sqrt{n}(\widehat{\bm{\beta}}-\bm{\beta}_0) = - (G_n(\bm{\beta}_0)'\Xi_nG_n(\bm{\beta}_0))^{-1}G_n(\bm{\beta}_0)\Xi_n \frac{S_n}{\sqrt{n}}+o_p(1),
    \end{equation*}
    where $o_p(1)$ is unconditional as we have unconditional convergence of $\underset{\bm{\beta}\in\mathcal{N}_{\beta_0}}{\sup}\lVert \widehat{G}_n(\bm{\beta})-G_n(\bm{\beta}) \rVert \overset{P}{\to} 0$, $\widehat{\Xi}_n\overset{P}{\to}\Xi_n$, and $\widehat{\bm{\beta}}\overset{P}{\to}\bm{\beta}$.    
    Here, $G_n(\bm{\beta}_0)'\Xi_nG_n(\bm{\beta}_0)$ is invertible because each of the elements are of full rank almost surely. Note that we have established $S_n/\sqrt{n}\overset{d}{\to}\Omega_{g,0}^{1/2} Z\,\,(\mathcal{C}\stably)$, and that $G_0$ and $\Xi_0$ are $\mathcal{C}$-measurable. Thus, in light of Proposition \ref{app_prop:KP2013_propA2}(b), for any $\mathcal{C}$-measurable random variable $\zeta$, we have
    \begin{equation*}
        (G_0,\Xi_0,S_n/\sqrt{n},\zeta) \overset{d}{\to} (G_0,\Xi_0,\Omega_{g,0}^{1/2} Z,\zeta).
    \end{equation*}
    Moreover, by Assumption \ref{asm:regularityCLT}(v) that $G_n(\bm{\beta}_0)\overset{P}{\to}G_0$ and $\Xi_n\overset{P}{\to}\Xi_0$, it follows that
    \begin{equation*}
        (G_n(\bm{\beta_0})-G_0,\Xi_n-\Xi_0,G_0,\Xi_0,S_n/\sqrt{n},\mathcal{\zeta})\overset{d}{\to} (0,0,G_0,\Xi_0,\Omega_{g,0}^{1/2} Z,\mathcal{\zeta}).
    \end{equation*}
    Since $G_0'\Xi_0G_0$ is positive definite a.s., we can apply the continuous mapping theorem to get
    \begin{equation*}
        (\sqrt{n}(\widehat{\bm{\beta}}-\bm{\beta}_0),\zeta) \overset{d}{\to} ((G_0'\Xi_0G_0)^{-1}G_0'\Xi_0\Omega_{g,0}^{1/2} Z,\zeta).
    \end{equation*}
    Thus, by Proposition \ref{app_prop:Daley_stableEQU}(iii), we conclude that
    \begin{equation*}
        \sqrt{n}(\widehat{\bm{\beta}}-\bm{\beta}_0) \overset{d}{\to} \Omega_{\beta,0}^{1/2}Z\,\,(\mathcal{C}\stably),
    \end{equation*}
    where $\Omega_{\beta,0}=(G_0'\Xi_0G_0)^{-1}G_0'\Xi_0\Omega_{g,0}\Xi_0G_0(G_0'\Xi_0G_0)^{-1}$ and $Z\sim N(0,I_{dim(g)})$, and $Z$ is independent of $\mathcal{C}$ and thus $\Omega_{\bm{\beta},0}$.
\end{proof}

\begin{proof}[\textbf{Proof of Theorem \ref{thm:GMM_Wald}}]
  Following the conclusion of Theorem \ref{thm:GMM_CLT}, we have
  \begin{equation*}
      \sqrt{n}(\widehat{\bm{\beta}}-\bm{\beta}_0) \overset{d}{\to} \Omega_{\beta,0}^{1/2}Z\,\,(\mathcal{C}\stably),
    \end{equation*}
    where $\Omega_{\beta,0}=(G_0'\Xi_0G_0)^{-1}G_0'\Xi_0\Omega_{g,0}\Xi_0G_0(G_0'\Xi_0G_0)^{-1}$ and $Z\sim N(0,I_{\dim(g)})$, and $Z$ is independent of $\mathcal{C}$ and thus $\Omega_{\bm{\beta},0}$. In light of the consistency and positive definiteness of $\widehat{\Omega}_\beta$, the consistency of $\widehat{\bm{\beta}}$ to $\bm{\beta}_0$, and that $R$ has full row rank $k$ and is continuous around $\bm{\beta}_0$, we derive
    \begin{equation*}
        (R(\widehat{\bm{\beta}})'\widehat{\Omega}_\beta R(\widehat{\bm{\beta}}))^{-1/2}\sqrt{n}\cdot r(\widehat{\bm{\beta}}) = (R(\bm{\beta}_0)'\Omega_{\beta,0}R(\bm{\beta}_0))^{-1/2} \sqrt{n}\cdot r(\widehat{\bm{\beta}}) + o_p(1),
    \end{equation*}
    where $o_p(1)$ is an error term that unconditionally converges into $0$ in probability due to unconditional consistencies: $\widehat{\bm{\beta}}\overset{P}{\to}\bm{\beta}_0$ and $\widehat{\Omega}_\beta\overset{P}{\to}\Omega_{\beta,0}$.
    Now, expanding $r(\widehat{\bm{\beta}})$ around $\bm{\beta}_0$ and noting $H_0:r(\bm{\beta}_0)=0$, we have
    \begin{equation*}
        (R(\bm{\beta}_0)'\Omega_{\beta,0}R(\bm{\beta}_0))^{-1/2} \sqrt{n}\cdot r(\widehat{\bm{\beta}}) =  (R(\bm{\beta}_0)'\Omega_{\beta,0}R(\bm{\beta}_0))^{-1/2} R(\bm{\beta}_0)'\sqrt{n}(\widehat{\bm{\beta}}-\bm{\beta}_0) +o_p(\lVert \widehat{\bm{\beta}}-\bm{\beta}_0 \rVert_2).
    \end{equation*}
    Observe that $B=(R(\bm{\beta}_0)'\Omega_{\beta,0}R(\bm{\beta}_0))^{-1/2} R(\bm{\beta}_0)'$ is $\mathcal{C}$-measurable and $B\Omega_{\beta,0}B'=I_{k}$. Then, it follows from Proposition \ref{app_prop:KP2013_propA2}(b) that
    \begin{equation*}
        (R(\widehat{\bm{\beta}})'\widehat{\Omega}_\beta R(\widehat{\bm{\beta}}))^{-1/2}\sqrt{n}\cdot r(\widehat{\bm{\beta}}) \overset{d}{\to} Z_k
    \end{equation*}
    with $Z_k\sim N(0,I_k).$ Therefore, $T_n$ converges in distribution to a chi-square distribution with $k$ degree of freedom by the continuous mapping theorem. Moreover, $\widehat{\Omega}_\beta^{-1/2}\sqrt{n}(\widehat{\bm{\beta}}-\bm{\beta}_0)$ is a special case by choosing $r(\bm{\beta})=\bm{\beta} -\bm{\beta}_0$.
\end{proof}

\subsection{Auxiliary Propositions}\label{sec:useful_propositions}
We summarize the equivalent statements regarding the stable convergence in the following proposition. The equivalence is shown in Proposition A3.2.IV in \cite{Daley2:03}.
\begin{prop}\label{app_prop:Daley_stableEQU}
    Let $\{W_n\}$, $W$, and $\mathcal{C}$ be as in Definition \ref{def:stable_cov}. Then the following are equivalent.
    \begin{itemize}
        \item[(i)] $W_n\overset{d}{\to}W$ ($\mathcal{C}$-stably).
        \item[(ii)] For all $\mathcal{C}$-measurable $P$-essentially bounded random variables $Z$ and all bounded continuous functions $h:\mathbbm{R}^p\to \mathbbm{R}$
        \begin{equation*}
            \underset{n\to \infty}{\lim}E[Zh(W_n)] = E[Zh(W)].
        \end{equation*}
        \item[(iii)] For all real valued $\mathcal{C}$-measurable random variables $Y$, we have
        \begin{equation*}
            (W_n,Y) \overset{d}{\to} (W,Y).
        \end{equation*}
        \item[(iv)] For all bounded continuous functions $g:\mathbbm{R}^p \times \mathbbm{R}\to \mathbbm{R}$, and all real valued $\mathcal{C}$-measurable random variables $Y$,
        \begin{equation*}
            g(W_n,Y) \overset{d}{\to} g(W,Y)\hspace{2mm} (C\stably).
        \end{equation*}
        \item[(v)] For all real vectors $t\in \mathbbm{R}^p$ and all $\mathcal{C}$-measurable $P$-essentially bounded random variables $Z$
        \begin{equation*}
            E[Z\, \exponential(it'W_n)] \to E[Z\, \exponential(it'W)]\hspace{2mm} \text{as }n\to \infty. 
        \end{equation*}
    \end{itemize}
\end{prop}

The following is a variant of Proposition A.2 in \cite{Kuersteiner2:13}.
\begin{prop}\label{app_prop:KP2013_propA2}
    Let $\{W_n\}$ and $\mathcal{C}$ be as in Definition \ref{def:stable_cov}, and let $\Xi$ be a $\mathcal{C}$-measurable, a.s. finite and positive definite $p\times p$ matrix. Suppose for any $c \in \mathbbm{R}^p$ with $\lVert c \rVert_q=1$ for some $q\in[1,\infty]$, we have
    \begin{equation}\label{equ:stableCLT_scalar}
        c'W_n \overset{d}{\to} \nu_c^{1/2}\xi_c\hspace{2mm} (\mathcal{C}\stably),
    \end{equation}
    with $\nu_c = c'\Xi c$, where $\xi_c$ is independent of $\mathcal{C}$ (and thus of $\Xi$) and $\xi_c\sim N(0,1)$. Then the characteristic function of $\nu_c^{1/2}\xi_c$ is given by $\phi_c(s) = E[ \text{exp} \{ -\frac{1}{2} (c'\Xi c)s^2 \}], s\in\mathbbm{R}$.

    \begin{itemize}
        \item[(a)] The above statement holds if and only if
        \begin{equation}\label{equ:stableCLT_vec}
            W_n \overset{d}{\to} \Xi^{1/2}\xi\hspace{2mm} (\mathcal{C}\stably),
        \end{equation}
        where $\xi$ is independent of $\mathcal{C}$ (and thus of $\Xi$) and $\xi\sim N(0,I_p)$. The characteristic function of $\Xi^{1/2}\xi$ is given by $\phi(t) = E[ \text{exp} \{ -\frac{1}{2} t'\Xi t \}], t\in\mathbbm{R}^p$.
        \item[(b)] Let $a$ be some $\mathcal{C}$-measurable vector, then (\ref{equ:stableCLT_scalar}) implies
        \begin{equation}
            (W_n',a') \overset{d}{\to}(W',a')\hspace{2mm} (\mathcal{C}\stably).
        \end{equation}
        Furthermore, let $A$ be some $p_* \times p$ matrix that is $\mathcal{C}$-measurable, a.s. finite and has full row rank. Then
        \begin{equation}
            AW_n \overset{d}{\to} A\Xi^{1/2}\xi\hspace{2mm} (\mathcal{C}\stably),
        \end{equation}
        where $\xi $ is as defined in part (a), and hence also
        \begin{equation}
            AW_n \overset{d}{\to} (A\Xi A')^{1/2}\xi_*\hspace{2mm} (\mathcal{C}\stably),
        \end{equation}
        where $\xi_*$ is independent of $\mathcal{C}$ (and thus of $A\Xi A'$) and $\xi_*\sim N(0,I_{p_*})$. Then the characteristic function of $A\Xi^{1/2}\xi$ and $(A\Xi A')^{1/2}\xi_*$ is given by $\phi_*(t_*) = E[ \text{exp} \{ -\frac{1}{2} (t_*'\Xi t_*) \}], t_*\in\mathbbm{R}^{p_*}$.
    \end{itemize}
\end{prop}
Note that we replace the restriction $c'c=1$ in \cite{Kuersteiner2:13} with $\lVert c \rVert_q=1$ for some $q\in[1,\infty]$. Since all norms are equivalent in finite dimensions, and they are just normalization for the Cram\'er–Wold device, the conclusion of the proposition remains.

When $X$ is $\mathcal{C}$-measurable, the stable convergence will be equivalent to convergence in probability. This is proven in \cite[][Corollary 3.6]{Häusler2:15}. We state the partial results in the following.
\begin{cor}\label{cor:stable_equ_covP}
    Let $\{W_n\}$, $W$, and $\mathcal{C}$ be as in Definition \ref{def:stable_cov} and $W$ is $\mathcal{C}$-measurable. The following statements are equivalent.
    \begin{itemize}
        \item[(i)] $W_n\overset{P}{\to} W$.
        \item[(ii)] $W_n\overset{d}{\to} \delta_W$ ($\mathcal{C}$-stably), where $\delta_W$ is the Dirac kernel associated with $W$ given by $\delta_W(\omega)=\delta_{W(\omega)}$.
    \end{itemize}
\end{cor}

Finally, it is also useful to include the conditional version of Slutsky and continuous mapping theorem. The following theorem is proven in \cite[][Theorem 3.18(b) and (c)]{Häusler2:15}.
\begin{thm}\label{app_thm:cond_slutsky}
    Assume $X_n\to X\hspace{2mm}\mathcal{C}$-stably and let $Y_n$ and $Y$ be random variables with values in $(\mathcal{Y},\mathcal{B}(\mathcal{Y}))$ for some separable metrizable space.
    \begin{itemize}
        \item[(i)] If $Y_n\to Y$ in probability and $Y$ is $\mathcal{C}$ measurable, then
        \begin{equation}\label{app_equ:cond_slutsky}
            (X_n,Y_n)\overset{d}{\to} (X,Y) \hspace{2mm} (\mathcal{C}\stably).
        \end{equation} 
        \item[(ii)] If function $g$ is Borel measurable and continuous almost surely with respect to $X$, then 
        \begin{equation}\label{app_equ:cond_CMP}
            g(X_n)\overset{d}{\to} g(X)\hspace{2mm} (\mathcal{C}\stably).
        \end{equation}
    \end{itemize}
\end{thm}
As we only consider random variables taking values in $\mathbbm{R}^k,k\geq 1$ equipped with the Borel sigma-field, the separability and metrizability requirements are always satisfied.
\end{appendices}

\vspace{1cm}
\bibliography{bib}

\end{document}


\begin{appendices}
\section{Additional Simulation}

\newpage

\begin{table}
\centering
\small
\caption{1,000 Monte Carlo Simulations for $\bm{\beta}$ on Ring Network ($c=0.3$)}
\label{tab:sim_betas_ring}
\begin{tabular}{lcccccccc}
    \toprule
    & \multicolumn{4}{c}{$n=250$} & \multicolumn{4}{c}{$n=1000$}  \\
    \cmidrule(lr){2-5}  \cmidrule(lr){6-9}
    & Bias & SE & RMSE & 95\% Coverage & Bias & SE & RMSE & 95\% Coverage\\
    \midrule
     $\beta_{D0}$         & 0.009  & 0.346 & 0.478 & 0.926  & 0.005  & 0.169 & 0.229 & 0.944 \\
     $\beta_{D1}$         & 0.026  & 0.271 & 0.371 & 0.938  & 0.006  & 0.133 & 0.181 & 0.948 \\
     $\lambda$            & -0.035 & 0.746 & 1.025 & 0.926  & -0.009 & 0.364 & 0.503 & 0.934 \\
     $\beta_{Y0}^{(0,0)}$  & -0.022 & 0.786 & 1.091 & 0.930  & 0.025  & 0.392 & 0.535 & 0.944 \\
     $\beta_{Y1}^{(0,0)}$  & -0.004 & 0.134 & 0.185 & 0.927  & 0.001  & 0.067 & 0.090 & 0.938 \\
     $\beta_{p1}^{(0,0)}$  & 0.042  & 1.254 & 1.749 & 0.934  & -0.042 & 0.630 & 0.862 & 0.940 \\
     $\beta_{Y0}^{(1,0)}$  & -0.031 & 0.439 & 0.639 & 0.897  & -0.016 & 0.224 & 0.304 & 0.947 \\
     $\beta_{Y1}^{(1,0)}$  & -0.002 & 0.176 & 0.250 & 0.913  & -0.003 & 0.090 & 0.125 & 0.928 \\
     $\beta_{p1}^{(1,0)}$  & 0.096  & 1.329 & 1.912 & 0.896  & 0.033  & 0.678 & 0.911 & 0.952 \\
     $\beta_{Y0}^{(0,1)}$  & -0.008 & 0.584 & 0.808 & 0.934  & 0.002  & 0.293 & 0.398 & 0.954 \\
     $\beta_{Y1}^{(0,1)}$  & 0.002  & 0.106 & 0.145 & 0.933  & -0.001 & 0.053 & 0.071 & 0.959 \\
     $\beta_{p1}^{(0,1)}$  & 0.013  & 0.903 & 1.251 & 0.931  & -0.005 & 0.454 & 0.618 & 0.942 \\
     $\beta_{Y0}^{(1,1)}$  & -0.014 & 0.334 & 0.465 & 0.936  & -0.001 & 0.168 & 0.229 & 0.948 \\
     $\beta_{Y1}^{(1,1)}$  & 0.000  & 0.122 & 0.165 & 0.927  & -0.001 & 0.061 & 0.085 & 0.934 \\
     $\beta_{p1}^{(1,1)}$  & 0.048  & 0.962 & 1.324 & 0.937  & -0.004 & 0.480 & 0.655 & 0.952 \\
    \bottomrule
\end{tabular}
\par \smallskip
\parbox{14cm}{\footnotesize Note: We consider the bandwidth and Parzen kernel HAC considered in \cite{Kojevnikov3:21} with constant term and $\epsilon$ in $b_n$ being 0.3 and 0.05.}
\end{table}

\begin{table}
\small
\caption{\centering 1,000 Monte Carlo Simulations for the marginal exposure response ($\overline{\psi}$) on Ring Network ($c=0.3$)}
\label{tab:sim_MER_ring}
\begin{tabular}{cccccccccc}
    \toprule
    && \multicolumn{4}{c}{$n=250$} & \multicolumn{4}{c}{$n=1000$}  \\
    \cmidrule[0.05em](l){3-6}  \cmidrule[0.05em](l){7-10}
    & $p$ & Bias & SE & RMSE & 95\% Coverage & Bias & SE & RMSE & 95\% Coverage\\
    \midrule
    \multirow{ 3}{*}{$MER^{(0,0)}$}
    &.2& -0.017 & 0.557 & 0.771 & 0.920 & 0.018 & 0.278 & 0.377 & 0.951 \\
    &.5& -0.005 & 0.240 & 0.331 & 0.930 & 0.005 & 0.119 & 0.160 & 0.951 \\
    &.8& 0.008  & 0.300 & 0.420 & 0.907 & -0.008 & 0.152 & 0.210 & 0.928 \\
    \multirow{ 3}{*}{$MER^{(1,0)}$}
    &.2& -0.014 & 0.290 & 0.412 & 0.906 & -0.012 & 0.145 & 0.200 & 0.928 \\
    &.5& 0.015  & 0.362 & 0.508 & 0.923 & -0.002 & 0.185 & 0.252 & 0.934 \\
    &.8& 0.043  & 0.705 & 1.005 & 0.909 & 0.008  & 0.360 & 0.486 & 0.948 \\
    \multirow{ 3}{*}{$MER^{(0,1)}$}
    &.2& -0.004 & 0.421 & 0.579 & 0.938 & 0.001 & 0.211 & 0.286 & 0.953 \\
    &.5& 0.000  & 0.194 & 0.263 & 0.937 & -0.001 & 0.097 & 0.131 & 0.962 \\
    &.8& 0.004  & 0.211 & 0.292 & 0.932 & -0.003 & 0.107 & 0.144 & 0.946 \\
    \multirow{ 3}{*}{$MER^{(1,1)}$}
    &.2& -0.004 & 0.211 & 0.292 & 0.928 & -0.004 & 0.106 & 0.146 & 0.952 \\
    &.5& 0.010  & 0.244 & 0.329 & 0.947 & -0.005 & 0.122 & 0.166 & 0.945 \\
    &.8& 0.024  & 0.491 & 0.668 & 0.938 & -0.006 & 0.245 & 0.334 & 0.942 \\
\bottomrule
\end{tabular}
\par\smallskip
\vspace{0.5pt}
\parbox{14cm}{\footnotesize Note: We consider the bandwidth and Parzen kernel HAC considered in \cite{Kojevnikov3:21} with constant term and $\epsilon$ being 0.3 and 0.05.}
\end{table}

\begin{table}
\centering
\small
\caption{1,000 Monte Carlo Simulations for $\bm{\beta}$ on Ring Network ($c=0.4$)}
\begin{tabular}{lcccccccc}
    \toprule
    & \multicolumn{4}{c}{$n=250$} & \multicolumn{4}{c}{$n=1000$}  \\
    \cmidrule(lr){2-5}  \cmidrule(lr){6-9}
    & Bias & SE & RMSE & 95\% Coverage & Bias & SE & RMSE & 95\% Coverage\\
    \midrule
     $\beta_{D0}$         & -0.002 & 0.341 & 0.470 & 0.928  & -0.003 & 0.170 & 0.231 & 0.937 \\
     $\beta_{D1}$         & 0.025  & 0.270 & 0.366 & 0.948  & 0.004  & 0.133 & 0.178 & 0.961 \\
     $\lambda$            & -0.001 & 0.733 & 1.011 & 0.925  & 0.008  & 0.367 & 0.494 & 0.946 \\
     $\beta_{Y0}^{(0,0)}$  & -0.003 & 0.776 & 1.080 & 0.919  & 0.008  & 0.391 & 0.535 & 0.946 \\
     $\beta_{Y1}^{(0,0)}$  & -0.001 & 0.132 & 0.186 & 0.907  & -0.002 & 0.067 & 0.091 & 0.949 \\
     $\beta_{p1}^{(0,0)}$  & 0.004  & 1.241 & 1.732 & 0.912  & -0.006 & 0.628 & 0.859 & 0.944 \\
     $\beta_{Y0}^{(1,0)}$  & 0.011  & 0.437 & 0.606 & 0.923  & 0.005  & 0.223 & 0.306 & 0.944 \\
     $\beta_{Y1}^{(1,0)}$  & 0.003  & 0.174 & 0.245 & 0.916  & -0.003 & 0.090 & 0.122 & 0.944 \\
     $\beta_{p1}^{(1,0)}$  & -0.042 & 1.340 & 1.861 & 0.927  & -0.007 & 0.675 & 0.927 & 0.942 \\
     $\beta_{Y0}^{(0,1)}$  & -0.012 & 0.581 & 0.807 & 0.940  & 0.004  & 0.293 & 0.400 & 0.934 \\
     $\beta_{Y1}^{(0,1)}$  & 0.002  & 0.105 & 0.143 & 0.939  & -0.002 & 0.053 & 0.073 & 0.936 \\
     $\beta_{p1}^{(0,1)}$  & 0.013  & 0.897 & 1.254 & 0.936  & -0.008 & 0.455 & 0.620 & 0.940 \\
     $\beta_{Y0}^{(1,1)}$  & 0.001  & 0.333 & 0.456 & 0.934  & -0.003 & 0.168 & 0.224 & 0.954 \\
     $\beta_{Y1}^{(1,1)}$  & -0.002 & 0.122 & 0.168 & 0.932  & -0.002 & 0.061 & 0.083 & 0.949 \\
     $\beta_{p1}^{(1,1)}$  & 0.001  & 0.957 & 1.309 & 0.944  & 0.024  & 0.481 & 0.638 & 0.956 \\
    \bottomrule
\end{tabular}
\par \smallskip
\parbox{14cm}{\footnotesize Note: We consider the bandwidth and Parzen kernel HAC considered in \cite{Kojevnikov3:21} with constant term and $\epsilon$ in $b_n$ being 0.4 and 0.05.}
\end{table}

\begin{table}
\small
\caption{\centering 1,000 Monte Carlo Simulations for the marginal exposure response ($\overline{\psi}$) on Ring Network ($c=0.4$)}
\begin{tabular}{cccccccccc}
    \toprule
    && \multicolumn{4}{c}{$n=250$} & \multicolumn{4}{c}{$n=1000$}  \\
    \cmidrule[0.05em](l){3-6}  \cmidrule[0.05em](l){7-10}
    & $p$ & Bias & SE & RMSE & 95\% Coverage & Bias & SE & RMSE & 95\% Coverage\\
    \midrule
    \multirow{ 3}{*}{$MER^{(0,0)}$}
    &.2& -0.004 & 0.549 & 0.765 & 0.912 & 0.005 & 0.276 & 0.376 & 0.950 \\
    &.5& -0.003 & 0.236 & 0.330 & 0.936 & 0.003 & 0.119 & 0.161 & 0.945 \\
    &.8& -0.001 & 0.298 & 0.419 & 0.911 & 0.002 & 0.152 & 0.209 & 0.938 \\
    \multirow{ 3}{*}{$MER^{(1,0)}$}
    &.2& 0.006  & 0.281 & 0.392 & 0.921 & 0.001 & 0.145 & 0.198 & 0.940 \\
    &.5& -0.007 & 0.366 & 0.511 & 0.915 & -0.001 & 0.185 & 0.253 & 0.955 \\
    &.8& -0.020 & 0.716 & 0.996 & 0.920 & -0.003 & 0.360 & 0.494 & 0.946 \\
    \multirow{ 3}{*}{$MER^{(0,1)}$}
    &.2& -0.007 & 0.419 & 0.581 & 0.934 & 0.001 & 0.211 & 0.289 & 0.944 \\
    &.5& -0.003 & 0.193 & 0.263 & 0.943 & -0.002 & 0.097 & 0.134 & 0.943 \\
    &.8& 0.001  & 0.210 & 0.292 & 0.931 & -0.004 & 0.107 & 0.147 & 0.937 \\
    \multirow{ 3}{*}{$MER^{(1,1)}$}
    &.2& -0.001 & 0.212 & 0.289 & 0.939 & -0.001 & 0.106 & 0.144 & 0.941 \\
    &.5& 0.000  & 0.242 & 0.333 & 0.940 & 0.006 & 0.122 & 0.164 & 0.956 \\
    &.8& 0.000  & 0.487 & 0.669 & 0.945 & 0.014 & 0.245 & 0.326 & 0.953 \\
\bottomrule
\end{tabular}
\par\smallskip
\vspace{0.5pt}
\parbox{14cm}{\footnotesize Note: We consider the bandwidth and Parzen kernel HAC considered in \cite{Kojevnikov3:21} with constant term and $\epsilon$ being 0.4 and 0.05.}
\end{table}

\begin{table}
\centering
\small
\caption{1,000 Monte Carlo Simulations for $\bm{\beta}$ on Ring Network ($c=0.5$)}
\begin{tabular}{lcccccccc}
    \toprule
    & \multicolumn{4}{c}{$n=250$} & \multicolumn{4}{c}{$n=1000$}  \\
    \cmidrule(lr){2-5}  \cmidrule(lr){6-9}
    & Bias & SE & RMSE & 95\% Coverage & Bias & SE & RMSE & 95\% Coverage\\
    \midrule
     $\beta_{D0}$         & -0.026 & 0.342 & 0.468 & 0.922  & -0.002 & 0.169 & 0.226 & 0.960 \\
     $\beta_{D1}$         & 0.044  & 0.270 & 0.372 & 0.936  & 0.012  & 0.133 & 0.182 & 0.948 \\
     $\lambda$            & 0.035  & 0.735 & 0.999 & 0.933  & -0.011 & 0.364 & 0.484 & 0.955 \\
     $\beta_{Y0}^{(0,0)}$  & 0.009  & 0.769 & 1.080 & 0.911  & 0.005  & 0.391 & 0.534 & 0.943 \\
     $\beta_{Y1}^{(0,0)}$  & 0.001  & 0.133 & 0.187 & 0.909  & 0.000  & 0.067 & 0.091 & 0.942 \\
     $\beta_{p1}^{(0,0)}$  & -0.008 & 1.233 & 1.731 & 0.918  & -0.009 & 0.628 & 0.856 & 0.943 \\
     $\beta_{Y0}^{(1,0)}$  & 0.009  & 0.432 & 0.602 & 0.915  & -0.005 & 0.223 & 0.308 & 0.931 \\
     $\beta_{Y1}^{(1,0)}$  & 0.000  & 0.175 & 0.253 & 0.895  & 0.002  & 0.090 & 0.123 & 0.929 \\
     $\beta_{p1}^{(1,0)}$  & -0.021 & 1.325 & 1.854 & 0.918  & 0.016  & 0.673 & 0.930 & 0.920 \\
     $\beta_{Y0}^{(0,1)}$  & 0.031  & 0.577 & 0.796 & 0.937  & -0.009 & 0.293 & 0.395 & 0.952 \\
     $\beta_{Y1}^{(0,1)}$  & -0.001 & 0.105 & 0.143 & 0.943  & 0.000  & 0.053 & 0.074 & 0.928 \\
     $\beta_{p1}^{(0,1)}$  & -0.042 & 0.894 & 1.239 & 0.934  & 0.017  & 0.455 & 0.615 & 0.951 \\
     $\beta_{Y0}^{(1,1)}$  & 0.001  & 0.334 & 0.463 & 0.934  & 0.000  & 0.169 & 0.229 & 0.927 \\
     $\beta_{Y1}^{(1,1)}$  & -0.004 & 0.120 & 0.169 & 0.921  & 0.002  & 0.062 & 0.084 & 0.949 \\
     $\beta_{p1}^{(1,1)}$  & -0.012 & 0.953 & 1.325 & 0.939  & 0.002  & 0.482 & 0.660 & 0.936 \\
    \bottomrule
\end{tabular}
\par \smallskip
\parbox{14cm}{\footnotesize Note: We consider the bandwidth and Parzen kernel HAC considered in \cite{Kojevnikov3:21} with constant term and $\epsilon$ in $b_n$ being 0.5 and 0.05.}
\end{table}

\begin{table}
\small
\caption{\centering 1,000 Monte Carlo Simulations for the marginal exposure response ($\overline{\psi}$) on Ring Network ($c=0.5$)}
\begin{tabular}{cccccccccc}
    \toprule
    && \multicolumn{4}{c}{$n=250$} & \multicolumn{4}{c}{$n=1000$}  \\
    \cmidrule[0.05em](l){3-6}  \cmidrule[0.05em](l){7-10}
    & $p$ & Bias & SE & RMSE & 95\% Coverage & Bias & SE & RMSE & 95\% Coverage\\
    \midrule
    \multirow{ 3}{*}{$MER^{(0,0)}$}
    &.2& 0.008  & 0.543 & 0.764 & 0.908 & 0.003 & 0.276 & 0.376 & 0.943 \\
    &.5& 0.006  & 0.233 & 0.327 & 0.924 & 0.000 & 0.119 & 0.161 & 0.949 \\
    &.8& 0.004  & 0.299 & 0.419 & 0.901 & -0.003 & 0.152 & 0.205 & 0.950 \\
    \multirow{ 3}{*}{$MER^{(1,0)}$}
    &.2& 0.005  & 0.283 & 0.402 & 0.920 & -0.001 & 0.145 & 0.201 & 0.947 \\
    &.5& -0.002 & 0.366 & 0.516 & 0.912 & 0.004  & 0.183 & 0.255 & 0.937 \\
    &.8& -0.008 & 0.711 & 0.997 & 0.915 & 0.009  & 0.357 & 0.496 & 0.933 \\
    \multirow{ 3}{*}{$MER^{(0,1)}$}
    &.2& 0.022  & 0.417 & 0.574 & 0.937 & -0.006 & 0.211 & 0.285 & 0.951 \\
    &.5& 0.009  & 0.192 & 0.263 & 0.945 & -0.001 & 0.097 & 0.131 & 0.944 \\
    &.8& -0.004 & 0.211 & 0.291 & 0.945 & 0.004  & 0.107 & 0.146 & 0.938 \\
    \multirow{ 3}{*}{$MER^{(1,1)}$}
    &.2& -0.005 & 0.211 & 0.298 & 0.917 & 0.002  & 0.106 & 0.144 & 0.957 \\
    &.5& -0.009 & 0.239 & 0.333 & 0.927 & 0.003  & 0.122 & 0.167 & 0.941 \\
    &.8& -0.012 & 0.483 & 0.670 & 0.930 & 0.003  & 0.246 & 0.338 & 0.944 \\
\bottomrule
\end{tabular}
\par\smallskip
\vspace{0.5pt}
\parbox{14cm}{\footnotesize Note: We consider the bandwidth and Parzen kernel HAC considered in \cite{Kojevnikov3:21} with constant term and $\epsilon$ being 0.5 and 0.05.}
\end{table}

\begin{table}
\centering
\small
\caption{1,000 Monte Carlo Simulations for $\bm{\beta}$ on Ring Network ($c=0.6$)}
\begin{tabular}{lcccccccc}
    \toprule
    & \multicolumn{4}{c}{$n=250$} & \multicolumn{4}{c}{$n=1000$}  \\
    \cmidrule(lr){2-5}  \cmidrule(lr){6-9}
    & Bias & SE & RMSE & 95\% Coverage & Bias & SE & RMSE & 95\% Coverage\\
    \midrule
     $\beta_{D0}$         & -0.006 & 0.340 & 0.468 & 0.934  & -0.004 & 0.169 & 0.230 & 0.944 \\
     $\beta_{D1}$         & 0.044  & 0.269 & 0.367 & 0.956  & 0.014  & 0.133 & 0.179 & 0.955 \\
     $\lambda$            & -0.010 & 0.728 & 1.014 & 0.924  & -0.004 & 0.365 & 0.497 & 0.936 \\
     $\beta_{Y0}^{(0,0)}$  & -0.028 & 0.766 & 1.088 & 0.908  & -0.009 & 0.388 & 0.530 & 0.942 \\
     $\beta_{Y1}^{(0,0)}$  & 0.002  & 0.133 & 0.186 & 0.906  & 0.001  & 0.067 & 0.092 & 0.934 \\
     $\beta_{p1}^{(0,0)}$  & 0.046  & 1.229 & 1.747 & 0.899  & 0.008  & 0.624 & 0.851 & 0.942 \\
     $\beta_{Y0}^{(1,0)}$  & 0.014  & 0.437 & 0.620 & 0.917  & -0.002 & 0.223 & 0.304 & 0.937 \\
     $\beta_{Y1}^{(1,0)}$  & -0.005 & 0.175 & 0.249 & 0.914  & -0.004 & 0.089 & 0.122 & 0.946 \\
     $\beta_{p1}^{(1,0)}$  & -0.041 & 1.323 & 1.867 & 0.913  & 0.011  & 0.677 & 0.908 & 0.950 \\
     $\beta_{Y0}^{(0,1)}$  & 0.027  & 0.577 & 0.787 & 0.937  & -0.002 & 0.293 & 0.401 & 0.940 \\
     $\beta_{Y1}^{(0,1)}$  & 0.001  & 0.106 & 0.147 & 0.934  & 0.002  & 0.053 & 0.072 & 0.942 \\
     $\beta_{p1}^{(0,1)}$  & -0.051 & 0.893 & 1.223 & 0.938  & 0.003  & 0.454 & 0.625 & 0.937 \\
     $\beta_{Y0}^{(1,1)}$  & -0.002 & 0.330 & 0.461 & 0.931  & -0.002 & 0.168 & 0.231 & 0.951 \\
     $\beta_{Y1}^{(1,1)}$  & 0.001  & 0.121 & 0.167 & 0.924  & 0.002  & 0.061 & 0.084 & 0.943 \\
     $\beta_{p1}^{(1,1)}$  & 0.001  & 0.942 & 1.318 & 0.932  & 0.007  & 0.480 & 0.664 & 0.943 \\
    \bottomrule
\end{tabular}
\par \smallskip
\parbox{14cm}{\footnotesize Note: We consider the bandwidth and Parzen kernel HAC considered in \cite{Kojevnikov3:21} with constant term and $\epsilon$ in $b_n$ being 0.6 and 0.05.}
\end{table}

\begin{table}
\small
\caption{\centering 1,000 Monte Carlo Simulations for the marginal exposure response ($\overline{\psi}$) on Ring Network ($c=0.6$)}
\begin{tabular}{cccccccccc}
    \toprule
    && \multicolumn{4}{c}{$n=250$} & \multicolumn{4}{c}{$n=1000$}  \\
    \cmidrule[0.05em](l){3-6}  \cmidrule[0.05em](l){7-10}
    & $p$ & Bias & SE & RMSE & 95\% Coverage & Bias & SE & RMSE & 95\% Coverage\\
    \midrule
    \multirow{ 3}{*}{$MER^{(0,0)}$}
    &.2& -0.017 & 0.539 & 0.766 & 0.903 & -0.006 & 0.275 & 0.376 & 0.943 \\
    &.5& -0.003 & 0.231 & 0.328 & 0.922 & -0.004 & 0.119 & 0.163 & 0.946 \\
    &.8& 0.011  & 0.300 & 0.422 & 0.905 & -0.002 & 0.150 & 0.206 & 0.941 \\
    \multirow{ 3}{*}{$MER^{(1,0)}$}
    &.2& 0.000  & 0.286 & 0.400 & 0.916 & -0.004 & 0.145 & 0.201 & 0.935 \\
    &.5& -0.012 & 0.358 & 0.500 & 0.926 & -0.001 & 0.185 & 0.251 & 0.960 \\
    &.8& -0.024 & 0.700 & 0.982 & 0.931 & 0.003  & 0.361 & 0.483 & 0.950 \\
    \multirow{ 3}{*}{$MER^{(0,1)}$}
    &.2& 0.018  & 0.415 & 0.565 & 0.945 & 0.001  & 0.211 & 0.289 & 0.945 \\
    &.5& 0.003  & 0.191 & 0.259 & 0.944 & 0.002  & 0.097 & 0.132 & 0.949 \\
    &.8& -0.012 & 0.212 & 0.292 & 0.930 & 0.002  & 0.107 & 0.147 & 0.936 \\
    \multirow{ 3}{*}{$MER^{(1,1)}$}
    &.2& -0.001 & 0.209 & 0.288 & 0.934 & 0.002  & 0.106 & 0.145 & 0.947 \\
    &.5& 0.000  & 0.239 & 0.331 & 0.932 & 0.004  & 0.121 & 0.168 & 0.931 \\
    &.8& 0.000  & 0.480 & 0.670 & 0.934 & 0.006  & 0.244 & 0.339 & 0.942 \\
\bottomrule
\end{tabular}
\par\smallskip
\vspace{0.5pt}
\parbox{14cm}{\footnotesize Note: We consider the bandwidth and Parzen kernel HAC considered in \cite{Kojevnikov3:21} with constant term and $\epsilon$ being 0.6 and 0.05.}
\end{table}

\begin{table}
\centering
\small
\caption{1,000 Monte Carlo Simulations for $\bm{\beta}$ on Ring Network ($c=0.7$)}
\begin{tabular}{lcccccccc}
    \toprule
    & \multicolumn{4}{c}{$n=250$} & \multicolumn{4}{c}{$n=1000$}  \\
    \cmidrule(lr){2-5}  \cmidrule(lr){6-9}
    & Bias & SE & RMSE & 95\% Coverage & Bias & SE & RMSE & 95\% Coverage\\
    \midrule
     $\beta_{D0}$         & -0.008 & 0.339 & 0.466 & 0.928  & 0.001  & 0.169 & 0.230 & 0.940 \\
     $\beta_{D1}$         & 0.043  & 0.271 & 0.376 & 0.940  & 0.009  & 0.133 & 0.179 & 0.961 \\
     $\lambda$            & -0.010 & 0.728 & 0.999 & 0.922  & -0.002 & 0.363 & 0.497 & 0.945 \\
     $\beta_{Y0}^{(0,0)}$  & -0.022 & 0.769 & 1.089 & 0.899  & 0.003  & 0.388 & 0.542 & 0.928 \\
     $\beta_{Y1}^{(0,0)}$  & 0.006  & 0.132 & 0.188 & 0.899  & 0.000  & 0.067 & 0.091 & 0.934 \\
     $\beta_{p1}^{(0,0)}$  & 0.052  & 1.230 & 1.758 & 0.898  & -0.005 & 0.624 & 0.866 & 0.930 \\
     $\beta_{Y0}^{(1,0)}$  & -0.036 & 0.436 & 0.620 & 0.897  & -0.015 & 0.225 & 0.308 & 0.942 \\
     $\beta_{Y1}^{(1,0)}$  & 0.005  & 0.174 & 0.250 & 0.898  & 0.001  & 0.090 & 0.119 & 0.971 \\
     $\beta_{p1}^{(1,0)}$  & 0.101  & 1.326 & 1.879 & 0.904  & 0.049  & 0.678 & 0.921 & 0.940 \\
     $\beta_{Y0}^{(0,1)}$  & 0.020  & 0.583 & 0.811 & 0.917  & 0.005  & 0.291 & 0.396 & 0.944 \\
     $\beta_{Y1}^{(0,1)}$  & 0.000  & 0.105 & 0.145 & 0.929  & -0.002 & 0.053 & 0.071 & 0.954 \\
     $\beta_{p1}^{(0,1)}$  & -0.036 & 0.903 & 1.261 & 0.919  & -0.008 & 0.452 & 0.616 & 0.946 \\
     $\beta_{Y0}^{(1,1)}$  & 0.001  & 0.332 & 0.465 & 0.911  & 0.006  & 0.168 & 0.226 & 0.957 \\
     $\beta_{Y1}^{(1,1)}$  & -0.001 & 0.121 & 0.168 & 0.931  & 0.002  & 0.062 & 0.084 & 0.931 \\
     $\beta_{p1}^{(1,1)}$  & 0.006  & 0.949 & 1.329 & 0.915  & -0.028 & 0.481 & 0.654 & 0.954 \\
    \bottomrule
\end{tabular}
\par \smallskip
\parbox{14cm}{\footnotesize Note: We consider the bandwidth and Parzen kernel HAC considered in \cite{Kojevnikov3:21} with constant term and $\epsilon$ in $b_n$ being 0.7 and 0.05.}
\end{table}

\begin{table}
\small
\caption{\centering 1,000 Monte Carlo Simulations for the marginal exposure response ($\overline{\psi}$) on Ring Network ($c=0.7$)}
\begin{tabular}{cccccccccc}
    \toprule
    && \multicolumn{4}{c}{$n=250$} & \multicolumn{4}{c}{$n=1000$}  \\
    \cmidrule[0.05em](l){3-6}  \cmidrule[0.05em](l){7-10}
    & $p$ & Bias & SE & RMSE & 95\% Coverage & Bias & SE & RMSE & 95\% Coverage\\
    \midrule
    \multirow{ 3}{*}{$MER^{(0,0)}$}
    &.2& -0.005 & 0.543 & 0.765 & 0.908 & 0.003  & 0.274 & 0.382 & 0.926 \\
    &.5& 0.010  & 0.234 & 0.327 & 0.923 & 0.001  & 0.118 & 0.164 & 0.941 \\
    &.8& 0.026  & 0.297 & 0.429 & 0.890 & 0.000  & 0.151 & 0.208 & 0.942 \\
    \multirow{ 3}{*}{$MER^{(1,0)}$}
    &.2& -0.011 & 0.284 & 0.406 & 0.902 & -0.004 & 0.145 & 0.194 & 0.949 \\
    &.5& 0.019  & 0.362 & 0.513 & 0.902 & 0.011  & 0.185 & 0.250 & 0.953 \\
    &.8& 0.049  & 0.706 & 0.999 & 0.916 & 0.025  & 0.361 & 0.489 & 0.939 \\
    \multirow{ 3}{*}{$MER^{(0,1)}$}
    &.2& 0.013  & 0.419 & 0.586 & 0.926 & 0.002  & 0.210 & 0.286 & 0.947 \\
    &.5& 0.002  & 0.192 & 0.265 & 0.935 & 0.000  & 0.097 & 0.131 & 0.945 \\
    &.8& -0.009 & 0.211 & 0.290 & 0.924 & -0.003 & 0.106 & 0.145 & 0.943 \\
    \multirow{ 3}{*}{$MER^{(1,1)}$}
    &.2& 0.001  & 0.209 & 0.297 & 0.914 & 0.003  & 0.106 & 0.144 & 0.962 \\
    &.5& 0.003  & 0.241 & 0.336 & 0.935 & -0.006 & 0.122 & 0.166 & 0.949 \\
    &.8& 0.004  & 0.485 & 0.675 & 0.924 & -0.014 & 0.245 & 0.333 & 0.949 \\
\bottomrule
\end{tabular}
\par\smallskip
\vspace{0.5pt}
\parbox{14cm}{\footnotesize Note: We consider the bandwidth and Parzen kernel HAC considered in \cite{Kojevnikov3:21} with constant term and $\epsilon$ being 0.7 and 0.05.}
\end{table}

\end{appendices}

\newpage
\bibliography{bib}